\documentclass[a4paper,12pt]{preprint}
\usepackage[margin=1.3in]{geometry}  
\usepackage{times}
\usepackage{hyperref}
\usepackage{breakurl}
\usepackage{comment}
\usepackage{microtype}

\usepackage{graphicx}              
\usepackage{amsmath}
\usepackage{amscd}               
\usepackage{amsfonts} 
\usepackage{amssymb}             
\usepackage{amsthm}                
\usepackage{mathrsfs}
\usepackage{amsbsy}
\usepackage{mhequ}
\usepackage{tikz}

\numberwithin{equation}{section}

\newtheorem{thm}{Theorem}[section]
\newtheorem{defn}[thm]{Definition}
\newtheorem{lem}[thm]{Lemma}
\newtheorem{prop}[thm]{Proposition}

\newtheorem{assumption}[thm]{Assumption}
\DeclareMathOperator{\id}{id}
\def\Wick#1{\,\colon\!\! #1 \!\colon}

\makeatletter
\def\th@newremark{\th@remark\thm@headfont{\bfseries}}
\makeatletter

\def\eps{\epsilon}
\theoremstyle{newremark}
\newtheorem{rmk}[thm]{Remark}

\definecolor{darkgreen}{rgb}{0.1,0.7,0.1}
\definecolor{darkred}{rgb}{0.7,0.1,0.1}
\definecolor{darkblue}{rgb}{0,0,0.7}
\newcommand{\PP}{\mathbb{P}}     

\newcommand{\TT}{\mathbf{T}}

\newcommand{\ZZ}{\mathbb{Z}}      

\newcommand{\bB}{\mathcal{B}}
\newcommand{\cC}{\mathcal{C}}
\newcommand{\dD}{\mathcal{D}}
\newcommand{\eE}{\mathcal{E}}
\newcommand{\fF}{\mathcal{F}}
\newcommand{\gG}{\mathcal{G}}
\newcommand{\hH}{\mathcal{H}}
\newcommand{\iI}{\mathcal{I}}

\newcommand{\kK}{\mathcal{K}}
\newcommand{\lL}{\mathcal{L}}
\newcommand{\mM}{\mathcal{M}}
\newcommand{\nN}{\mathcal{N}}
\newcommand{\oO}{\mathcal{O}}
\newcommand{\pP}{\mathcal{P}}
\newcommand{\qQ}{\mathcal{Q}}
\newcommand{\rR}{\mathcal{R}}

\newcommand{\tT}{\mathcal{T}}
\newcommand{\uU}{\mathcal{U}}
\newcommand{\vV}{\mathcal{V}}
\newcommand{\wW}{\mathcal{W}}
\newcommand{\xX}{\mathcal{X}}

\def\CO{\mathcal{O}}

\newcommand{\fM}{\mathfrak{M}}
\newcommand{\fR}{\mathfrak{R}}
\newcommand{\fs}{\mathfrak{s}}
\newcommand{\fK}{\mathfrak{K}}

\newcommand{\E}{\mathbf{E}}
\newcommand{\R}{\mathbf{R}}
\newcommand{\T}{\mathbf{T}}

\newcommand{\1}{\mathbf{1}}
\newcommand{\x}{\mathbf{x}}
\newcommand{\sE}{\mathscr{E}}

\newcommand{\sJ}{\mathscr{J}}
\newcommand{\sL}{\mathscr{L}}
\newcommand{\sM}{\mathscr{M}}
\newcommand{\sR}{\mathscr{R}}

\def\Wck{{\scriptscriptstyle\mathrm{Wick}}}
\newcommand{\wL}{L^{\Wck}}
\newcommand{\wM}{M^{\Wck}}

\newcommand{\wDelta}{\Delta^{\Wck}}
\newcommand{\wPi}{\Pi^{\Wck}}
\newcommand{\wf}{f^{\Wck}}

\newcommand{\ha}{\widehat{a}}

\newcommand{\heE}{\widehat{\eE}}

\newcommand{\hPi}{\widehat{\Pi}}

\newcommand{\emb}{\hookrightarrow}

\colorlet{symbols}{blue!90!black}
\colorlet{testcolor}{green!60!black}

\def\${|\!|\!|}
\def\ex{{\text{ex}}}

\usetikzlibrary{shapes.misc}
\usetikzlibrary{shapes.symbols}
\usetikzlibrary{snakes}
\usetikzlibrary{decorations}
\usetikzlibrary{decorations.markings}

\def\drawx{\draw[-,solid] (-3pt,-3pt) -- (3pt,3pt);\draw[-,solid] (-3pt,3pt) -- (3pt,-3pt);}
\tikzset{
	root/.style={circle,fill=testcolor,inner sep=0pt, minimum size=2mm},
	dot/.style={circle,fill=black,inner sep=0pt, minimum size=1mm},
	sdot/.style={circle,fill=black,inner sep=0pt, minimum size=0.5mm},
	var/.style={circle,fill=black!10,draw=black,inner sep=0pt, minimum size=
	2mm},
    sou/.style={circle,fill=black,inner sep=0pt, minimum size=1.5mm},
    nou/.style={circle,draw=black,inner sep=0pt, minimum size=1.5mm},
	dotred/.style={circle,fill=black!50,inner sep=0pt, minimum size=2mm},
	generic/.style={semithick,shorten >=1pt,shorten <=1pt},
	ageneric/.style={semithick},
	dist/.style={ultra thick,draw=testcolor,shorten >=1pt,shorten <=1pt},
	testfcn/.style={ultra thick,testcolor,shorten >=1pt,shorten <=1pt,<-},
	testfcnx/.style={ultra thick,testcolor,shorten >=1pt,shorten <=1pt,<-,
		postaction={decorate,decoration={markings,mark=at position 0.6 with {\drawx}}}},
	kepsilon/.style={semithick,shorten >=1pt,shorten <=1pt,densely dashed,->},
	kprimex/.style={semithick,shorten >=1pt,shorten <=1pt,densely dashed,->,
		postaction={decorate,decoration={markings,mark=at position 0.4 with {\drawx}}}},
	kernel/.style={semithick,shorten >=1pt,shorten <=1pt,->},
	akernel/.style={semithick,->},
	multx/.style={shorten >=1pt,shorten <=1pt,
		postaction={decorate,decoration={markings,mark=at position 0.5 with {\drawx}}}},
	kernelx/.style={semithick,shorten >=1pt,shorten <=1pt,->,
		postaction={decorate,decoration={markings,mark=at position 0.4 with {\drawx}}}},
	kernel1/.style={->,semithick,shorten >=1pt,shorten <=1pt,postaction={decorate,decoration={markings,mark=at position 0.45 with {\draw[-] (0,-0.1) -- (0,0.1);}}}},
	kernel2/.style={->,semithick,shorten >=1pt,shorten <=1pt,postaction={decorate,decoration={markings,mark=at position 0.45 with {\draw[-] (0.05,-0.1) -- (0.05,0.1);\draw[-] (-0.05,-0.1) -- (-0.05,0.1);}}}},
	kernelBig/.style={semithick,shorten >=1pt,shorten <=1pt,decorate, decoration={zigzag,amplitude=1.5pt,segment length = 3pt,pre length=2pt,post length=2pt}},
	gepsilon/.style={dotted,semithick,shorten >=1pt,shorten <=1pt},
	renorm/.style={shape=circle,fill=white,inner sep=1pt},
	labl/.style={shape=rectangle,fill=white,inner sep=1pt},
	xi/.style={circle,fill=symbols!10,draw=symbols,inner sep=0pt,minimum size=1.2mm},
	xix/.style={crosscircle,fill=symbols!10,draw=symbols,inner sep=0pt,minimum size=1.2mm},
	xib/.style={circle,fill=symbols!10,draw=symbols,inner sep=0pt,minimum size=1.6mm},
	xibx/.style={crosscircle,fill=symbols!10,draw=symbols,inner sep=0pt,minimum size=1.6mm},
	not/.style={circle,fill=symbols,draw=symbols,inner sep=0pt,minimum size=0.5mm},
	>=stealth,
	}

\makeatletter
\def\DeclareSymbol#1#2#3{\expandafter\gdef\csname MH@symb@#1\endcsname{\tikz[baseline=#2,scale=0.15,draw=symbols]{#3}}\expandafter\gdef\csname MH@symb@#1s\endcsname{\scalebox{0.7}{\tikz[baseline=#2,scale=0.15,draw=symbols]{#3}}}}
\def\<#1>{\csname MH@symb@#1\endcsname}
\makeatother

\DeclareSymbol{Xi22}{0.5}{\draw (0,0) node[xi] {} -- (-1,1) node[not] {} -- (0,2) node[xi] {}; }

\def\scal#1{\langle#1\rangle}
\def\d{\partial}

\begin{document}

\title{Large scale behaviour of 3D continuous\\ phase coexistence models}
\author{Martin Hairer$^1$ and Weijun Xu$^2$}
\institute{University of Warwick, UK, \email{m.hairer@warwick.ac.uk}
\and University of Warwick, UK, \email{weijun.xu@warwick.ac.uk}}

\maketitle

\begin{abstract}
We study a class of three dimensional continuous phase coexistence models, and show that, under different symmetry assumptions on the potential, the large-scale behaviour of such models near a bifurcation point is described by the dynamical $\Phi^p_3$ models for $p \in \{2,3,4\}$.
This result is specific to space dimension $3$ and does not hold in dimension $2$.
\end{abstract}

\tableofcontents

\section{Introduction}

The aim of this article is to study the large scale behaviour of phase coexistence models of the type
\begin{equation} \label{eq:general_model}
\partial_{t} u = \Delta u - \epsilon V_{\theta}'(u) + \delta \widehat{\xi}\;,
\end{equation}
in three spatial dimensions, where $V_\theta$ denotes a potential depending on some parameter $\theta$
and $\eps$, $\delta$ are two small parameters. Throughout this article, $\widehat{\xi}$ is assumed to be a continuous space-time Gaussian random field modelling the local fluctuations, with covariance having compact support and integrating to $1$. The potential $(\theta, u)\mapsto V_{\theta}(u)$ is a sufficiently regular function (depending on the regime, we will actually assume that it is polynomial in $u$). Regarding the two 
parameters $\epsilon$ and $\delta$, we will consider two extremal regimes: either 
$\epsilon = o(1), \delta \approx 1$, which we call the weakly nonlinear regime, or $\delta = o(1), \epsilon \approx 1$, which we call the weak noise regime. However, our results would easily carry over to 
intermediate regimes as well. Also, the spatial domain of the process $u$ is a large three dimensional torus whose size depends on $\epsilon$ (see Remark~\ref{rem:rescaling} for more details).

For the sake of the present discussion, consider the weakly nonlinear regime, i.e.\ set $\delta = 1$ in
\eqref{eq:general_model}. It is then natural to consider scalings of the type
$u_\lambda(t,x) = \lambda^{-1/2}u(t \lambda^{-2}, x\lambda^{-1})$ which leave invariant the stochastic heat equation, so that $u_{\eps^\alpha}$ satisfies
\begin{equ}[e:rescaled]
\partial_{t} u_{\eps^\alpha} = \Delta u_{\eps^\alpha} - \eps^{1-5\alpha/2} V_{\theta}'(\eps^{\alpha/2} u_{\eps^\alpha}) + \xi_{\eps^\alpha}\;,
\end{equ}
where $\xi_{\bar \eps}$ denotes a suitable rescaling of $\widehat{\xi}$ which approximates space-time
white noise at scales larger than $\bar \eps$.

\begin{rmk}
Since the process $u$ in \eqref{eq:general_model} itself depends on $\epsilon$, one should really 
write $u_{\epsilon,\epsilon^{\alpha}}$ for the rescaled process in \eqref{e:rescaled} to avoid ambiguity. 
However, we still write the ambiguous one $u_{\epsilon^{\alpha}}$ here in order to keep the notations simple. 
\end{rmk}

The form \eqref{e:rescaled} suggests that if we start \eqref{eq:general_model}
with an initial condition located at a local minimum of $V$, then at scales of order 
$\eps^{-1/2}$ (i.e.\ setting $\alpha = {1\over 2}$ in \eqref{e:rescaled}) solutions should
be well approximated by solutions to an Ornstein-Uhlenbeck process of the type
\begin{equ}[e:OU]
\d_t v = \Delta v - c v + \xi\;,
\end{equ}
for some $c > 0$ and $\xi$ a space-time white noise. As we will see in Theorem~\ref{th:main_asymmetric} below,
this is in general false, unless $V$ is harmonic to start with. 
Instead, one should compute from $V_\theta$ an effective potential
$\scal{V_\theta}$ in the following way. Consider the space-time stationary 
solution $\Psi$ to the linearised equation 
\begin{equation} \label{eq:linearised}
\d_t \Psi = \Delta \Psi + \widehat \xi\;. 
\end{equation}
Since we are in dimension $3$, such a solution exists and is Gaussian with finite variance $C_0$. 
We then set 
\begin{align*}
\langle V_{\theta} \rangle (x) = \int_{\R} V_{\theta}(x+y) \mu (dy), 
\end{align*}
where $\mu = \nN(0, C_{0})$. In other words, $\scal{V_\theta}$ is the effective potential obtained
by averaging $V$ against the stationary measure of $\Psi$. 
We show in Theorem~\ref{th:main_asymmetric} that if we start with an initial condition located at a local minimum of
$\scal{V_\theta}$, then it is indeed the case that the behaviour at scales of order $\eps^{-1/2}$
is described by \eqref{e:OU}.

These considerations suggest that more interesting nonlinear scaling limits can arise in regimes
where $\theta \mapsto \scal{V_{\theta}}$ undergoes a bifurcation, and this is the main object of study
of this article. In particular, if $\scal{V_{\theta}}$ is symmetric and undergoes a pitchfork bifurcation
at some $\theta = \theta_0$, then
one would expect the large-scale behaviour to be described near $\theta_0$ by 
the dynamical $\Phi^4_3$ model built in \cite{Hai14a}
and further investigated in \cite{CC13,Antti}.
Similarly, near a saddle-node bifurcation, one would expect the large-scale behaviour
to be described by the dynamical $\Phi^3_3$ model built in \cite{EJS13}
using the techniques developed in \cite{DPD02,DPD03}.

Recall that, at least formally, the dynamical $\Phi^p_3 (a)$ model is given by
the family of equations
\begin{equation} \label{eq:Phip3}
\partial_{t} \Phi = \Delta \Phi - a \Phi^{p-1} + \lambda \Phi^{p-3} + \xi, 
\end{equation}
where $\xi$ is the space-time white noise, and the spatial variable belongs to the three-dimensional
torus $\TT^{3}$. In this article, we will only ever consider $p \in \{3,4\}$,
with $p=2$ corresponding to the Ornstein-Uhlenbeck process \eqref{e:OU} (but then there is no term involving $\lambda$). 
Also, the constant $a$ in front of $\Phi^{p-1}$ can be set to $1$ by a formal scaling
\begin{equ}[e:rescalecanonical]
\Phi (t,x) \mapsto \nu^{-\frac{1}{2}} \Phi(t/\nu^{2}, x/\nu) \quad \text{with} \quad \nu^{3-\frac{p}{2}} = a, 
\end{equ}
since the transformation $\xi(t,x) \mapsto \nu^{-\frac{1}{2}} \xi(t/\nu^{2}, x/\nu)$ leaves the white noise invariant. The equation \eqref{eq:Phip3} with $a=1$ is the standard dynamical $\Phi^p_3$ model. In this article, we will however keep $a$ in the equation since it is convenient for the scalings later. 

For $p \in \{3,4\}$, the interpretation of \eqref{eq:Phip3} is not clear
a priori since solutions are distribution-valued so that the term
$\Phi^{p-1}$ lacks a canonical interpretation. However, they can be
constructed as limits of solutions to
\begin{equation} \label{eq:Phiapprox}
\partial_{t} \Phi_{\eps}^\lambda = \Delta \Phi_{\eps}^\lambda - a (\Phi_{\eps}^\lambda)^{p-1} + (C_\eps + \lambda) (\Phi_{\eps}^\lambda)^{p-3} + \xi_\eps\;, 
\end{equation}
for a regularisation $\xi_\eps$ of space-time white noise
and a suitable diverging sequence of constants $C_\eps$.
In the case $p =3$, this turns the term $\Phi^2$ into the Wick product 
$\Wick{\Phi^2}$ with respect to the Gaussian structure 
induced by the stationary solution to the corresponding linearised 
equation (see \cite{EJS13} for more details). 
In the case $p = 4$, the situation is more delicate and additional logarithmic
divergences arise due to higher order effects, see \cite{GJ,Feldman,Hai14a}. 

At this stage, it is important to note that the notation
\eqref{eq:Phip3}, even when interpreted as limit of processes of the type \eqref{eq:Phiapprox}, 
is really an abuse of notation: since one could always change the value of $C_{\epsilon}$ in \eqref{eq:Phiapprox} by a finite quantity, it is not clear which process should be associated to any fixed value of $\lambda$, and it is only
the whole family of processes, 
indexed by that finite quantity, which has a canonical meaning. We call the resulting family of 
solutions the $\Phi^p_3(a)$ family. 
Henceforth, when we say that a sequence of processes $\Psi_\eps^\lambda$ ``converges to the $\Phi^p_3(a)$ family
indexed by $\lambda$'', we mean that there exists a choice of $C_\eps$ (independent of $\lambda$)
such that $\lim_{\eps \to 0}\Psi_\eps^\lambda = \lim_{\eps \to 0}\Phi_\eps^\lambda =: \Phi^\lambda$ in law, 
for every $\lambda$.
The precise notion of convergence appearing here slightly depends on $p$ since the $\Phi^3_3$ process
may explode in finite time, while the $\Phi^4_3$ process doesn't \cite{Konstantin,Phi4_global}.
This will be clarified in \eqref{e:funnydistance} below.
Let us point out that, without the presence of the diverging counter-term $C_{\epsilon}$, the 
sequence $\Phi_{\epsilon}$ for $p=4$ would converge to $0$ in a sufficiently weak topology depending on 
the dimension $d$ (see \cite{HRW12} for more details). 

Formally, the equilibrium measure of the dynamics \eqref{eq:Phip3} for $p=4$ is the measure on Schwartz 
distributions associated to Bosonic Euclidean quantum field theory. This can also be justified 
rigorously, see \cite{Konstantin}.
The construction of this measure was 
a major achievement of constructive field theory; see the articles \cite{EO, Feldman, FO, GJ, Glimm} and references therein. In two spatial dimensions, the equation \eqref{eq:Phip3} was treated in \cite{AR91, DPD03, Phi42global}. For $d \ge 4$, one does not expect to be able to obtain any non-trivial 
scaling limit, see \cite{Frohlich,Aizenman,Brydges}. 

Another reason why the dynamical $\Phi^4_3$ is interesting is that it is expected to describe the $3$D Ising model with Glauber dynamics and Kac interactions near critical temperature (as conjectured in \cite{GLP99}). In fact, the one dimensional version of this result was shown in \cite{BPRS93} at the critical temperature. The two dimensional case is more difficult, as the equation itself requires renormalisation. It was shown recently in \cite{Ising2d} that the $2$D Kac-Ising model does rescale to $\Phi^4_2$ near critical temperature, and the renormalisation constant has a nice interpretation as the shift of critical temperature from its mean field value. See also the article \cite{Simon} which however required a two-step
procedure to obtain $\Phi^4_2$ from an Ising model.

We now turn back to the rescaled process \eqref{e:rescaled}. As suggested by the form of renormalisation in \eqref{eq:Phiapprox}, it is reasonable to expect that the behaviour of $u_{\epsilon}$ at scale $\alpha = 1$ and $\theta$ at (or near) a pitchfork bifurcation
should be well approximated by the dynamical $\Phi^4_3$ model. However, it turns out that this is \textit{not} true in full generality. The main result of this article is that, although $u_{\epsilon}$ converges to $\Phi^4_3$ for all symmetric polynomial potentials, for generic non-symmetric potentials, after proper re-centering and rescaling, the large scale behaviour of the system will always be described either by $\Phi^3_3$ or by the O.U. process of the type \eqref{e:OU}. One way to understand this is that, as is well-known from dynamical systems, pitchfork bifurcations are structurally unstable: small generic perturbations tend to turn them into a saddle-node bifurcation taking place very close to a local minimum. One can then argue (this is quite clear in Wilson's renormalisation group picture which has recently been applied to the construction of the dynamical $\Phi^4_3$ model in \cite{Antti}) that the effective potential experienced by the process at large scale is not $\scal{V_\theta}$ but some small perturbation thereof, thus 
reconciling our results with intuition.

\subsection{Weakly nonlinear regime} \label{sec:weak_nonlinear}

We start with the weakly nonlinear regime given by
\begin{equation} \label{eq:micro_model}
\partial_{t} u = \Delta u - \epsilon V_{\theta}'(u) + \widehat{\xi}, 
\end{equation}
where we assume that $V_\theta$ is a \textit{polynomial} whose coefficients depend smoothly on $\theta$.
Defining $\scal{V_\theta}$ as above, we thus write
\begin{align*}
\langle V_{\theta}' \rangle (u) = \sum_{j=0}^{m} \ha_{j} (\theta) u^{j}\;,
\end{align*}
for some smooth functions $\ha_j$.
For notational simplicity, we let $\ha_{j}, \ha_{j}'$ and $\ha_{j}''$ denote the value and first two 
derivatives of $\ha_{j}(\theta)$ at $0$. We will always assume that $\scal{V_\theta}$ has a critical 
point at the origin (which could easily be enforced by just translating $u$),
so that $\ha_{0} = 0$.

\begin{rmk}\label{rem:rescaling}
From now on, we will always assume that \eqref{eq:micro_model} is considered on a periodic domain
of the relevant size. In particular, we define $u_{\eps^\alpha}$ directly as the solution
to \eqref{e:rescaled} on a domain of size $\CO(1)$ (the precise size is irrelevant, but it should
be bounded and no longer depend on $\eps$). Ideally, one would like to extend the convergence
results of this article to all of $\R^3$, which would be much more canonical, but this
requires some control at infinity which is lacking at present.
\end{rmk}

\begin{rmk} \label{rem:noise_domain}
In principle, the noise $\widehat{\xi}$ appearing in \eqref{eq:micro_model} also depends on $\eps$, 
since it is defined on a torus of size $\eps^{-\alpha}$ for some $\alpha > 0$ depending on the regime we
consider. However, since we assume that its correlation function is fixed (independent of $\eps$)
and has compact support, the noises on domains of different sizes agree in law when considered on
an identical patch, as long as a suitable fattening of that patch remains simply connected. 
\end{rmk}

In the simplest case when $\ha_{1} \neq 0$, it is not very difficult to show that at 
scale $\alpha = \frac{1}{2}$, $u_{\epsilon^{\alpha}}$ converges in probability to 
the O.U.\ process. 
Interesting phenomena occur when $(0,0)$ is a bifurcation point for $\scal{V_\theta}$, which gives the necessary bifurcation condition
\begin{equation} \label{eq:bifurcation}
\ha_{0} = \ha_{1} = 0\;. 
\end{equation}
The saddle-node bifurcation further requires that $\ha_{0}' \neq 0$ and $\ha_{2} \neq 0$, and in this case one should choose $\alpha = \frac{2}{3}$ so that as long as $\theta = \oO(\epsilon^{\frac{2}{3}})$, the macroscopic process $u_{\epsilon^{\alpha}}$ converges to $\Phi^3_3$ family. In fact, the terms in $V_{\theta}'(\epsilon^{\alpha/2} u_{\epsilon^\alpha})$ in \eqref{e:rescaled} are Hermite polynomials in $u_{\epsilon^\alpha}$ whose coefficients are precisely $\ha_{j}(\theta)$'s with corresponding powers of $\epsilon$. Thus, the Wick renormalisation is already taken account of, and this is the reason why the bifurcation assumption naturally appears for $\scal{V_\theta}$ but not $V_\theta$.

The most interesting case arises when $(0,0)$ is a pitchfork bifurcation point of $\scal{V_\theta}$ so 
that in addition to \eqref{eq:bifurcation}, one has
\begin{equation} \label{eq:pitchfork}
\ha_{0}' = 0, \quad \ha_{1}' < 0, \quad \ha_{2} = 0, \quad \ha_{3} > 0\;. 
\end{equation}
As mentioned above, from \eqref{eq:Phiapprox}, it is natural to expect that at scale $\alpha = 1$, and with a suitable choice of $\theta$, the processes $u_{\epsilon^{\alpha}}$ should converge to the solution of the $\Phi^4_3$ model. As already alluded to earlier, this turn out to be true if and only if the quantity
\begin{equation} \label{eq:a}
A = \int P(z) \phantom{1} \E \big(V_{0}'(\Psi(0)) V_{0}''(\Psi(z)) \big) dz 
\end{equation}
vanishes, where $P$ is the heat kernel, $z$ denotes the space time variable $(t,x)$, and the expectation is taken with respect to the stationary measure of $\Psi$ as defined in \eqref{eq:linearised}. 
For general $V_0$, this integral diverges since the heat kernel $P$ is not integrable at large scales.
It turns out however that this expression is finite provided that 
\begin{align*}
\ha_0 \ha_1 = \ha_2 = 0\;, 
\end{align*}
which is certainly the case when $\scal{V_\theta}$ has a pitchfork bifurcation at the origin. 
The quantity $A$ can 
be written in terms of the coefficients of $\scal{V}$ as
\begin{equation} \label{eq:a_expression}
A = \sum_{j=3}^{m-1} (j+1)! \cdot \ha_{j} \ha_{j+1} C_{j}, 
\end{equation}
where the $C_{j}$ (to be defined in Section~\ref{sec:convModels} below) 
are explicit constants depending only on the covariance of $\widehat{\xi}$. It is clear from this expression that $A$ 
vanishes if $V$ is symmetric. 

If $A \neq 0$, then in order to obtain a nontrivial limit, it is necessary to slightly 
shift the potential from the origin, so we set
\begin{equation} \label{eq:macro_intro_a}
u_{\epsilon^{\alpha}} (t,x) = \epsilon^{-\frac{\alpha}{2}} \big( u(t/\epsilon^{2\alpha}, x/\epsilon^{\alpha}) - h_\eps \big)\;, 
\end{equation}
for some small $h_\eps$. The process $u_{\epsilon^\alpha}$ above then satisfies the equation
\begin{equation} \label{eq:shifted_equation}
\partial_{t} u_{\epsilon^\alpha} = \Delta u_{\epsilon^\alpha} - \epsilon^{1 - 5 \alpha / 2} V_{\theta}'(\epsilon^{\alpha/2} u_{\epsilon^\alpha} + h_{\epsilon}) + \xi_{\epsilon^\alpha}. 
\end{equation}
From now on, in both weakly nonlinear and weak noise regimes, we will use $u_{\epsilon^\alpha}$ to denote the re-centred process, and the process in \eqref{e:rescaled} is a special case of \eqref{eq:shifted_equation} when $h = 0$. We also assume the rescaled initial conditions $u_{\epsilon^{\alpha}}(0, \cdot)$ converge to a function $u_{0}(\cdot)$ in some sense (essentially in some low regularity H\"older norm at large scales and some high regularity H\"older norm at small scales -- this will be made precise in Definition~\ref{defn:weighted_function_space} and Section~\ref{sec:limits} below, the same is true for the symmetric case $A=0$), and we identify the limit of the solution sequence $\{u_{\epsilon^{\alpha}}\}$ for appropriate choices of $\theta$ and $h_{\epsilon}$. 

If one then takes $\theta \sim \epsilon^{\beta}$ for some $\beta < \frac{2}{3}$, 
then there are three different choices
of $h_{\epsilon}$'s such that the shifted process $u_{\epsilon^{\alpha}}$ converges to O.U.\ for $\alpha = \frac{1+\beta}{2}$. As expected, two of the possible limiting O.U.\ processes are stable, and the third one is unstable\footnote{Usually, the 
O.U.\ process is defined as the solution of \eqref{e:OU} only for $c>0$. But for the sake of simplicity of the presentation here, we call solutions to \eqref{e:OU} an O.U.\ process for every $c \in \R$. We call it a stable O.U.\ if $c > 0$, and unstable if $c \leq 0$. }. 
If $\theta \sim \epsilon^{\beta}$ for some $\beta > \frac{2}{3}$ on the other hand, 
then there is a unique choice
of $h_\eps$ such that at scale $\alpha = \frac{5}{6}$, the process $u_{\epsilon^{\alpha}}$ converges to 
a stable O.U. process. 

At the critical case $\theta = c \epsilon^{\frac{2}{3}}$, there is a constant $c^{*}$ such that for $c > c^{*}$ and $c < c^{*}$, at scale $\alpha = \frac{5}{6}$, $u_{\epsilon^{\alpha}}$ either converges 
to three O.U.'s or just one O.U., respectively. At $c = c^{*}$, there are two possible choices of $h_\eps$. One of them again yields a stable O.U.\ process at scale $\frac{5}{6}$ in the limit, but the other one 
yields $\Phi^3_3$ at scale $\alpha = \frac{8}{9}$. Note that this scale is much larger 
than the scale ${2\over 3}$ at which one obtains $\Phi^3_3$ in the case of a simple
saddle-node bifurcation. We summarise them in the following theorem. The precise statements can be found in Theorems~\ref{th:main_symmetric} and~\ref{th:main_asymmetric}.

\begin{thm}
Let $\langle V_{\theta} \rangle$ have a pitchfork bifurcation at
the origin, and let $u_{\epsilon^{\alpha}}$ be the solution 
to \eqref{e:rescaled} on $[0,T] \times \TT^{3}$. 

If the quantity $A$ given by \eqref{eq:a} is $0$, then there exists $\mu < 0$ such that at the distance to criticality
\begin{align*}
\theta = \mu \epsilon |\log \epsilon| + \lambda \epsilon + \oO(\epsilon^2), 
\end{align*}
scale $\alpha = 1$ and $h = 0$, the process $u_{\epsilon}$ converges to the $\Phi^4_3 (\ha_{3})$ family indexed by $\lambda$, where $\ha_{3}$ is the coefficient of the cubic term in $\scal{V_{0}'}$, the derivative of the averaged potential at $\theta = 0$. 

If $A \neq 0$, then the large scale behaviour of $u_{\epsilon^\alpha}$ depends on the value
\begin{align*}
\theta = \rho \epsilon^\beta, \rho > 0. 
\end{align*}
In fact, there exists $\rho^{*} > 0$ such that if $\beta < \frac{2}{3}$, or if $\beta = \frac{2}{3}$ and $\rho > \rho^{*}$, then there exist three choices of $h_{\epsilon}$'s such that at scale $\alpha = \frac{1+\beta}{2}$, two of the resulting processes $u_{\epsilon^\alpha}$ converge to a stable O.U.\ process, and the other converges to an unstable one.  

If $\beta > \frac{2}{3}$, or if $\beta = \frac{2}{3}$ and $\rho < \rho^{*}$, then there exists a choice of $h_{\epsilon}$ such that at scale $\alpha = \frac{5}{6}$, the process $u_{\epsilon^{\alpha}}$ converges to a stable O.U.\ process. 

At the critical value $\beta = \frac{2}{3}$ and $\rho = \rho^{*}$, there exist two choices of $h_{\epsilon}$ such that one of the resulting processes converges to a stable O.U.\ process at scale $\alpha = \frac{5}{6}$, and the other converges to $\Phi^3_3$ at scale $\alpha = \frac{8}{9}$. 
\end{thm}

The intuitive explanation why this is so is that $\scal{V}$ is really only a $0$-th order approximation to the ``real'' effective potential felt by the system at large scales. Since pitchfork bifurcations are structurally unstable, one would indeed expect higher-order corrections to $\scal{V}$ to turn this into a saddle-node bifurcation for generic non-symmetric potentials. 

The following picture illustrates our results, with the light shaded curve representing the symmetric case and the black curve representing the generic case when $\scal{V}$ undergoes a pitchfork bifurcation. Here, the field $\Phi$ is represented on the horizontal axis and the bifurcation parameter $\theta$ on
the vertical axis (with positive direction pointing downwards).
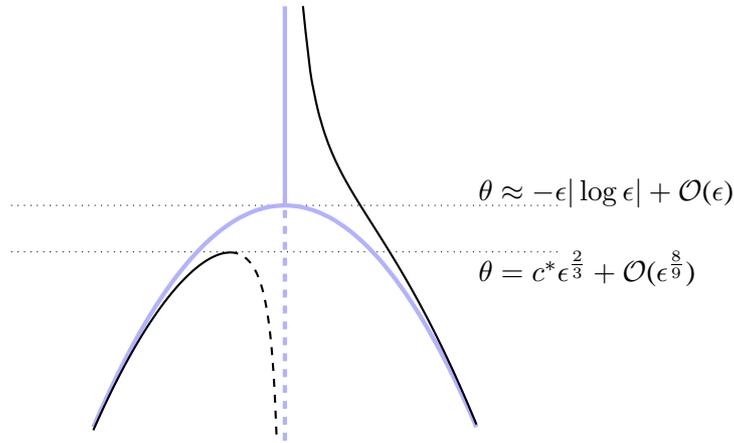
\begin{align*}
\begin{tikzpicture} [xscale=1.2,yscale=0.8,baseline=0cm]
\draw [dotted] (-3,0.3) -- (3,0.3); 
\draw [dotted] (-3,-0.47) -- (3,-0.47); 
\node[anchor=west] at (2,0.5) {\small $\theta \approx - \epsilon |\log \epsilon| + \oO(\epsilon)$}; 
\node[anchor=west] at (2,-0.75) {\small $\theta = c^{*} \epsilon^{\frac{2}{3}} + \oO(\epsilon^{\frac{8}{9}})$}; 
\draw[scale=0.6,domain=-3.5:3.5,smooth,ultra thick,blue!30,variable=\x] plot ({\x},{0.5-0.5*\x*\x}); 
\draw[scale=0.6,domain=-6:0.5,smooth,ultra thick,blue!30,dashed,variable=\x] plot ({0},{\x}); 
\draw[scale=0.6,domain=0.5:6,smooth,ultra thick,blue!30,variable=\x] plot ({0},{\x}); 
\draw[scale=0.6,domain=0.33:3.5,smooth,thick,variable=\x] plot ({\x},{-0.5*\x*\x + 2/(\x)}); 
\draw[scale=0.6,domain=-3.5:-1,smooth,thick,variable=\x] plot ({\x},{-0.5*\x*\x + 1/(\x) + 0.7}); 
\draw[scale=0.6,domain=-1:-0.15,smooth,thick,dashed,variable=\x] plot ({\x},{-0.5*\x*\x + 1/(\x) + 0.7}); 
\end{tikzpicture}
\end{align*}
The reason why, in the symmetric case, we see the bifurcation at $\theta \approx - \epsilon |\log \epsilon|$
rather than $\theta \approx \eps$
is due to the additional mass renormalisation appearing in $\Phi^4_3$. In the generic case where $\scal{V}$ is asymmetric (and the quantity $A$ defined in \eqref{eq:a} is non-zero), we can see that the asymmetry separates one local minimum from two other critical points, and creates 
a saddle-node bifurcation. It turns out that this bifurcation then occurs at $\theta = c^{*} \epsilon^{\frac{2}{3}} + \oO(\epsilon^{\frac{8}{9}})$ for an explicitly given constant $c^*$. All these results will be formulated precisely in Section~\ref{sec:limits} below. 

\begin{rmk}
The coefficient of the Wick term $\Wick{u^2}$ in the critical $\Phi^3_3$ case is proportional to $A^\frac{1}{3}$. If $A = 0$, then the process becomes a free field, and one can then further enlarge the scale to $1$, and adjust $\theta$ and $h$ to get $\Phi^4_3$. Also, the coefficient of the term $\Phi^{p-1}$ in the limiting equation depends on various coefficients of $\scal{V_{0}}$, but we can rescale them while leaving invariant the white noise such that they all become $1$. 
\end{rmk}

\begin{rmk}
In the non-symmetric case ($A \neq 0$), one can actually expand $\theta$ to the second order such that in the branch containing the saddle point, the scale increases continuously from $0$ up to $\frac{8}{9}$ with respect to $\theta$ (see Remark~\ref{rm:chart}). Similar results also hold in the symmetric case, but this is not important here, so we omit the details. 
\end{rmk}

%
%

\subsection{Weak noise regime} \label{sec:weak_noise}

There is another regime of microscopic models in which the nonlinear dynamics dominates the noise. The local mean field fluctuation is given by the equation
\begin{equ} \label{eq:weakNoise}
\partial_{t} u = \Delta u - V_{\theta}'(u) + \epsilon^{\frac{1}{2}} \widehat{\xi}, 
\end{equ}
where $V_\theta$ is a potential with sufficient regularity, not necessarily a polynomial. 
More precisely, we assume $V: \theta \mapsto V_{\theta}(\cdot)$ is a smooth function in 
the space of $\cC^{8}$ functions. Thus, we can Taylor expand $V_{\theta}'$ around $x=0$ as
\begin{equation} \label{eq:uniform_regular}
V_{\theta}'(x) = \sum_{j=0}^{6} a_{j}(\theta) x^{j} + F_{\theta}(x), 
\end{equation}
where $a_{j}$'s are smooth functions in $\theta$, and $|F_{\theta}(x)| \lesssim |x|^{7}$ uniformly over $|\theta| < 1$ and $|x| < 1$. 

Since the noise now has strength of order $\epsilon^{\frac{1}{2}}$, the large scale behaviour of 
\eqref{eq:weakNoise} is determined by the behaviour of $V_\theta$ itself near the origin, 
and not by that of an effective potential. Again, in order to observe an interesting limit, 
we assume that $V$ 
has a pitchfork bifurcation at $(0,0)$, namely one has
\begin{equation} \label{eq:pitchfork_noise}
a_{0} = a_{0}' = a_{1} = a_{2} = 0, \qquad a_{1}' < 0, \qquad a_{3} > 0, 
\end{equation}
where the $a_{j}(\theta)$ are the coefficients of the Taylor series of 
$V_{\theta}'(\phi)$ around $\phi = 0$. For $\epsilon > 0$, we set similarly to before
\begin{align*}
u_{\epsilon^{\alpha}} (t,x) = \epsilon^{-\frac{1+\alpha}{2}} \big( u (t/\epsilon^{2\alpha}, x/\epsilon^{\alpha}) - h_\eps  \big)\;, 
\end{align*}
where $h_\eps$ is a small parameter as before. We see that this time $u_{\eps^\alpha}$ solves the PDE
\begin{equation} \label{eq:noise_pde_intro}
\partial_{t} u_{\eps^\alpha} = \Delta u_{\eps^\alpha} - \eps^{-(\frac{1}{2} + \frac{5\alpha}{2})} V_{\theta}'(\eps^{\frac{1}{2} + \frac{\alpha}{2}} u_{\eps^\alpha} + h_{\epsilon}) + \xi_{\eps^\alpha}\;. 
\end{equation}
While this appears to be identical to \eqref{e:rescaled} modulo the substitution
$\alpha \mapsto \alpha + 1$, it genuinely differs from it in that the driving noise
still has correlation length $\eps^\alpha$ and not $\eps^{\alpha + 1}$.
In order for $u_{\eps^\alpha}$ to converge to $\Phi^4_3$, it then seems natural to choose 
$\alpha = 1$, thus guaranteeing that the coefficient of the cubic term in the Taylor 
expansion of $V_{\theta}'$ is of order $1$. 
But this creates the divergences in both linear and constant 
terms on the right hand side of the equation. Since $a_{0} = a_{1} = 0$, and we have 
two parameters $\theta$ and $h$ to tune, it looks like that we could kill the 
divergences by choosing the proper values of $\theta$ and $h$ and get $\Phi^4_3$ in the limit. 

Unfortunately, this turns out to be impossible. When tuning $\theta$ to its correct value to 
kill the linear divergence, the terms involving the leading order of $h$ will be precisely 
be canceled out so that $h$ could only have a second order effect, which is far from enough to kill the divergence in the constant term. Thus, one cannot make both 
linear and constant terms convergent unless the coefficients of $V$ itself are balanced. It 
turns out that similar to before, whether $u_{\eps}$ converges to $\Phi^4_3$ depends on the quantity
\begin{align*}
B = a_{4} + \frac{3 a_{0}'' a_{3}^{2}}{2 a_{1}'^{2}} - \frac{a_{2}' a_{3}}{a_{1}'}\;. 
\end{align*}
Indeed, what happens here is essentially the same as the weakly nonlinear regime except that the critical value $\theta$ at which one sees a bifurcation is different. Similar as above, we also require the convergence of the initial data $u_{\epsilon^\alpha}(0, \cdot)$. The main result can be loosely stated as follows, and the precise statements are in Theorems~\ref{th:noise_symmetric} and~\ref{th:noise}. 

\begin{thm} \label{th:noise_intro}
	Assume $V: \theta \mapsto V_{\theta}(\cdot)$ is a smooth function in the space of $\cC^{8}$ functions, and exhibits a pitchfork bifurcation at the origin $(\theta,x) = (0,0)$. Let $u_{\epsilon^{\alpha}}$ solves the PDE \eqref{eq:noise_pde_intro}. 
	
	If $B = 0$, then there exist choices of $\theta$ and $h$ of the form
	\begin{align*}
	\theta = a \epsilon + b \epsilon^{2} \log \epsilon + \oO(\epsilon^{2})\;, \qquad h = \rho_1 \epsilon + \rho_2 \epsilon^{2}
	\end{align*}
	such that $u_{\epsilon^{\alpha}}$ converges to $\Phi^4_3(a_{3})$ family at scale $\alpha = 1$\footnote{In order to get convergence to $\Phi^4_3$, one needs to choose $\rho_2$ depending on the coefficient of the $\epsilon^{2}$ term in $\theta$; otherwise one will get a shifted $\Phi^4_3$ of the form
		\begin{align*}
		\partial u = \Delta u - a_{3} (u^{3} - \infty u) + \xi + C
		\end{align*}
		with an additional constant $C$. This constant can be killed by a proper choice of $\rho_{2}$.}.
	
	If $B \neq 0$, then there exist $\rho_{j}^{*} > 0$ for $j = 1, 2, 3$ such that if
	\begin{align*}
	\theta = \theta^{*} = \rho_{1}^{*} \epsilon + \rho_{2}^{*} \epsilon^{\frac{4}{3}} + \rho_{3}^{*} \epsilon^{\frac{5}{3}} + \oO(\epsilon^{\frac{16}{9}}), 
	\end{align*}
	then there exist two choices of $h_{\epsilon}$ such that one of the resulting processes $u_{\epsilon^{\alpha}}$ converges to $\Phi^3_3$ at scale $\alpha = \frac{7}{9}$, and the other 
	one converges to a stable OU process at $\alpha = \frac{2}{3}$. 
	
	If $\theta > \theta^{*}$ (resp.\ $\theta < \theta^{*}$), then there exist three (resp.\ one) choices of $h_{\epsilon}$ such that the resulting $u_{\epsilon^{\alpha}}$ converge to OU processes. In the former case, two of the OU processes are stable and the last one is unstable; in the latter case the OU 
	process is stable. 
\end{thm}

\begin{rmk}
Similar to the weakly nonlinear case, the coefficient of the Wick term for $\Phi^3_3$ is proportional to $B^{\frac{1}{3}}$. A symmetric potential $V$ will give $B = 0$, but it is not clear whether the quantity $B$ has a probabilistic meaning as in the case of $A$ \eqref{eq:a}. Also, as explained just before \eqref{e:rescalecanonical}, one could rescale the solution leaving invariant the white noise such that all the limits are of the form \eqref{eq:Phip3} with $a=1$. 
\end{rmk}

The precise statement will be given in Theorems~\ref{th:noise_symmetric} and~\ref{th:noise}.

\subsection{Some remarks and structure of the article}

Before describing the structure of this article,
we discuss two possible natural generalisations of our results. 

\begin{enumerate}
\item We expect that analogous results still hold when the noise $\widehat{\xi}$ 
is not assumed to be Gaussian, but still satisfies good enough integrability and mixing 
conditions. The techniques developed in \cite{KPZCLT} should apply here as well. Indeed, in \cite{phi4_non_gaussian}, the authors showed convergence of the weakly nonlinear regime to $\Phi^4_3$ under symmetry assumption on both the potential and the noise. Note however that if the noise is non-symmetric, then we do not expect to see 
$\Phi^4_3$ at large scales generically, even if $V_\theta$ is symmetric. 

\item The assumption that $V_{\theta}$ is a polynomial can probably be relaxed (see \cite{KPZ_equilibrium}
for a result similar to those of \cite{HQ15} in the context of the KPZ equation). 
It is however not clear at all at this stage how the methods in this article could be carried 
over to handle this case. 
\end{enumerate}

It turns out that, as in \cite{HQ15}, the weak noise regime can be treated as a 
perturbation of the weakly nonlinear regime, so we will mainly focus on the latter case. 
The main strategy to prove the above results is the recently developed theory of regularity 
structures (\cite{Hai14a}), combined with the results of (\cite{HQ15}), where results
analogous to ours are obtained for the KPZ equation. 
The idea is to lift and solve \eqref{e:rescaled} in an abstract regularity structure space that is purposed built for this equation, and then pull the solution back to the usual distribution spaces after suitable renormalisation. 

The article is organised as follows. In Section~\ref{sec:structure}, we construct the regularity structure as well as the renormalisation maps that allow us to treat the equations of the form \eqref{e:rescaled}. 
Section~\ref{sec:sol} is devoted to construction of the solution to the abstract equation. In Section~\ref{sec:convModels}, we prove the convergence of the renormalised models. Finally, in Section~\ref{sec:limits}, we collect all the previous results to identify the limit of the renormalised solutions.

\subsection*{Acknowledgements}

{\small
MH gratefully acknowledges financial support from 
the Philip Leverhulme trust and the European Research Council.
}

\section{Construction of the regularity structure}
\label{sec:structure}

In this section, we build a regularity structure that is sufficiently rich to solve the fixed point problem for the equation
\begin{equation} \label{eq:mollified_equation}
\partial_{t} u_{\epsilon} = \Delta u_{\epsilon} - \epsilon^{-\frac{3}{2}} V_{\theta}'(\epsilon^{\frac{1}{2}} u_{\epsilon}) + \xi_{\epsilon}
\end{equation}
in the abstract space of modelled distributions. Here, $\xi_{\epsilon}$ is a mollified version of the space-time white noise $\xi$ at scale $\epsilon$, and $V_{\theta}'$ is a polynomial of degree $m$. Note that \eqref{eq:mollified_equation} corresponds to the weakly nonlinear regime with scale $\alpha = 1$, and we do not restrict $V$ to be symmetric here. Since this is the largest scale we will look at, all other situations (including the weak noise regime) will follow as a perturbation of the above equation. 

The construction of the regularity structure mainly follows the methodologies and set up in \cite{Hai14a} and \cite[Sec.~3]{HQ15}, with some slight modifications to accommodate the particular form of the equation \eqref{eq:mollified_equation}. We will refer to the precise statements in those two papers when we state a result from there without proof. More gentle introductions to regularity structures can be found in \cite{Hai14b}, \cite{Phi4review}, \cite{ICM} and \cite{HendrikNotes}.

\subsection{The (extended) regularity structure}

Recall that a regularity structure is a pair $(\tT, \gG)$, where 
$\tT = \bigoplus_{\alpha \in A} \tT_\alpha$ is a vector space
that is graded by some (bounded below, locally finite) 
set $A \subset \R$ of homogeneities, and $\gG$ is a group of linear transformations
of $\tT$ such that, for every $\Gamma \in \gG$, $\Gamma - \id$ is strictly upper triangular
with respect to the graded structure. 

For the purpose of this article, we build basis vectors $\tT$ similarly to \cite[Sec.~3.1]{HQ15}
as a collection of formal expressions built from the symbols $\1$, $\Xi$, $\{X_i\}_{i=0}^3$
and operators $\iI$ and $\eE^\beta$ for \textit{half integers} $\beta > 0$. As usual,
we assume that all symbols and sub-expressions commute and that $\1$ is neutral for the 
product, so we identify
for example $\iI(\Xi X_1) \Xi$ and $\Xi \1\iI(X_1\Xi)$.
Given a multi-index $k = (k_0,\cdots,k_3)$, we also write $X^k$ as a shorthand
for $X_0^{k_0}\cdots X_3^{k_3}$ (with the convention $X_i^0 = \1$), 
and $|k| = 2k_0 + \sum_{i=1}^3 k_i$ for its
parabolic degree. 

With these notations, we define two sets $\uU$ and $\vV$ of such expressions 
as the smallest sets such that $X^{k} \in \uU$, $\Xi \in \vV$, and such that for 
every $k \in \{1,\ldots,m-3\}$,
\begin{equation}\label{eq:graded_vector_space} 
	\begin{split}
\{\tau_{1}, \cdots, \tau_{k}\} \subset \uU  \qquad  &\Rightarrow \qquad \{\tau_1\tau_2\tau_3\;,\;\eE^{\frac{k}{2}} (\tau_{1} \cdots \tau_{k+3})\} \subset \vV\;, \\
\tau \in \vV \qquad &\Rightarrow \qquad \iI(\tau) \in \uU\;.
\end{split}
\end{equation}
We then set $\wW = \uU \cup \vV$ and we associate to each element of $\wW$ a homogeneity in the following
way. We set
\begin{equ}
|\Xi| = -\frac{5}{2} - \kappa, \qquad |X^{k}| = |k|\;, 
\end{equ}
where $\kappa$ is a small positive number to be fixed later, and we extend this to every formal 
expression in $\wW$ by
\begin{equation} \label{eq:homogeneity}
|\tau \bar{\tau}| = |\tau| + |\bar{\tau}|, \quad |\iI(\tau)| = |\tau| + 2, \quad |\eE^{\beta} (\tau)| = \beta + |\tau|\;.
\end{equation}
We then write $\tT_\alpha$ for the free vector space generated by $\{\tau \in \wW\,:\, |\tau| = \alpha\}$.
In this article, we will only ever use basis vectors with homogeneity less than $2$, 
we therefore take for $\tT$ the space of all finite linear
combinations of elements of $\wW$ of homogeneity less than $2$, i.e.\ $\tT = \bigoplus_{\alpha < 2}\tT_\alpha$.

The main reason for introducing $\eE^{\beta}$ as in \eqref{eq:graded_vector_space} rather than treating $\epsilon$ as a fixed real number is the following crucial fact. It reflects that \eqref{eq:mollified_equation} is subcritical under the scaling reflected by our regularity structure.

\begin{lem} \label{le:subcritical}
If $\kappa < \frac{1}{8m}$, then for every $\gamma > 0$, the set $\{ \tau \in \wW: |\tau| < \gamma \}$ is finite. 
\end{lem}

As in \cite{HQ15}, it will be convenient to consider $\eE^{\beta}$ as a linear map such that
$\eE^{\beta}:  \tau \mapsto \eE^{\beta} (\tau)$. 
The problem is that the product $\tau_{1} \cdots \tau_{\ell+3}$ appearing in
\eqref{eq:graded_vector_space} does in general not belong to $\tT$. 
Just as in \cite[Sec.~3.3]{HQ15}, one way to circumvent this problem is to introduce 
the extended regularity structure $\tT_{\ex}$, given by the linear span of 
\begin{align*}
\wW_{\ex} = \wW \cup \big\{\tau_{1} \cdots \tau_{m}: \tau_{j} \in \uU  \big\}. 
\end{align*}
In this way, we can view $\eE^{\beta}$ as a linear map defined on (a subspace of) $\tT_{\ex}$. 

We now start to describe the structure group $\gG$ for $\tT_\ex$. 
For this, we introduce the following three sets of formal symbols: 
\begin{equs}[e:defT+]
\fF_{1} &= \big\{ \1, X \big\}, \qquad \fF_{2} = \big\{ \sJ_{\ell}(\tau): \phantom{1} \tau \in \wW \setminus \{X^k\}, \phantom{1} |\tau| + 2 > \ell \big\}, \\
\fF_{3} &= \big\{ \sE^{\frac{k}{2}}_{\ell}(\tau_{1} \cdots \tau_{k+3}): \phantom{1} \tau_{j} \in \uU, \phantom{1} \frac{k}{2} + \sum_{j} |\tau_{j}| > |\ell| \geq \sum_{j} |\tau_{j}|\big\}\;.  
\end{equs}
We then let $\tT_{+}$ be the commutative algebra generated by the elements in $\fF_{1} \cup \fF_{2} \cup \fF_{3}$ and we define a linear map $\Delta: \tT \rightarrow \tT \otimes \tT_{+}$ 
in the same way as in \cite[Sec.~3.1]{HQ15}.

For any linear functional $g: \tT_{+} \rightarrow \R$, one obtains a linear 
map $\Gamma_{g}: \tT \rightarrow \tT$ by $\Gamma_{g} \tau = (\id \otimes g) \Delta \tau$. 
Denoting by $\gG_{+}$ the set of multiplicative linear functionals $g$ on $\tT_+$, we then set 
\begin{align*}
\gG_{+} = \big\{ g \in \tT_{+}^{*}: \phantom{1} g(\tau \bar{\tau}) = g(\tau) g(\bar{\tau}), \phantom{1} \forall \tau, \bar{\tau} \in \tT_{+} \big\}\;, 
\end{align*}
and we define $\gG$ by
\begin{equation} \label{eq:structure_group}
\gG = \big\{ \Gamma_{g}: g \in \gG_{+}  \big\}\;. 
\end{equation}
It is straightforward to verify that $\gG$ has the desired properties, including the fact
that its elements respect the product structure of $\tT$ in the sense that 
$\Gamma(\tau \bar \tau) = \Gamma \tau \cdot \Gamma \bar \tau$.
Furthermore, $\gG$ preserves not only $\tT_\ex$, but also $\tT$, so that it also serves as
the structure group for $\tT$.

\subsection{Admissible models} \label{sec:admissible_model}

We now start to introduce a class of admissible models 
for our regularity structure. 
As in \cite{Hai14a}, we fix a truncation $K$ of the heat kernel which coincides with it near the 
origin and annihilates polynomials of degree up to $3$. 
The existence of such a kernel $K$ is easy to show, and can be found, for example, in 
\cite[Sec.~5.1]{Hai14a}. 

We equip $\R^{1+3}$ the parabolic metric so that
\begin{equation} \label{eq:parabolic_metric}
|z| = |(t,x)| = |t|^{\frac{1}{2}} + \sum_{j=1}^{3} |x_{j}|.  
\end{equation}
We let $\dD'$ denote the space of Schwartz distributions on $\R^{1+3}$ and $\lL(\tT, \dD')$ the space of linear maps from $\tT$ to $\dD'$. Furthermore, for any test function $\varphi: \R^{1+3} \rightarrow \R$, any $z \in \R^{1+3}$ and $\lambda \in \R^{+}$, we use $\varphi_z^{\lambda}$ to denote
$\varphi_{z}^{\lambda}(z') = \lambda^{-5} \varphi \big( (t'-t){\lambda^{-2}}, (x'-x){\lambda^{-1}}  \big)$. 
We also write $\bB$ for the set of smooth functions $\varphi\colon \R^4 \to \R$ 
that are compactly supported in 
$\{ |z| \leq 1 \}$ whose derivatives up to order three (including the value of the function) 
are uniformly bounded by $1$. 

Recall that a model for $(\tT, \gG)$ consists of a pair $(\Pi, F)$ of functions
\begin{equs}[2]
\Pi: \R^{1+3} &\rightarrow \lL(\tT, \dD') \qquad&\qquad F: \R^{1+3} &\rightarrow \gG \\
z &\mapsto \Pi_{z} \qquad & \qquad z &\mapsto F_{z}
\end{equs}
satisfying the identity
\begin{equation} \label{eq:model_algebraic}
\Pi_{z} F_{z}^{-1} = \Pi_{\bar{z}} F_{\bar{z}}^{-1}, \qquad \forall z, \bar{z},  
\end{equation}
as well as the bounds
\begin{equation} \label{eq:model_analytic}
|(\Pi_{z} \tau)(\varphi_{z}^{\lambda})| \lesssim \lambda^{|\tau|}, \qquad | \Gamma_{z,\bar{z}} \tau |_{\sigma} \lesssim |z-\bar{z}|^{|\tau| - |\sigma|}
\end{equation}
uniformly over all $\varphi \in \bB$, all space-time points $z, \bar{z}$ in compact domains and every $\tau \in \wW$, where we used the shorthand
$\Gamma_{z,\bar{z}} = F_{z}^{-1} \circ F_{\bar{z}}$, and the proportionality constant depends on the compact domain $\fK$. We will write $f_{z}$ for the element in $\gG_{+}$ such that $F_{z} = \Gamma_{f_{z}}$. We will give explicit expressions for $f_{z}$, and will write the notation $(\Pi,f)$ for a model frequently. We also write $|\tau|_\sigma$ for the norm
of the component of $\tau$ in $\tT_\sigma$ (the precise choice of norm does not matter since these
spaces are all finite-dimensional). We define the norm of a model $\fM = (\Pi, f)$ to be the smallest constant that makes both bounds in \eqref{eq:model_analytic} to hold, and denote it by $\$ \fM \$_{\fK}$. Since in most of the situations, $F$ is completely determined by $\Pi$, we sometimes also write $\$ \Pi \$$ instead of $\$ \fM \$$, and we omit the domain $\fK$ wherever no confusion may arise. With these notations, we can define what we mean by an \textit{admissible} model.

\begin{defn}
A model $(\Pi, f)$ is admissible if for every multi-index $k$, one has
\begin{equation} \label{eq:model_polynomial}
(\Pi_{z} X^{k})(\bar{z}) = (\bar{z} - z)^{k}, \qquad f_{z}(X^{k}) = (-z)^{k}
\end{equation}
and for every $\tau \in \wW$ with $\iI(\tau) \in \tT$, one has
\begin{equation} \label{eq:model_integration}
\begin{split}
f_{z}(\sJ_{\ell} \tau) &= - \int D^{\ell}K(z - \bar{z}) (\Pi_{z} \tau) (d \bar{z}), \qquad |\ell| < |\tau| + 2 \\
\Pi_{z} \iI(\tau)(\bar{z}) &= (K * \Pi_{z} \tau)(\bar{z}) + \sum_\ell \frac{(\bar{z} - z)^{\ell}}{\ell !} \cdot f_{z}(\sJ_{\ell} \tau).  
\end{split}
\end{equation}
Here, we set $\sJ_{\ell} \tau = 0$ if $|\ell| \geq |\tau| + 2$, so the sum is always finite. 
\end{defn}

See \cite[Rem.~5.10]{Hai14a} for the correct way of interpreting these identities in case $\Pi_z$ 
contains distributions that are not functions.

\subsection{Canonical lift to $\tT_{\ex}$} \label{sec:canonical_lift}

Given any \textit{smooth} space-time function $\widehat{\xi}$ and any real number $\epsilon$, there is a canonical way to build an admissible model $\sL_{\epsilon}(\widehat{\xi}) = (\Pi^{\epsilon}, f^{\epsilon})$ for the regularity structure $(\tT_{\ex}, \gG)$ as follows. We first set
\begin{align*}
(\Pi^{\epsilon}_{z} \Xi)(\bar{z}) = \widehat{\xi}(\bar{z}), 
\end{align*}
independent of $\epsilon$ and the base point $z$. We then define $\Pi^{\epsilon}_{z} \tau$ recursively for other $\tau \in \wW$ by \eqref{eq:model_integration} as well as the identities
\begin{equation} \label{eq:canonical_multiplication}
(\Pi^{\epsilon}_{z} \tau \bar{\tau})(\bar{z}) = (\Pi^{\epsilon}_{z} \tau)(\bar{z}) \cdot (\Pi^{\epsilon}_{z} \bar{\tau})(\bar{z})
\end{equation}
and
\begin{equation} \label{eq:model_epsilon}
\begin{split}
f^{\epsilon}_{z}(\sE^{\beta}_{\ell}\tau) &= - \epsilon^{\beta} \big( D^{\ell} (\Pi^{\epsilon}_{z} \tau) \big)(z), \\
(\Pi_{z} \eE^{\beta} \tau) (\bar{z}) &= \epsilon^{\beta} (\Pi^{\epsilon}_{z} \tau)(\bar{z}) + \sum_{\ell} \frac{(\bar{z} - z)^{\ell}}{\ell !} \cdot f^{\epsilon}_{z} (\sE^{\beta}_{\ell} \tau). 
\end{split}
\end{equation}
Here, we again adopt the convention $\sE^{\beta}_{\ell} (\tau) = 0$ if $|\ell| \geq \beta  + |\tau|$. This 
construction makes sense only when $\Pi_{z} \tau$ is sufficiently regular, and this is indeed 
the case if $\widehat{\xi}$ is smooth. 
We then have the following fact, the proof of which can be found in \cite[Sec.~3.6]{HQ15}.

\begin{prop}
Let $\widehat{\xi}$ be a smooth space-time function, and $\epsilon \geq 0$. Then, the canonical model $\sL_{\epsilon}(\widehat{\xi}) = (\Pi^{\epsilon}, f^{\epsilon})$ defined via the identities \eqref{eq:model_polynomial} -- \eqref{eq:model_epsilon} is an admissible model. 
\end{prop}

Later on, we will consider the situation where $\widehat{\xi} = \xi_{\epsilon}$, a regularised version of the space-time white noise $\xi$, so we are led to the canonical model $\sL_{\epsilon}(\xi_{\epsilon})$. However, it is important to note that at this stage nothing forces the values of the two $\epsilon$'s 
to be identical: it is perfectly legitimate to consider the model $\sL_{\epsilon} (\xi_{\delta})$ for any pair of $(\epsilon, \delta)$.

Also, one would like the linear map $\eE^{\beta}$ to represent the multiplication by $\epsilon^{\beta}$. This is however not quite true in view of \eqref{eq:model_epsilon}, and it suggests that we should introduce a new map $\heE^{\beta}$ on the $\dD^{\gamma}$ space of modelled distributions (see Sec.$3$ in \cite{Hai14a} for a definition) by
\begin{equation} \label{eq:multiplication_epsilon}
(\heE^{\beta} U)(z) = \eE^{\beta} U(z) - \sum_{\ell} \frac{X^{\ell}}{\ell !} f_{z} \big( \sE^{\beta}_{\ell}(U(z)) \big). 
\end{equation}
Then, as long as the model is admissible and satisfies \eqref{eq:model_epsilon}, the map $\heE^{\beta}$ does indeed represent multiplication by $\epsilon^{\beta}$ in the sense
that $\rR \hat \eE^\beta U = \eps^\beta \rR U$ for $\rR$ the reconstruction operator.

\subsection{Renormalisation} \label{sec:renormalisation}

The aim of this section is to build a group $\fR$ of transformations that we 
can use to ``renormalise'' our models. It is crucial for our purpose that such a 
renormalisation procedure satisfies the following three properties: 
\begin{enumerate}
\item $\fR$ acts on the space $\sM$ of \textit{admissible} models.
\item $\fR$ is sufficiently rich so that one can find elements $M_\eps \in \fR$
such that $M_\eps \sL_\eps(\xi_\eps)$ converges to a limit in $\sM$, where
$\sL_\eps$ denotes the ``canonical lift'' of the regularised noise $\xi_\eps$.
\item Solving the fixed point problem \eqref{eq:abstract_eq} for a model
of the type $M \sL_\eps(\eta)$ for a smooth space-time function $\eta$ and $M \in \fR$
leads to the solution of a modified PDE. 
\end{enumerate}

The transformations $M \in \fR$ we consider here will be composed
by two linear maps $M_{0}$ and $\wM$ on $\tT_\ex$.
The map $\wM$ encodes ``Wick renormalisation'', while $M_{0}$ has 
the interpretation as mass 
renormalisation in the quantum field theory. 
From now on, we will use the shorthand $\Psi = \iI(\Xi)$. We start with the standard Wick renormalisation map $\wM$ on $\tT_{\ex}$. Define the generator $\wL$ by
\begin{align*}
\wL \Xi = \wL X^{k} = 0, \qquad \wL \Psi^{k} = \begin{pmatrix} k \\ 2 \end{pmatrix} \Psi^{k-2}, 
\end{align*}
and extend this to the whole of $\tT_{\ex}$ by
\begin{align*}
\wL (\tau \iI(\bar \tau)) = \wL(\tau) \iI(\bar \tau) + \tau \iI (\wL \bar \tau)\;, 
\end{align*}
for $\bar \tau \neq \Xi$, as well as
\begin{align*}
\wL \iI(\tau) = \iI (\wL \tau), \phantom{11} \wL (\eE^{\beta} \tau) = \eE^\beta(\wL \tau), \phantom{11} \wL(X^{k}\tau) = X^{k}\wL\tau\;. 
\end{align*}
The map $\wM: \tT_{\ex} \rightarrow \tT_{\ex}$ is then defined by
\begin{equation} \label{eq:Wick}
\wM = \exp (-C_{1} \wL)\;. 
\end{equation}
The definition of $\wL$ ensures that $\wM$ commutes with $X^{k}$ as well as with the abstract integration maps $\iI$ and $\eE^\beta$. $\wM$ has the interpretation as Wick renormalisation in the sense that
\begin{equation} \label{eq:Wick_Hermite}
\wM \Psi^{k} = C_{1}^{\frac{k}{2}} H_{k}(\Psi/\sqrt{C_{1}}) =: H_{k}(\Psi; C_{1})\;, 
\end{equation}
where $H_{k}(\cdot)$ is the $k$-th Hermite polynomial whose leading order coefficient is normalised to $1$. For example, we have
\begin{align*}
H_{1}(\Psi; C_1) = \Psi, \qquad H_{2}(\Psi;C_1) = \Psi^2 - C_1, \qquad H_{3}(\Psi;C_1) = \Psi^3 - 3C_1 \Psi. 
\end{align*}
Note that although we will always consider the case where $C_{1} \geq 0$, the above expression $H_{k}(\Psi; C_{1})$ actually does not require $C_{1}$ to be positive. 

We now describe the effect of $\wM$ on the canonical model $(\Pi, f)$. Following 
\cite[Sec.~8.1]{Hai14a} and \cite[Sec.~5.2]{HQ15}, for the map $\wM$ defined above, there is a unique pair of linear maps
\begin{align*}
\wDelta: \tT_{\ex} \rightarrow \tT_{\ex} \otimes \tT_{+}, \qquad \widehat{M}^{\Wck}: \tT_{+} \rightarrow \tT_{+}
\end{align*}
satisfying
\begin{equation} \label{eq:Wick_associated_maps}
\begin{split}
\widehat{M}^{\Wck} \sJ_{\ell} &= \mM (\sJ_{\ell} \otimes \id) \wDelta, \\
\widehat{M}^{\Wck} \sE^\beta_{\ell} &= \mM (\sE^\beta_{\ell} \otimes \id) \wDelta, \\
(\id \otimes \mM)(\Delta \otimes \id) \wDelta &= (\wM \otimes \widehat{M}^{\Wck}) \Delta, \\
\widehat{M}^{\Wck}(\tau_{1} \tau_{2}) &= (\widehat{M}^{\Wck} \tau_{1}) (\widehat{M}^{\Wck} \tau_{2}), \qquad \widehat{M}^{\Wck} X^{k} = X^{k}, 
\end{split}
\end{equation}
where $\mM: \tT_{+} \rightarrow \tT_{+}$ denotes the multiplication in the Hopf algebra $\tT_{+}$. 
As in \cite[Sec.~5.2]{HQ15}, one can verify that both $\widehat{M}^\Wck$ and $\wDelta$
have the relevant triangular structure, so that if, given an admissible 
model $(\Pi, f)$, we define $(\wPi, \wf)$ by
\begin{equation} \label{eq:Wick_model}
\wPi_{z} \tau = (\Pi_{z} \otimes f_{z}) \wDelta \tau, \qquad \wf_{z}(\sigma) = f_{z} (\widehat{M}^{\Wck} \sigma)\;,
\end{equation}
then $(\wPi, \wf)$ is again an admissible model. Furthermore, 
as a consequence of the second identity in \eqref{eq:Wick_associated_maps} 
and the fact that $M^\Wck$ commutes with $\eE^\beta$, 
if $(\Pi, f)$ satisfies \eqref{eq:model_epsilon}  for some $\eps$,
then so does $(\wPi, \wf)$. 

We now turn to describing the map $M_{0}$. For $n \geq 2$, we define 
linear maps $L_{n}$ and $L_{n}'$ on $\tT_\ex$ by setting
\begin{align*}
L_{n}: \qquad &\eE^{\frac{n}{2}-1} \big(\Psi^{n} \iI(\eE^{\frac{n}{2}-1} \Psi^{n}) \big) \mapsto n! \cdot \1, \\
&\eE^{\frac{n}{2}-1} \big( \Psi^{n} \iI(\eE^{\frac{n}{2}-1} \Psi^{n+1}) \big) \mapsto (n+1)! \cdot \Psi, \\
&\eE^{\frac{n}{2}-\frac{1}{2}} \big( \Psi^{n} \iI(\eE^{\frac{n}{2}-\frac{3}{2}} \Psi^{n}) \big) \mapsto n! \cdot \1, \quad n \geq 3, \\
&\eE^{\frac{n}{2}-\frac{1}{2}} \big( \Psi^{n+1} \iI( \eE^{\frac{n}{2}-\frac{3}{2}}\Psi^{n} ) \big) \mapsto (n+1)! \cdot \Psi, \quad n \geq 3, \\[.5em]
L_{n}': \qquad &\eE^{\frac{n}{2}-1} \big( \Psi^{n} \iI(\eE^{\frac{n}{2} - \frac{3}{2}} \Psi^{n}) \big) \mapsto n! \cdot 1, \quad n \geq 3, 
\end{align*}
(we use the convention $\eE^0 = \id$) and $L_{n} \tau = 0$, $L_{n}' \tau = 0$ for any other basis vector $\tau \in \wW$. 
Given these maps, we then consider maps on $\tT_\ex$ of the form
\begin{align*}
M_{0} := \exp \bigg(- \sum_{n \geq 2} C_{n} L_{n} - \sum_{n \geq 2} C_{n}' L_{n}' \bigg). 
\end{align*}
As we will see in \eqref{eq:renormalised_equation}, at the level of abstract equation, $M_0$ has the simple effect of adding a linear term to the right hand side of the equation. 
Actually, $M_{0}$ is equivalently given by
\begin{align*}
M_{0} = \id - \sum_{n \geq 2} C_{n} L_{n} - \sum_{n \geq 2} C_{n}' L_{n}'\;.
\end{align*}
Furthermore, it commutes with $\gG$ in the sense that $M_{0} \Gamma \tau = \Gamma M_{0} \tau$ for any $\tau \in \tT$ and $\Gamma \in \gG$. As a consequence, given an admissible model 
$(\bar{\Pi}, \bar{f})$, if we set 
\begin{equation} \label{eq:mass_model}
\bar{\Pi}_{z}^{M_{0}} \tau := \bar{\Pi}_{z} M_{0} \tau, \qquad \bar{f}_{z}^{M_{0}} \sigma = \bar{f}_{z}(\sigma)\;,
\end{equation}
then $(\bar{\Pi}^{M_{0}}, \bar{f}^{M_{0}})$ is also an admissible model.
Given $M = (M_{0},\wM)$ with $M_0$ and $\wM$ as above, we then define
the renormalised model $(\Pi^{M}, f^{M})$ by
\begin{equation} \label{eq:renormalised_model}
\Pi_{z}^{M} \tau = (\Pi_{z} \otimes f_{z}) \wDelta (M_{0}\tau), \qquad f_{z}^{M}(\sigma) = \wf_{z}(\widehat{M}^{\Wck} \sigma). 
\end{equation}

\begin{rmk}
	Note that although in many cases one has $(\Pi_{z}^{M} \tau)(z) = (\Pi_{z} M \tau)(z)$, this is in general not true. For example, for $\tau = \eE \Psi^{4}$, we have $(\Pi_{z}^{M} \tau)(z) = \epsilon \big(\widehat{\xi}^{4}(z) - 6C_{1} \widehat{\xi}^{2}(z) + 3C_{1}^{2}\big)$, while $(\Pi_{z} M \tau)(z) = \epsilon \big(\widehat{\xi}^{4}(z) - 6C_{1} \widehat{\xi}^{2}(z) \big)$. 
\end{rmk}

\section{Abstract fixed point problem}
\label{sec:sol}

In this section, we translate \eqref{e:rescaled} into a fixed point problem in a suitable
space of modelled distributions. It is natural to consider the fixed point problem
\begin{equation} \label{eq:abstract_eq}
\Phi = \pP \1_{+} \bigg(  \Xi - \sum_{j=4}^{m}  \lambda_{j} \qQ_{\leq 0} \widehat{\eE}^{\frac{j-3}{2}} \big(\qQ_{\leq 0} (\Phi^{j}) \big) - \sum_{j=0}^{3} \lambda_{j} \qQ_{\leq 0} (\Phi^{j}) \bigg) + \widehat{P} u_{0}, 
\end{equation}
where $\qQ_{\leq \alpha}$ denote the projection onto the subspace $\bigoplus_{\beta \leq \alpha} \tT_{\beta}$ in $\tT_{\text{ex}}$, $\widehat{P} u_{0}$ is the canonical lift of the solution to the deterministic heat equation with initial data $u_{0}$ to the regularity structure, and $\pP$ denotes the operator given by
\begin{align*}
\pP = \kK + \widehat{R} \rR, 
\end{align*}
where $\kK$ is the abstract integration operator defined from the truncated heat kernel $K$ as in \cite[Sec.~4]{Hai14a}, $\rR$ is the reconstruction operator, and $\widehat{R} u$ is the Taylor expansion of the smooth function $(P-K) * u$ up to order $\gamma$. 

To solve such a fixed point problem, at first glance, it seems that one can simply follow the procedure 
in  \cite[Sec.~7]{Hai14a} to obtain a unique solution to \eqref{eq:abstract_eq} in a space
$\dD^{\gamma,\eta}$ as in \cite[Sec.~6]{Hai14a} for 
suitable $\gamma$ and $\eta$. Unfortunately, as in \cite[Sec.~4]{HQ15}, 
this argument only works for sufficiently regular initial data (it needs
to be ``almost continuous'' for large values of $m$). 
Since the dynamical $\Phi^4_3$ model only has regularity $\cC^{\eta}$ for $\eta < -\frac{1}{2}$, 
this would prevent us from using a continuation argument to control the 
convergence of our models on any fixed time interval. In addition, such a continuation argument also requires one to be able to evaluate the reconstructed solution $\rR \Phi$ in a suitable space of distributions at any fixed time. However, as one can easily see, the solution to \eqref{eq:abstract_eq} contains the term $\Psi = \iI(\Xi)$ which has negative homogeneity, and a priori there is no clear way to give meaning to $\rR \Psi$ at any fixed time $t$. The second issue is not a serious problem here since, for the natural model constructed from space-time white noise, 
$\rR \Psi$ can indeed be regarded as a continuous function (in time) in a suitable space of distributions.
(See for example \cite{Hai14a,EJS13}.)

To resolve the first issue, we introduce $\eps$-dependent norms to enforce suitable control on both our admissible models and the initial condition as $\epsilon \rightarrow 0$. In a way, this allows us to ``trade'' the singularities near $t = 0$ and at small scales for powers of $\eps$. 

In what comes below, we will mainly follow \cite[Sec.~4]{HQ15} to build such weighted spaces. It turns out that the algebraic structure of these spaces are essentially the same as those in introduced in \cite{HQ15}, except that the values of $\gamma$ and $\eta$ are different. We will therefore mostly give statements and refer to \cite{HQ15} for detailed proofs.

\subsection{The $\epsilon$-dependent spaces and models}

Below, we use $\varphi$ to denote a space-time test function belonging to $\bB$, 
$\phi$ to denote such a test function that furthermore integrates to $0$, 
and $\psi$ to denote a test function that annihilates affine functions of the spatial variables. 

Recall that our definition of an admissible model in the previous section does not specify any 
relationship between its actions on $\tau$ and $\eE^{\beta}(\tau)$. In order to formulate the cancellation 
of the singularity in time by the small parameter $\epsilon$ in the limiting process 
$\epsilon \rightarrow 0$, we introduce the space of models $\sM_{\epsilon}$ which consists of 
all \textit{admissible models} $(\Pi, f)$ with the further restriction that
\begin{equs}[2] \label{eq:various_model_bounds}
\big|f_{z}(\sE^{\beta}_{\ell}(\tau)) \big| &\lesssim \epsilon^{\beta - |\ell| + |\tau|}, \quad&\quad \tau &\in \wW\;, \\
\big|\langle \Pi_{z}\tau, \psi_{z}^{\lambda} \rangle \big| &\lesssim \lambda^{\zeta} \cdot \epsilon^{|\tau|-\zeta}, \quad&\quad \tau &\in \uU, \phantom{1} \zeta = \frac{6}{5}\;.
\end{equs}
Here, all the bounds are to hold uniformly over all space-time points $z$ in compact sets, all $\lambda \in (0, \epsilon)$ 
and all test functions $\psi \in \bB$ that annihilate affine functions. We also require that, for 
some sufficiently large $\eta < -\frac{1}{2}$ (to be fixed below),
\begin{align*}
\sup_{t \in [0,1]} \| \Pi_0 \Psi (t, \cdot) \|_{\cC^{\eta}} < +\infty. 
\end{align*}
We will verify later in Section 5 that the models considered 
in this article do indeed belong to $\sM_{\epsilon}$ with uniform controls as $\eps \to 0$. 

We let $\|\Pi\|_{\epsilon}$ denote the smallest proportionality constant for both bounds in \eqref{eq:various_model_bounds}, and define a ``norm'' on $\sM_{\epsilon}$ by
\begin{align*}
\$ \Pi \$_{\epsilon} := \$ \Pi \$ + \|\Pi\|_{\epsilon} +  \sup_{t \in [0,1]} \| \Pi_0 \Psi (t, \cdot) \|_{\cC^{\eta}}, 
\end{align*}
where $\$ \Pi \$$ is the usual ``norm'' 
on admissible models introduced in Section~\ref{sec:admissible_model}. Again, these norms all depend on the compact set $\fK$ where the supremum of $z$ is taken over, which we have omitted for notational simplicity. 

\begin{rmk}
This is of course an abuse of notation since $\$ \Pi \$$ and $\|\Pi\|_{\epsilon}$ both depend
not only on $\Pi$ but also on $F$, and $F$ can in general not be recovered uniquely
from $\Pi$ and the knowledge that the model is admissible (unlike in the situations
considered in \cite{Hai14a}). We chose to nevertheless keep this notation for the sake of conciseness. Also, the norm $\$ \Pi \$_{\epsilon}$ depends not only on $\epsilon$ but also on $\eta$. Since we will fix the value $\eta<-\frac{1}{2}$ below, we omit $\eta$ in the notation. 
\end{rmk}

We compare two models in $\sM_{\epsilon}$ by
\begin{align*}
\$  \Pi; \bar{\Pi} \$_{\epsilon} = \$  \Pi; \bar{\Pi} \$ + \| \Pi - \bar{\Pi} \|_{\epsilon} + \sup_{t \in [0,1]} \| \Pi_0 \Psi(t, \cdot) - \bar\Pi_0 \Psi(t, \cdot) \|_{\cC^{\eta}}. 
\end{align*}
We also denote by $\sM_{0}$ the class of admissible models such that $f_{z}(\sE^{\beta}_{\ell}(\tau)) \equiv 0$. It is natural to compare two elements $(\Pi^{(\epsilon)}, \Gamma^{(\epsilon)}) \in \sM_{\epsilon}$ and $(\Pi, \Gamma) \in \sM_{0}$ by
\begin{align*}
\$ \Pi^{(\epsilon)}; \Pi \$_{\epsilon;0} = \$  \Pi^{(\epsilon)}; \Pi \$ + \| \Pi^{(\epsilon)} \|_{\epsilon} + \sup_{t \in [0,1]} \| \Pi_0^{(\epsilon)}\Psi(t, \cdot) - \Pi_0 \Psi(t, \cdot) \|_{\cC^{\eta}}. 
\end{align*}
Note that $\sM_{\epsilon}$ and $\sM_{\epsilon'}$ consists of exactly the same class of models for each $\epsilon, \epsilon' > 0$, but with different scales on their norms. The point here is that we will consider models with $\$ \Pi^{(\epsilon)}; \Pi \$_{\epsilon;0} \rightarrow 0$ for some limiting model $\Pi$. We first give a useful lemma.

\begin{lem} \label{le:model_sandwich}
There exists $C>0$ such that, for $\Pi \in \sM_{\epsilon}$ and $\tau \in \uU$, we have
\begin{equation} \label{eq:implied_model_bounds}
|\langle  \Pi_{z}\tau, \varphi_{z}^{\lambda} \rangle| < C \| \Pi \|_{\epsilon} \epsilon^{|\tau|} \;,\qquad
|\langle  \Pi_{z}\tau, \phi_{z}^{\lambda} \rangle| < C \| \Pi \|_{\epsilon} \lambda \epsilon^{|\tau|-1}\;, 
\end{equation}
uniformly over all $\lambda < \epsilon < 1$, all space-time points $z$ in compact sets and all test functions $\varphi, \phi \in \bB$ with the further restriction that $\phi$ integrates to $0$. 
\end{lem}
\begin{proof}
We first prove the second bound. Let $\phi$ be a test function integrating to $0$, let $\lambda \in (0,1]$, 
and let $N$ be the integer such that
\begin{equ}[e:defN]
\lambda \cdot 2^{N} \leq \epsilon < \lambda \cdot 2^{N+1}. 
\end{equ}
We then write $\phi_{z}^{\lambda}$ as a telescope sum by
\begin{align*}
\phi_{z}^{\lambda} = \sum_{k=0}^{N-1} (2^{-k} \phi_{z}^{\lambda \cdot 2^{k}} - 2^{-(k+1)} \phi_{z}^{\lambda \cdot 2^{k+1}}) + 2^{-N} \cdot \phi_{z}^{\lambda \cdot 2^{N}}
=: \sum_{k=0}^{N-1}\delta\phi_{z}^{\lambda,k}  + 2^{-N} \cdot \phi_{z}^{\lambda \cdot 2^{N}}\;. 
\end{align*}
For each $k$ appearing in this sum, 
$\delta\phi_{z}^{\lambda,k}$ is localised at scale $\lambda \cdot 2^{k} < \epsilon$ and 
integrates to $0$ since the function $\phi$ does. Furthermore, the factor $2^{-k}$ is chosen such that 
the integral of $2^{-k} \phi_{z}^{\lambda \cdot 2^{k}}$ against linear functions does not depend on $k$, 
so that $\delta\phi_{z}^{\lambda,k}$ annihilates all affine functions. 
Thus, we can use the second bound in \eqref{eq:various_model_bounds} to deduce that for each $k$, we have
\begin{align*}
\bigl|\langle \Pi_{z}\tau,  \delta\phi_{z}^{\lambda,k} \rangle\bigr| < C\| \Pi \|_{\epsilon} 2^{-k}  (\lambda  2^{k})^{\zeta} \cdot \epsilon^{|\tau| - \zeta}\;. 
\end{align*}
Summing over $k$ from $0$ to $N-1$, and using the fact that $\zeta > 1$ and 
$\lambda \cdot 2^{N} \sim \epsilon$, we 
conclude that 
$\sum_{k=0}^{N-1} \langle \Pi_{z}\tau, \delta\phi_{z}^{\lambda,k} \rangle < C \| \Pi \|_{\epsilon} \lambda \epsilon^{|\tau| -1}$. 
The same bound holds for the term $2^{-N} \cdot \phi_{z}^{\lambda \cdot 2^{N}}$ as a direct consequence of  \eqref{eq:model_analytic}, so we obtain the second bound in \eqref{eq:implied_model_bounds}. 

To prove the first one, fix a test function $\varphi$, and write it as
\begin{equation} \label{eq:telescope}
\varphi_{z}^{\lambda} = \sum_{k=0}^{N-1} (\varphi_{z}^{\lambda \cdot 2^{k}} - \varphi_{z}^{\lambda \cdot 2^{k+1}}) + \varphi^{\lambda \cdot 2^{N}}_{z}. 
\end{equation}
This time, each function in the parenthesis integrates to $0$ so we can use the second bound just proved above, and the first one follows easily.
\end{proof}

We now turn to dealing with the irregularity of the initial data. At this point, our definitions 
start to differ from those in \cite{HQ15} in order to encode the regularities of 
terms in \eqref{eq:abstract_eq}. We first introduce a new space for the initial condition $u_{0}$.

\begin{defn} \label{defn:weighted_function_space}
Let $\gamma \in (1,2)$, $\eta < 0$ and $\eps > 0$. The space $\cC^{\gamma,\eta}_{\epsilon}$ consists of $\cC^{\gamma}$ functions $f\colon \R^3 \to \R$ with norm
\begin{equation} \label{eq:weighted_function_space}
\| f \|_{\gamma, \eta; \epsilon} = \| f^{(\epsilon)} \|_{\cC^{\eta}} + \epsilon^{-\eta}\| f^{(\epsilon)} \|_{\infty} + \epsilon^{\gamma-\eta}  \sup_{|x-y|<\epsilon} \frac{|Df^{(\epsilon)}(x) - Df^{(\epsilon)}(y)|}{|x-y|^{\gamma-1}}. 
\end{equation}
Furthermore, we set $\cC^{\gamma,\eta}_{0} = \cC^{\eta}$. The distance between two elements $f^{(\epsilon)} \in \cC^{\gamma,\eta}_{\eps}$ and $f \in \cC^{\eta}$ is defined by
\begin{equation} \label{eq:weighted_function_difference}
\| f^{(\epsilon)}; f \|_{\gamma, \eta; \epsilon} = \| f^{(\epsilon)} - f \|_{\cC^{\eta}} + \epsilon^{-\eta} \|f^{(\epsilon)}\|_{\infty} +  \epsilon^{\gamma-\eta}  \sup_{|x-y|<\epsilon} \frac{|Df^{(\epsilon)}(x) - Df^{(\epsilon)}(y)|}{|x-y|^{\gamma-1}}. 
\end{equation}
\end{defn}

The reason we do not include a bound on $\|Df^{(\epsilon)}\|_{\infty}$ on the right hand side is 
that such a bound follows the bounds on $\|f^{(\epsilon)}\|_{\infty}$ and $\|Df^{(\epsilon)}\|_{\cC^{\gamma-1}}$. More precisely, one has

\begin{lem}
There exists a constant $C$ such that, for every 
$f^{(\epsilon)} \in \cC^{\gamma,\eta}_{\epsilon}$ one has 
\begin{equation} \label{eq:implied_function_bounds}
\| Df^{(\epsilon)} \|_{\infty} < C \| f^{(\epsilon)} \|_{\gamma,\eta;\epsilon} \cdot \epsilon^{\eta - 1}\;.
\end{equation}
\end{lem}
\begin{proof}
The proof is straightforward and we leave it as an exercise.
\end{proof}

One should think of functions in $\cC^{\gamma,\eta}_{\epsilon}$ as behaving like elements of $\cC^{\eta}$ at large scales, while being of class $\cC^{\gamma}$ at small scales, with $\eps$
determining where the cutoff between ``small'' and ``large'' lies. 
The reason why only $f^{(\epsilon)}$ appears in the last two terms of \eqref{eq:weighted_function_difference} is that these two quantities 
are not even finite for general $f \in \cC^{\eta}$.

Following \cite[Sec.~3.5]{HQ15}, we define $\dD^{\gamma,\eta}$ space to be the set of functions $U$ taking values in $\tT$ with norm
\begin{align*}
\|U\|_{\gamma,\eta} := &\sup_{z} \sup_{|\tau|<\gamma} |U(z)|_{\tau} + \sup_{z} \sup_{|\tau|<\gamma} \frac{|U(z)|_{\tau}}{\sqrt{|t|}^{(\eta-|\tau|) \wedge 0}} \\
&+ \sup_{|z-z'|<\sqrt{|t| \wedge |t'|}} \sup_{|\tau| < \gamma} \frac{|U(z) - \Gamma_{z,z'}U(z')|_{\tau}}{|z-z'|^{\gamma-|\tau|} \sqrt{|t| \wedge |t'|}^{\eta - \gamma}}, 
\end{align*}
where $|z-z'|$ is measured in the parabolic distance defined in \eqref{eq:parabolic_metric}. Note that this definition is slightly different from the original one in \cite{Hai14a} in the sense that it allows $U(z)$ to have components in $\tT_{\geq \gamma}$. We now introduce the weighted spaces $\dD^{\gamma,\eta}_{\epsilon}$ that are suitable for our fixed point problem. 

\begin{defn}\label{def:modelled_distribution}
For each $\epsilon, \gamma, \eta$, and each model $(\Pi, \Gamma) \in \sM_{\epsilon}$, the space $\dD^{\gamma,\eta}_{\epsilon}$ consists of modelled distributions $U$ with norm given by
\begin{align*}
\| U \|_{\gamma,\eta;\epsilon} &= \| U \|_{\gamma,\eta} + \sup_{z} \sup_{\tau} \frac{|U(z)|_{\tau}}{\epsilon^{(\eta - |\tau|) \wedge 0}} + \sup_{(z,z') \in D_{\epsilon}} \sup_{|\tau| < \gamma} \frac{|U(z) - \Gamma_{z,z'}U(z')|_{\tau} }{|z-z'|^{\gamma - |\tau|} \epsilon^{\eta - \gamma} }\;.
\end{align*}
Here, the supremum is taken over all space-time points $(z,z') \in D_{\epsilon}$, defined by
\begin{align*}
D_{\epsilon} = \big\{ (z,z'): |z-z'| < \epsilon \wedge \sqrt{|t| \wedge |t'|}  \big\}\;, 
\end{align*}
where $z=(t,x)$, $z'= (t',x')$,
and $\| \cdot \|_{\gamma,\eta}$ is the norm of the usual $\dD^{\gamma,\eta}$ spaces introduced in \cite[Sec.~6]{Hai14a}. 
\end{defn}

In short, the above definition says that modelled distributions $U$ in $\dD^{\gamma,\eta}_{\epsilon}$ satisfy the bounds
\begin{align*}
|U(z)|_{\tau} &\lesssim (\epsilon + \sqrt{|t|})^{(\eta - |\tau|) \wedge 0}, \\
|U(z) - \Gamma_{z,z'}U(z')|_{\tau} &\lesssim |z-z'|^{\gamma-|\tau|} (\epsilon + \sqrt{|t| \wedge |t'|})^{\eta - \gamma}. 
\end{align*}
Note that $\dD^{\gamma,\eta}_{\epsilon}$ is a linear space once the model is fixed, and so the distance between two elements can be simply compared by $\| U - \bar{U} \|_{\gamma,\eta;\epsilon}$. Also, in all the cases we consider below, $\eta$ is always smaller than the regularity of the sector in consideration. Thus, we will have $\eta < |\tau|$, and can simply replace $(\eta - |\tau|) \wedge 0$ by $\eta - |\tau|$ in all the situations below. Similar as before, we compare two elements $U^{(\epsilon)} \in \dD^{\gamma,\eta}_{\epsilon}$ and $U \in \dD^{\gamma,\eta}$ by
\begin{align*}
\| U^{(\epsilon)}; U \|_{\gamma,\eta;\epsilon} &= \| U^{(\epsilon)}; U \|_{\gamma,\eta} + \sup_{z} \sup_{\tau} \frac{|U^{(\epsilon)}(z)|_{\tau}}{\epsilon^{(\eta-|\tau|) \wedge 0}} \\
&\phantom{1}+ \sup_{(z,z') \in D_{\eps}} \sup_{|\tau| < \gamma} \frac{|U^{(\epsilon)}(z) - \Gamma_{z,z'} U^{(\epsilon)}(z')|}{|z-z'|^{\gamma-|\tau|} \epsilon^{\eta - \gamma}}. 
\end{align*}
The reason why only $U^{(\epsilon)}$ appears on the latter two terms on the right hand side above is the same as before: these quantities are in general not finite for $U \in \dD^{\gamma,\eta}$. The main motivation for the introduction of these $\epsilon$-dependent spaces is that they contain the solution
to the heat equation with initial condition in $\cC_\eps^{\gamma,\eta}$, with bounds independent of $\eps$. This is the content of the following proposition, the proof of which is identical
to that of \cite[Prop.~4.7]{HQ15}, so we do not repeat the details here.

\begin{prop} \label{pr:lift_initial}
Let $\eta < 0$, $\gamma\in (1,2)$, $\eps \in (0,1]$, and $u \in \cC^{\gamma, \eta}_{\epsilon}$. Let $\widehat{P} u$ denote the canonical lift of the harmonic extension of $u$ via its truncated Taylor expansion of order $\gamma$. 
Then, $\widehat{P} u \in \dD^{\gamma,\eta}_{\epsilon}$ and one has
\begin{align*}
\| \widehat{P} u \|_{\gamma,\eta;\epsilon} < C \| u \|_{\gamma,\eta;\epsilon}. 
\end{align*}
Furthermore, if $u^{(\epsilon)} \in \cC^{\gamma, \eta}_{\epsilon}$ and $u \in \cC^{\eta}$, then one has
\begin{align*}
\| \widehat{P}u^{(\epsilon)}; \widehat{P}u \|_{\gamma,\eta;\epsilon} < C \| u^{(\epsilon)}; u \|_{\gamma,\eta;\epsilon}. 
\end{align*}
\end{prop}

The following proposition will be needed later when we continue local solutions 
up to their (potential) explosion time. It says that the initial data of 
the restarted solution still belongs to the $\cC^{\gamma,\eta}_{\epsilon}$ 
space with norms uniform in $\epsilon$.

\begin{prop} \label{pr:restart_solution}
Let $\gamma \in (1, \frac{6}{5})$ and $\eta \in (-\frac{m+1}{2m+1}, -\frac{1}{2})$. Let $(\Pi^{\epsilon}, f^{\epsilon}) \in \sM_{\epsilon}$. Let $\uU$ be a sector of the regularity structure as defined in \eqref{eq:graded_vector_space}. If $\rR^{\epsilon}$ is the associated reconstruction map for $\dD^{\gamma,\eta}_{\epsilon}(\uU)$ and $U^{(\epsilon)} \in \dD^{\gamma,\eta}_{\epsilon}(\uU)$ is the abstract solution to \eqref{eq:abstract_eq}, then for every $t > 0$, $u_{t}^{(\epsilon)} := \rR^{\epsilon} U^{(\epsilon)}(t, \cdot) $ belongs to $\cC^{\gamma,\eta}_{\epsilon}$ with
\begin{align*}
\| u_{t}^{(\epsilon)} \|_{\gamma,\eta;\epsilon} < C \| U^{(\epsilon)} \|_{\gamma,\eta;\epsilon} \$ \Pi^{(\epsilon)} \$_{\epsilon}. 
\end{align*}
If $(\Pi, f)$ is another such model with reconstruction operator $\rR$, and $U \in \dD^{\gamma,\eta}$ solves \eqref{eq:abstract_eq} based on $\Pi$, then $u_{t} := \rR U (t, \cdot)$ belongs to $\cC^{\eta}$ and one has
\begin{align*}
\| u_{t}^{(\epsilon)}; u_{t} \|_{\gamma,\eta;\epsilon} \lesssim \|U^{(\epsilon)}; U\|_{\gamma,\eta;\epsilon} (\$ \Pi \$ + \$ \Pi^{\epsilon} \$_{\epsilon}) +  \$ \Pi^{\epsilon}; \Pi \$_{\epsilon,0} (\|U^{(\epsilon)}\|_{\gamma,\eta;\epsilon} + \|U\|_{\gamma,\eta})\;. 
\end{align*}
\end{prop}
\begin{proof}
We first prove the first claim of the proposition. For that, we bound separately the three terms 
appearing in the definition \eqref{eq:weighted_function_space} of the spaces $\cC^{\gamma,\eta}_{\epsilon}$. 
We first notice that any solution $U^{(\epsilon)}$ to \eqref{eq:abstract_eq} is necessarily of the form
\begin{align*}
U^{(\epsilon)}(z) = \Psi + V^{(\epsilon)}(z)\;.
\end{align*}
Since
the structure group acts trivially on $\Psi$, the constant function $\Psi$ belongs to 
all spaces $\dD^{\gamma,\eta}_{\epsilon}$, so that if $U^{(\epsilon)} \in \dD^{\gamma,\eta}_{\epsilon}$, then 
so does $V^{(\eps)}$.
Since, in the above decomposition, $V^{(\epsilon)}(z)$ belongs to the linear span of 
$\{\1\}\cup \{\tau\,:\,|\tau| > 0\}$, the desired bound for 
$\| \rR^{\epsilon} V(t, \cdot) \|_{\cC^{\eta}}$ follows from \cite[Prop.~3.28]{Hai14a}.
Regarding the term $\Psi$, one has $\rR^{\epsilon} \Psi = \Pi_0^{\epsilon} \Psi$ 
so that, by the definition of $\sM_{\epsilon}$, we have
\begin{align*}
\sup_{t \in [0,1]} \| \big(\rR^{\epsilon} \Psi\big)(t, \cdot) \|_{\cC^{\eta}} < C \$ \Pi^{\epsilon} \$_{\epsilon}, 
\end{align*}
and the required bound for $\|u_{t}^{(\epsilon)}\|_{\cC^{\eta}}$ thus follows. 

For the remaining two terms on the right hand side of \eqref{eq:weighted_function_space}, we will prove a stronger bound by showing $u^{(\epsilon)} = \rR^{\epsilon} U^{(\epsilon)}$ is a space-time function with desired regularity, rather just being a function in space for fixed time.

For the second term, since the lowest homogeneity in $\uU$ is $-\frac{1}{2} - \kappa$, an application of the reconstruction theorem together with Lemma~\ref{le:model_sandwich} gives
\begin{align*}
\sup_{\lambda < \epsilon} \sup_{z} \sup_{\varphi \in \bB} |\langle u^{(\epsilon)}, \varphi_{z}^{\lambda} \rangle| < C \|U^{(\epsilon)}\|_{\gamma,\eta;\epsilon} \$ \Pi^{\epsilon} \$_{\epsilon} \cdot \epsilon^{-\frac{1}{2} - \kappa}. 
\end{align*}
On the other hand, it follows directly from the definition of a model that
\begin{align*}
\sup_{\lambda \geq \epsilon} \sup_{z} \sup_{\varphi \in \bB} \lambda^{\frac{1}{2} + \kappa} |\langle u^{(\epsilon)}, \varphi_{z}^{\lambda} \rangle| < C \|U^{(\epsilon)}\|_{\gamma,\eta;\epsilon} \$ \Pi^{\epsilon} \$_{\epsilon}. 
\end{align*}
Combining the above two bounds and using the fact that $\kappa$ is arbitrarily small so that $\eta < -\frac{1}{2} - \kappa$, we conclude that $u^{(\epsilon)}$ is a continuous function with
\begin{align*}
\epsilon^{-\eta} \|u^{(\epsilon)}\|_{\infty} < C \epsilon^{-\eta-\frac{1}{2}-\kappa} \cdot \|U^{(\epsilon)}\|_{\gamma,\eta;\epsilon} \$ \Pi^{\epsilon} \$_{\epsilon}. 
\end{align*}
We now turn to the third term on the right hand side of \eqref{eq:weighted_function_space}. In order to show $Du^{(\epsilon)} \in \cC^{\gamma-1}$, we test it against test functions that integrate to $0$. Using the definition of the distributional derivative and then the triangle inequality, we get
\begin{align*}
\lambda^{1-\gamma} |\langle Du^{(\epsilon)}, \phi_{z}^{\lambda} \rangle| \leq \lambda^{-\gamma}  |\langle \Pi_{z}^{\epsilon} U^{(\epsilon)}(z), (D\phi)_{z}^{\lambda} \rangle| + \lambda^{-\gamma} |\langle u^{(\epsilon)} - \Pi_{z}^{\epsilon}U^{(\epsilon)}(z), (D \phi)_{z}^{\lambda}  \rangle|. 
\end{align*}
It follows from the reconstruction theorem that the second term on the right hand side above is uniformly bounded by a constant. For the first term, since the assumption that $\phi$ integrates to $0$ implies $D \phi$ annihilates affine functions, we can use the second bound in \eqref{eq:various_model_bounds} to obtain
\begin{align*}
\lambda^{-\gamma}  |\langle \Pi_{z}^{\epsilon} U^{(\epsilon)}(z), (D\phi)_{z}^{\lambda} \rangle| < C \|U^{(\epsilon)}\|_{\gamma,\eta;\epsilon} \$ \Pi^{\epsilon} \$_{\epsilon} \cdot \lambda^{\zeta-\gamma} \epsilon^{-\frac{1}{2}-\kappa-\zeta}, 
\end{align*}
where we again used the fact that the lowest homogeneity in $\uU$ is $-\frac{1}{2}-\kappa$. The desired bound then follows immediately. 

For the second claim, the only problematic term is $\|u_{t}^{(\epsilon)}; u_{t}\|_{\cC^{\eta}}$, but again the desired bound for this term follows in the same way as $\|u_{t}^{(\epsilon)}\|_{\cC^{\eta}}$. 
\end{proof}

Before we proceed to further properties of the $\dD^{\gamma,\eta}_{\epsilon}$ spaces, we first make a few remarks about these spaces and our notation.

\begin{itemize}
	\item The set $D_{\eps}$ in Definition~\ref{def:modelled_distribution} is taken to be $\{ |z-z'| < \epsilon \wedge \sqrt{|t| \wedge |t'|} \}$. This is sufficient since for the pairs $(z,z')$ such that $\epsilon \leq |z-z'| < \sqrt{|t| \wedge |t'|}$, we have $\epsilon + \sqrt{|t| \wedge |t'|} < 2 \sqrt{|t| \wedge |t'|}$, so the bound on the last term in Definition~\ref{def:modelled_distribution} follows automatically from the bound on $\| \cdot \|_{\gamma,\eta}$.
	
\item We use the notation $\| F - \bar F \|_{\gamma,\eta;\epsilon}$ to compare two functions in the \textit{same} $\dD^{\gamma,\eta}_{\epsilon}$ space with the \textit{same} underlying model. On the other hand, whenever we write $\| F; \bar F \|_{\gamma,\eta;\epsilon}$, it should be understood that we are comparing $F \in \dD^{\gamma,\eta}_{\epsilon}$ with $\bar F \in \dD^{\gamma,\eta}$, typically based
on a different model. As we will 
never compare two functions belonging to $\dD^{\gamma,\eta}_{\epsilon}$ spaces with the same 
$\epsilon$ but different underlying models, these notations are sufficient. 
\end{itemize}

It turns out that these spaces behave as expected under multiplication and action of $\widehat{\eE}^{k}$ and $\pP$. We state a few of the properties we will be using later;
all the proofs can be found in \cite[Sec.~4.3]{HQ15}.

\begin{prop} \label{pr:multiplication}
Let $U_{i} \in \dD^{\gamma_{i}, \eta_{i}}_{\epsilon}(V^{(i)})$ for $i = 1, 2$, where $V^{(1)}$ and $V^{(2)}$ are sectors of respective regularities $\alpha_{1}$ and $\alpha_{2}$. If
\begin{align*}
\gamma = (\gamma_{1}+\alpha_{2}) \wedge (\gamma_{2}+\alpha_{1}), \qquad \eta = (\eta_{1}+\alpha_{2}) \wedge (\eta_{2}+\alpha_{1}) \wedge (\eta_{1}+\eta_{2}), 
\end{align*}
then their pointwise product $U = U_{1} U_{2}$ is in $\dD^{\gamma,\eta}_{\epsilon}$ with
\begin{align*}
\| U \|_{\gamma, \eta; \epsilon} < C  \|U_{1}\|_{\gamma_{1}, \eta_{1}; \epsilon}  \|U_{2}\|_{\gamma_{2}, \eta_{2}; \epsilon}.  
\end{align*}
Furthermore, if $\bar{U}_{i} \in \dD^{\gamma_{i}, \eta_{i}}_{\epsilon}$, then $\bar{U} = \bar{U}_{1} \bar{U}_{2} \in \dD^{\gamma,\eta}$ with the same $\eta, \gamma$ as above, and we have
\begin{align*}
\| U; \bar{U} \|_{\gamma,\eta;\epsilon} < C \big( \| U_{1}; \bar{U}_{1} \|_{\gamma,\eta;\epsilon} + \| U_{2}; \bar{U}_{2} \|_{\gamma,\eta;\epsilon} + \| \Gamma - \bar{\Gamma} \| \big), 
\end{align*}
where $C$ is proportional to $\sum_{i} (\| U_{i} \| + \| \bar{U}_{i} \|) + \| \Gamma \| + \| \bar{\Gamma} \|$. 
\end{prop}

\begin{prop} \label{pr:projection}
	Let $U \in \dD^{\gamma,\eta}_{\epsilon}$ with $\eta \leq \gamma$. If $\alpha \geq \gamma$, then $\qQ_{\leq \alpha} U \in \dD^{\gamma,\eta}_{\epsilon}$ with
	\begin{align*}
	\| \qQ_{\leq \alpha} U \|_{\gamma,\eta;\epsilon} \lesssim \| U \|_{\gamma,\eta;\epsilon}. 
	\end{align*}
\end{prop}

\begin{prop} \label{pr:epsilon}
Let $U \in \dD^{\gamma, \eta}_{\epsilon}$ with $\gamma \in (-\beta, 1-\beta)$. Then $\widehat{\eE}^{\beta} U \in \dD^{\gamma', \eta'}_{\epsilon}$ with
\begin{align*}
\gamma' = (\gamma + \beta) \wedge \inf_{|\tau| < \gamma} (\gamma - |\tau|), \qquad \eta' = \eta + \beta, 
\end{align*}
and we have the bound
\begin{align*}
\| \widehat{\eE}^{\beta}U \|_{\gamma', \eta'; \epsilon} < C (1 + \| \Pi \|_{\epsilon}) \| U \|_{\gamma, \eta; \epsilon}. 
\end{align*}
In addition, if $\bar{U} \in \dD^{\gamma,\eta}$ with model $\bar{\Pi} \in \sM_{0}$, we have
\begin{align*}
\| \widehat{\eE}^{\beta}U; \widehat{\eE}^{\beta}\bar{U} \|_{\gamma', \eta'; \epsilon} < C (1 + \| \Pi \|_{\epsilon}) \big( \| U; \bar{U} \|_{\gamma,\eta;\epsilon}  + |\!|\!| \Pi; \bar{\Pi} |\!|\!|_{\epsilon;0} \big)
\end{align*}
with the same $\gamma'$ and $\eta'$. 
\end{prop}

\begin{prop} \label{pr:integration}
Let $U \in \dD^{\gamma,\eta}_{\epsilon}(V)$, where $V$ is a sector of regularity $\alpha$ with $-2 < \eta < \gamma \wedge \alpha$. Then, provided that $\gamma$ and $\eta$ are not integers, we have $\pP U \in \dD^{\bar{\gamma}, \bar{\eta}}_{\epsilon}$ with $\bar{\gamma} = \gamma + 2$ and $\bar{\eta} = \eta + 2$, and we have the bound
\begin{align*}
\| \pP U \|_{\bar{\gamma}, \bar{\eta}; \epsilon} < C \| U \|_{\gamma,\eta;\epsilon}. 
\end{align*}
Furthermore, if $\bar{U} \in \dD^{\gamma,\eta}_{0}$, then we also have
\begin{align*}
\| \pP U; \pP \bar{U} \|_{\bar{\gamma}, \bar{\eta}; \epsilon} < C \big( \| U; \bar{U} \|_{\gamma,\eta;\epsilon} + \$  \Pi, \bar{\Pi} \$_{\epsilon} \big). 
\end{align*}
\end{prop}

\subsection{Solution to the fixed point problem and convergence}

We now have all the ingredients in place to build our solution with uniform (in $\epsilon$) bounds in suitable $\dD^{\gamma,\eta}_{\epsilon}$ spaces. The equation we consider is of a general form
that it sufficiently flexible to cover all the concrete cases to be considered later. 
We first show the existence and uniqueness of local solutions.

\begin{thm}\label{th:fixed_pt}
Let $m \geq 1$, $\gamma \in (1, \frac{6}{5})$, $\eta \in \big( -\frac{m+1}{2m+1}, -\frac{1}{2} \big)$, and $\kappa > 0$ be sufficiently small. Let $\phi_{0} \in \cC^{\gamma, \eta}_{\epsilon}$, and consider the equation
\begin{equation} \label{eq:fixed_pt}
\Phi = \pP \1_{+} \bigg(  \Xi - \sum_{j=4}^{m} \lambda_{j} \qQ_{\leq 0} \heE^{\frac{j-3}{2}} (\qQ_{\leq 0} \Phi^{j}) - \sum_{j=0}^{3} \lambda_{j} \qQ_{\leq 0} (\Phi^{j}) \bigg) + \widehat{P} \phi_{0}. 
\end{equation}
Then, for every sufficiently small $\epsilon$ and every model in $\sM_{\epsilon}$, there exists $T > 0$ such that the equation \eqref{eq:fixed_pt} has a unique solution in $\dD^{\gamma,\eta}_{\epsilon}$ up to time $T$. Moreover, $T$ can be chosen uniformly over any fixed bounded set of initial data in $\cC^{\gamma,\eta}_{\epsilon}$, any bounded set of models in $\sM_{\epsilon}$, any bounded set of parameters $\lambda_{j}^{(\epsilon)}$, and all sufficiently small $\epsilon$. 

Let $\phi_{0}^{(\epsilon)}$ be a sequence of elements in $\cC^{\gamma,\eta}_{\epsilon}$ such that $\| \phi_{0}^{(\epsilon)}; \phi_{0} \|_{\gamma,\eta;\epsilon} \rightarrow 0$ for some $\phi_{0} \in \cC^{\eta}$, $\Pi^{\epsilon} \in \sM_{\epsilon}$, $\Pi \in \sM_{0}$ be models such that $|\!|\!| \Pi^{\epsilon}; \Pi |\!|\!|_{\epsilon;0} \rightarrow 0$, and let $\lambda_{j}^{(\epsilon)} \rightarrow \lambda_{j}$ for each $j$. If $\Phi \in \dD^{\gamma,\eta}$ solves the fixed point problem \eqref{eq:fixed_pt}with model $\Pi$, initial data $\phi_{0}$ and coefficients $\lambda_{j}$ up to time $T$, then for all small enough $\epsilon$, there is a unique solution $\Phi^{(\epsilon)} \in \dD^{\gamma,\eta}_{\epsilon}$ to \eqref{eq:fixed_pt} with $\Pi^{\epsilon}, \phi_{0}^{(\epsilon)}$ and $\lambda_{j}^{(\epsilon)}$ up to the same time $T$, and we have
\begin{align*}
\lim_{\epsilon \rightarrow 0}  \| \Phi^{(\epsilon)}; \Phi  \|_{\gamma,\eta;\epsilon} = 0\;,\quad
\lim_{\eps \to 0} \sup_{t \in [0,T]} \|(\rR^{(\epsilon)} \Phi^{(\epsilon)})(t,\cdot)-(\rR \Phi)(t,\cdot)\|_\eta = 0\;. 
\end{align*}
\end{thm}
\begin{proof}
We first prove that the fixed point problem \eqref{eq:fixed_pt} can be solved in $\dD^{\gamma,\eta}_{\epsilon}$ with local existence time uniform in $\epsilon$. Let $\mM_{T}^{(\epsilon)}$ denote the map
\begin{equation} \label{eq:fixed_pt_map}
\mM^{(\epsilon)}_T (\Phi) = \pP \1_{+} \bigg(  \Xi - \sum_{j=4}^{m} \lambda_{j} \qQ_{\leq 0} \heE^{\frac{j-3}{2}} (\qQ_{\leq 0} \Phi^{j}) - \sum_{j=0}^{3} \lambda_{j} \qQ_{\leq 0} (\Phi^{j}) \bigg) + \widehat{P} \phi_{0}\;,
\end{equation}
where $T$ denotes the length of the time interval on which the argument $\Phi$ is defined. Note that although the terms in \eqref{eq:fixed_pt_map} does not explicitly depend on $T$, their domains of definition and norms $\|\cdot\|_{\gamma,\eta;\epsilon}$ do depend on it. We will show that, for $T$ sufficiently small, $\mM_{T}^{(\epsilon)}$ is a contraction mapping a centered ball in $\dD^{\gamma,\eta}_{\epsilon}$ of a large enough radius $\Lambda$ into a ball of radius $\frac{\Lambda}{2}$. 
In what follows, we will omit the subscript $T$ and simply denote this map as $\mM^{(\epsilon)}$.

We first show that $\mM^{(\epsilon)}$ maps $\dD^{\gamma,\eta}_{\epsilon}$ into itself. By Proposition~\ref{pr:lift_initial}, we have $\widehat{P}\phi_{0}^{(\epsilon)} \in \dD^{\gamma,\eta}_{\epsilon}$. In addition, the noise term $\pP \1_{+} \Xi$ also belongs to $\dD^{\gamma,\eta}_{\epsilon}$. As for the non-linearity, if $j \leq 3$, it is straightforward to see that $\Phi^{j} \in \dD^{\delta, 3 \eta}_{\epsilon}$ for some positive $\delta$. We can choose $\delta$ small enough so that there is no basis vector with homogeneity between $0$ and $2 \delta$, and Proposition~\ref{pr:projection} then implies that $\qQ_{\leq 0} (\Phi^{j}) = \qQ_{\leq 2 \delta} (\Phi^{j}) \in \dD^{\delta,3 \eta}_{\epsilon}$. As an immediate application of Proposition~\ref{pr:integration}, one sees that the map
\begin{align*}
\Phi \mapsto \pP \1_{+} \bigg( \sum_{j=0}^{2} \lambda_{j} \qQ_{\leq 0} (\Phi^{j}) \bigg)
\end{align*}
is locally Lipschitz from $\dD^{\gamma,\eta}_{\epsilon}$ into $\dD^{\delta+2,3 \eta+2}_{\epsilon}$. By the assumption on the range of $\gamma$ and $\eta$, we have
\begin{equ}
\delta + 2 > \gamma, \qquad 3 \eta+2 > \eta, 
\end{equ}
so we have the natural embedding $\dD^{\delta+2, 3\eta+2}_{\epsilon} \emb \dD^{\gamma,\eta}_{\epsilon}$, and hence the map is locally Lipschitz from $\dD^{\gamma,\eta}_{\epsilon}$ into itself. Moreover, since the kernel is non-anticipative, by \cite[Thm.$7.1$, Lem.~7.3]{Hai14a} and the definition of the $\dD^{\gamma,\eta}_{\epsilon}$ space, we know the local Lipschitz constant is bounded by $(T+\epsilon)^{\theta}$ for some positive $\theta$. More precisely, there exists $C, \theta > 0$ such that
\begin{align*}
\bigg\| \pP \1_{+} \bigg( \sum_{j=0}^{3} \lambda_{j} \qQ_{\leq 0}(\Phi^{j}) \bigg) \bigg\|_{\gamma,\eta;\epsilon} < C (T+\epsilon)^{\theta} \|\Phi\|_{\gamma,\eta;\epsilon}. 
\end{align*}
We now turn to the nonlinear term $\pP \1_{+} \big( \qQ_{\leq 0} \heE^{\frac{j-3}{2}} \qQ_{\leq 0} (\Phi^{j}) \big)$ for $j \geq 4$. Let
\begin{align*}
\gamma_{j} = \gamma - \frac{j-1}{2} - (j-1) \kappa, \qquad \eta_{j} = j \eta, \qquad \bar{\eta}_{j} = j \eta + \frac{j-3}{2}. 
\end{align*}
Then by Propositions~\ref{pr:multiplication} and~\ref{pr:projection}, we have $\qQ_{\leq 0}(\Phi^{j}) \in \dD^{\gamma_{j}, \eta_{j}}_{\epsilon}$ with
\begin{align*}
\| \qQ_{\leq 0}(\Phi^{j}) \|_{\gamma_{j},\eta_{j};\epsilon} < C \| \Phi \|_{\gamma,\eta;\epsilon}^{j}. 
\end{align*}
The assumption $\gamma > 1$ implies $\gamma_{j} > - \frac{j-3}{2}$ if $\kappa$ is sufficiently small, so applying Proposition~\ref{pr:epsilon} with $\beta = \frac{j-3}{2}$, we know that there exists $\delta > 0$ such that $\heE^{\frac{j-3}{2}} \qQ_{\leq 0}(\Phi^{j}) \in \dD^{\delta, \bar{\eta}_{j}}_{\epsilon}$ with
\begin{align*}
\| \heE^{\frac{j-3}{2}} \qQ_{\leq 0}(\Phi^{j}) \|_{\delta,\bar{\eta}_{j};\epsilon} < C (1 + \| \Pi \|_{\epsilon}) \|\Phi\|_{\gamma,\eta;\epsilon}^{j}. 
\end{align*}
Similar as before, we can again choose $\delta$ to be small enough so that $\qQ_{\leq 0} \heE^{\frac{j-3}{2}} \qQ_{\leq 0}(\Phi^{j}) = \qQ_{\leq 2 \delta} \heE^{\frac{j-3}{2}} \qQ_{\leq 0}(\Phi^{j})$ also belongs to $\dD^{\delta,\bar{\eta}_{j}}_{\epsilon}$ with the same bound. Since $\bar{\eta}_{j} > -2$, an application of Proposition~\ref{pr:integration} implies that there exists $\theta > 0$ such that
\begin{align*}
\| \pP \1_{+} \heE^{\frac{j-3}{2}} \big(\qQ_{\leq 0} (\Phi^{j}) \big) \|_{\gamma,\eta;\epsilon} < C (T + \epsilon)^{\theta} (1 + \|\Pi\|_{\epsilon}) \|\Phi\|_{\gamma,\eta;\epsilon}^{j+2}. 
\end{align*}
This shows $\mM^{(\epsilon)}$ indeed maps $\dD^{\gamma,\eta}_{\epsilon}$ into itself. In particular, if $\Lambda$ is big enough with
\begin{align*}
\|\Phi\|_{\gamma,\eta;\epsilon} < \Lambda, \qquad \|u_{0}^{(\epsilon)}\|_{\gamma,\eta;\epsilon} < \frac{\Lambda}{C}, 
\end{align*}
then we can choose $T$ small enough depending on $\Lambda$, $\|\Pi\|_{\epsilon}$ and $\lambda_{j}^{(\epsilon)}$'s only such that
\begin{align*}
\| \mM^{(\epsilon)} (\Phi) \|_{\gamma,\eta;\epsilon} < \frac{\Lambda}{2}. 
\end{align*}
In order to show $\mM^{(\epsilon)}$ is also a contraction for small $T$, we first note that since there is only one model concerned in $\sM_{\epsilon}$, we can simply compare the difference $\mM^{(\epsilon)} (\Phi) - \mM^{(\epsilon)} (\tilde{\Phi})$ for two elements $\Phi, \tilde{\Phi} \in \dD^{\gamma,\eta}_{\epsilon}$. In fact, we have
\begin{align*}
\mM^{(\epsilon)} (\Phi) - \mM^{(\epsilon)} (\tilde{\Phi}) = &- \sum_{j=4}^{m} \sum_{k=0}^{j-1} \lambda_{j} \pP \1_{+} \qQ_{\leq 0} \heE^{\frac{j-3}{2}} \qQ_{\leq 0} \big( (\Phi-\tilde{\Phi}) \Phi^{j-1-k} \tilde{\Phi}^{k} \big) \\
&- \qQ_{\leq 0} (\Phi - \tilde{\Phi} ) \big( \lambda_{3} (\Phi^{2} + \Phi \tilde{\Phi} + \tilde{\Phi}^{2}) + \lambda_{2}  (\Phi + \tilde{\Phi}) + \lambda_{1} \big). 
\end{align*}
By linearity, $\Phi - \tilde{\Phi} \in \dD^{\gamma,\eta}_{\epsilon}$, so all the bounds obtained above also apply for $\mM^{(\epsilon)}(\Phi) - \mM^{(\epsilon)}(\tilde{\Phi})$ except that one power of $\|\Phi\|_{\gamma,\eta;\epsilon}$ is replaced by $\|\Phi - \tilde{\Phi}\|_{\gamma,\eta;\epsilon}$. Thus, we get
\begin{align*}
&\phantom{111} \|\mM^{(\epsilon)} (\Phi) - \mM^{(\epsilon)} (\tilde{\Phi})\|_{\gamma,\eta;\epsilon} \\
&< C (T + \epsilon)^{\theta} \|\Phi - \tilde{\Phi}\|_{\gamma,\eta;\epsilon} (1 + \|\Pi\|_{\epsilon}) (1 + \|\Phi\|_{\gamma,\eta;\epsilon} + \|\tilde{\Phi}\|_{\gamma,\eta;\epsilon})^{m-1}. 
\end{align*}
Again, if we restrict ourselves to centered balls with radius $\Lambda$ in $\dD^{\gamma,\eta}_{\epsilon}$, then as soon as we choose
\begin{equation} \label{eq:short_existence_time}
(T+\epsilon)^{\theta} < \frac{1}{C(1 + \|\Pi\|_{\epsilon}) (1 + 2 \Lambda)^{m-1}}, 
\end{equation}
the map $\mM^{(\epsilon)} = \mM^{(\epsilon)}_{T}$ is a contraction and there is a unique solution to \eqref{eq:fixed_pt}. This is possible for all small $\epsilon$. In addition, it is clear that if the coefficients $\lambda_{j}^{(\epsilon)}$'s, the norms $\|\Pi\|_{\epsilon}$ and $\|u_{0}^{(\epsilon)}\|_{\gamma,\eta;\epsilon}$ are uniformly bounded as $\epsilon \rightarrow 0$, then this short existence time $T$ could be chosen independent of $\epsilon$ provided $\epsilon$ is small enough. 

We now turn to the second part of the theorem, namely the convergence of local solutions $\Phi^{(\epsilon)}$ to $\Phi$ up to the time $T$ when $\Phi$ is defined. By the arguments above, there exists a time $S < T$ such that \eqref{eq:fixed_pt} has a fixed point solution $\Phi^{(\epsilon)}$ in $\dD^{\gamma,\eta}_{\epsilon}$ up to time $S$ for all small $\epsilon$. We first show the convergence of $\Phi^{(\epsilon)}$ to $\Phi$ up to time $S$, and iterate the relative bounds to get existence and convergence to time $T$. 

Let $\mM^{(\epsilon)}: \dD^{\gamma,\eta}_{\epsilon} \rightarrow \dD^{\gamma,\eta}_{\epsilon}$ denote the map
\begin{align*}
\mM: \Phi \mapsto \pP \1_{+} \bigg(  \Xi - \sum_{j=4}^{m} \lambda_{j}^{(\epsilon)} \qQ_{\leq 0} \heE^{\frac{j-3}{2}} \big(\qQ_{\leq 0}(\Phi^{j}) \big) - \sum_{j=0}^{3} \lambda_{j}^{(\epsilon)} \qQ_{\leq 0} (\Phi^{j}) \bigg) + \widehat{P} \phi_{0}^{(\epsilon)}. 
\end{align*}
up to time $S$, and $\mM: \dD^{\gamma,\eta} \rightarrow \dD^{\gamma,\eta}$ be the map of the same form except that $\lambda_{j}^{(\epsilon)}$ and $\phi_{0}^{(\epsilon)}$ are replaced by $\lambda_{j}$ and $\phi_{0}$. Following the same line of argument as in the proof for the first half, we have
\begin{align*}
&\phantom{1111} \| \mM^{(\epsilon)}(\Phi^{(\epsilon)}); \mM(\Phi) \|_{\gamma,\eta;\epsilon} \\
&\lesssim (S + \epsilon)^{\theta} \| \Phi^{(\epsilon)}; \Phi \|_{\gamma,\eta;\epsilon} + \sup_{j} |\lambda_{j}^{(\epsilon)} - \lambda_{j}| + \$ \Pi^{(\epsilon)}; \Pi \$_{\epsilon,0} + \| \phi_{0}^{(\epsilon)}; \phi_{0} \|_{\gamma, \eta;\epsilon} , 
\end{align*}
where the proportionality constant depends on the norm of the relevant models, the size of the ball in $\dD^{\gamma,\eta}_{\epsilon}$, the initial data and the coefficients. For small enough $S$, using the knowledge that $\Phi^{(\epsilon)}$ and $\Phi$ are the fixed points in $\dD^{\gamma,\eta}_{\epsilon}$ and $\dD^{\gamma,\eta}_{0}$ respectively, we easily get
\begin{align} \label{eq:iterate_contraction}
\| \Phi^{(\epsilon)}; \Phi \|_{\gamma,\eta;\epsilon} \lesssim \sup_{j} |\lambda_{j}^{(\epsilon)} - \lambda_{j}| + \$ \Pi^{(\epsilon)}; \Pi \$_{\epsilon,0} + \| u_{0}^{(\epsilon)}; u_{0} \|_{\gamma,\eta;\epsilon}. 
\end{align}
This gives the desired convergence of $\| \Phi^{(\epsilon)}; \Phi \|_{\gamma,\eta;\epsilon}$ to $0$ up to time $S$. We now need to extend the solutions to time $T$, up to when the solution $\Phi$ to \eqref{eq:fixed_pt} is defined with model $\Pi \in \sM_{0}$. It suffices to have bounds for $\rR^{(\epsilon)} \Phi^{(\epsilon)} (t, \cdot)$ and $\rR^{(\epsilon)} \Phi^{(\epsilon)} (t, \cdot) - (\rR \Phi)(t, \cdot)$ in $\cC^{\gamma, \eta}_{\epsilon}$ so that we can restart the solution from time $t$. In fact, these are precisely what we obtained in Proposition~\ref{pr:restart_solution}. Thus, one could iterate \eqref{eq:iterate_contraction} up to time $T$, and this completes the proof. 
\end{proof}

\subsection{Renormalised equation}

We now turn to studying the effect of the renormalisation maps defined in Section~\ref{sec:renormalisation} on the solutions to the fixed point problem \eqref{eq:fixed_pt}. For simplicity, we write
\begin{align*}
\fF := \sum_{j=3}^{m} \lambda_{j} \eE^{\frac{j-3}{2}} \Psi^{j}, 
\end{align*}
and, for $n \geq 1$, we define the $n$-th `derivative' of $\fF$ to be
\begin{align*}
\fF^{(n)} := \sum_{j=3}^{m} j(j-1) \cdots (j-n+1) \lambda_{j} \eE^{\frac{j-3}{2}} \Psi^{j-n}. 
\end{align*}
If $(\bar{\Pi}, \bar{f})$ is an admissible model and $\gamma \in (1,\frac{6}{5})$, then the solution to the fixed point problem \eqref{eq:fixed_pt} in $\dD^{\gamma,\eta}_{\epsilon}$ necessarily has the form
\begin{equation} \label{eq:abstract_solution}
\Phi = \Psi + \varphi \cdot \1 - \iI(\fF) - \lambda_{2} \iI(\Psi^{2}) - \varphi \cdot \iI(\fF') + \varphi' \cdot X = \Psi + U, 
\end{equation}
where $U$ denotes the part that contains all basis vectors except $\Psi$. Therefore, the right hand side of \eqref{eq:fixed_pt} (including all terms with homogeneities up to $0$) can be written as
\begin{equs} \label{eq:rhs}
\hH(z) &= \Xi - \sum_{j=4}^{m} \lambda_{j} \heE^{\frac{j-3}{2}} \Phi^{j} - \sum_{j=0}^{3} \lambda_{j} \Phi^{j} \\
&= \Xi - \fF - \lambda_{2} \Psi^{2} - \varphi \cdot \fF' - \frac{1}{2} \varphi^{2} \cdot \fF'' - (\lambda_{1} + 2 \varphi \lambda_{2}) \Psi + \fF' \iI(\fF) \\
&+ \lambda_{2} \fF' \iI(\Psi^{2}) + \varphi \cdot \fF' \iI(\fF') + \varphi \cdot \fF'' \iI(\fF) + 2 \lambda_{2} \Psi \iI(\fF) - \varphi' \fF'X - \frac{1}{6} \varphi^{3} \cdot \fF''' \\
&+ \bigg( \lambda_{0} + \lambda_{1} \varphi + \lambda_{2} \varphi^{2} + \sum_{j \geq 4} \sum_{n=4}^{j} \begin{pmatrix} j \\ n \end{pmatrix} \varphi^{n} \overline{f}_{z}\big( \sE^{\frac{j-3}{2}}_{0} (\Psi^{j-n} U(z)^{n})  \big) \bigg)  \cdot \1. 
\end{equs}
We then have the following theorem, the proof of which 
is essentially the same as that in \cite[Sec.~5]{HQ15}, so we 
omit the details here.


\begin{thm} \label{th:renormalised_equation}
	Let $\phi_{0} \in \cC^{1}$, $\epsilon \geq 0$, and $\widehat{\xi}$ be a smooth space-time function. Let $(\Pi, f) = \sL_{\epsilon}(\widehat{\xi})$ be the canonical model as in Section 2, $M = (M_{0}, \wM)$, and $(\Pi^{M}, f^{M}) = M \sL_{\epsilon}(\widehat{\xi})$ be the renormalised model described in Section~\ref{sec:renormalisation}. If $\Phi \in \dD^{\gamma,\eta}_{\epsilon}$ is the local solution to the fixed point problem \eqref{eq:fixed_pt} for the model $(\Pi^{M}, f^{M})$, then the function $u = \rR^{M} \Phi$ is the classical solution to the PDE
	\begin{equation} \label{eq:renormalised_equation}
	\partial_{t} u = \Delta u - \sum_{j=4}^{m} \lambda_{j} \epsilon^{\frac{j-3}{2}} H_{j}(u;C_{1}) - \sum_{j=0}^{3} \lambda_{j} H_{j}(u; C_{1}) - (C u + C' + 6 \lambda_{2} \lambda_{3} C_{2}) + \widehat{\xi}
	\end{equation}
	with initial data $\phi_{0}$, and the constants $C$ and $C'$ are given by
	\begin{equation} \label{eq:renormalisation_constant}
	\begin{split}
	C &= \sum_{n=2}^{m-1} (n+1)^{2} n! \cdot \lambda_{n+1}^{2} C_{n} + \sum_{n=3}^{m-2} (n+2)! \cdot \lambda_{n} \lambda_{n+2} C_{n}\;, \\
	C' &= \sum_{n=3}^{m-1} (n+1)! \cdot \lambda_{n} \lambda_{n+1} C_{n}'\;. 
	\end{split}
	\end{equation}
\end{thm}

\section{Convergence of the renormalised models}
\label{sec:convModels}

In this section, we will show how to choose the correct constants so that the action of the 
renormalisation maps built in Section~\ref{sec:renormalisation} on the canonical model 
yields convergence to a limit, and we will also identify the limiting model. The 
identification of the limiting equation will be given in Section~\ref{sec:limits}.

\subsection{Main statement and convergence criterion} \label{sec:main_statement}

Let $\xi$ denote space-time white noise on $\R\times \T^3$. Fix a smooth compactly supported function $\rho: \R^{1+3} \rightarrow \R$ integrating to $1$, and set
\begin{equation} \label{eq:convolution_noise}
\rho_{\epsilon}(t,x) = \epsilon^{-5} \rho(t/\epsilon^{2}, x/\epsilon), \qquad \xi_{\epsilon} = \rho_{\epsilon} * \xi, 
\end{equation}
where `$*$' denotes space-time convolution. Then, the correlation of $\xi_{\epsilon}$ is
\begin{align*}
\E \xi_{\epsilon}(s,x) \xi_{\epsilon}(t,y) = \int_{\R} \int_{\TT^{3}} \rho_{\epsilon}(s-u,x-z) \rho_{\epsilon}(t-u,y-z) dz du. 
\end{align*}
If the noise $\widehat{\xi}$ is obtained from the convolution of the space time white noise defined on $\R \times (\epsilon^{-1} \T)^{3}$ with the mollifier $\rho$, then we actually have
\begin{align*}
\xi_{\epsilon} (t,x) \stackrel{\text{law}}{=} \epsilon^{-\frac{5}{2}} \widehat{\xi} (t/\epsilon^{2}, x/\epsilon). 
\end{align*}
From now on, we will always assume that the noise $\xi_{\epsilon}$ relates to $\xi$ by \eqref{eq:convolution_noise}. When we consider scale $\alpha < 1$ later, we simply replace $\epsilon$ by $\epsilon^{\alpha}$ in that expression. 
We also let
\begin{align*}
K_{\epsilon} = K * \rho_{\epsilon}, \qquad G_{\epsilon} = K_{\epsilon} * K_{\epsilon}, 
\end{align*}
where $K$ coincides the heat kernel in $\{|z| < 1\}$, has compact support, and annihilates polynomials up to degree $3$, as introduced at the beginning of Section~\ref{sec:admissible_model}. A crucial ingredient in proving the correct behaviour of various stochastic objects arising from the equation are the following bounds for the kernels $K_{\epsilon}$ and $G_{\epsilon}$. The proof can be found in \cite[Lemma $10.14$]{Hai14a}.

\begin{prop} \label{pr:kernel_convolution}
We have
\begin{equ}
D^{\ell} K_{\epsilon}(z) \lesssim (|z| + \epsilon)^{-3 - |\ell|}, \qquad G_{\epsilon}(z) \lesssim (|z|+\epsilon)^{-1}, 
\end{equ}
uniformly over all $\epsilon < 1$ and space-time points $z$ with $|z| < 1$. 
\end{prop}

\begin{rmk} \label{rm:parabolic_degree}
Here, $\ell = (\ell_{0}, \ell_{1}, \ell_{2}, \ell_{3})$ is a multi-index, and $|\ell| = 2 \ell_{0} + \sum_{i=1}^{3} \ell_{i}$ reflects the parabolic scaling. In what follows, we will always use the notation $|\cdot|$ to denote the parabolic degree of such indices. Also note that we do not require bounds on the derivatives of $G_{\epsilon}$, since none of the appearances of this kernel carries any renormalisation (in the sense that will become clear later). 
\end{rmk}

The main theorem of this section is the following.

\begin{thm} \label{th:main_convergence}
Let $M_{\epsilon} \in \fR$ denote the renormalisation map
\begin{equs}
M_{\epsilon} =\Big( \exp \big( - \sum_{n \geq 2} C_{n}^{(\epsilon)} L_{n} - \sum_{n \geq 3} C_{n}'^{(\epsilon)} L_{n}'  \big), \phantom{1} \exp \big(-C_{1}^{(\epsilon)} L_{1} \big) \Big),
\end{equs}
with $L_{n}$ and $L_{n}'$ as in Section~\ref{sec:renormalisation}. Let $\sL_{\epsilon}(\xi_{\epsilon})$ be canonical lift of $\xi_\eps$ to the regularity structure $\tT$ 
as in Section~\ref{sec:canonical_lift}, and consider the sequence of models
\begin{equ}
\fM_{\epsilon} := M_{\epsilon} \sL_{\epsilon}(\xi_{\epsilon})\;. 
\end{equ}
Then, there exists a choice of constants $C_{n}^{(\epsilon)}$, $C_{n}'^{(\epsilon)}$, and a random model $\fM \in \sM_{0}$ such that
\begin{equ}
\$ \fM_{\epsilon}; \fM \$_{\epsilon,0} \rightarrow 0
\end{equ}
in probability as $\epsilon \rightarrow 0$. Furthermore, the limiting model $\fM = (\hPi, \widehat{f})$ satisfies $\hPi_{z} \tau = 0$ for every $z$ and every basis vector $\tau$ that contains an occurrence of $\eE^\beta$ for some $\beta > 0$. 
\end{thm}

The readers may have already realised that with proper choices of $C_{1}^{(\epsilon)}$ and $C_{2}^{(\epsilon)}$, the action of the model $\fM_{\epsilon}$ on basis vectors without an appearance of $\eE$ is exactly as those in the regularised $\Phi^4_3$ equation (see \cite[Section 10]{Hai14a} for details). Thus, the action of the limiting model $\fM$ on those basis vectors is precisely the same as that of the limiting $\Phi^4_3$ model.

However, the effect of the models $\fM_{\epsilon}$ on symbols that contain $\eE$'s is more complicated. In order to prove Theorem~\ref{th:main_convergence}, we first give a useful criterion for the convergence of models in $\sM_{\epsilon}$. The proof of this criterion is essentially the same as Propositions 6.2 and 6.3 in \cite{HQ15}, so we only give the statement without proofs.

\begin{prop} \label{pr:convergence_criterion}
Let $(\tT, \gG)$ be the regularity structure given in Section~\ref{sec:structure}, and consider a family of random models $(\widehat{\Pi}^{\epsilon}, \widehat{f}^{\epsilon})$ in $\sM_{\epsilon}$. Assume there exists $\theta > 0$ such that for every test function $\varphi \in \bB$, every $\tau \in \wW$ with $|\tau| < 0$, every space-time point $z$ and every $\lambda \in (0,1)$, there exists a random variable $(\hPi_{z} \tau)(\varphi_{z}^{\lambda})$ such that
\begin{equation} \label{eq:negative_bound}
\E |(\hPi_{z}^{\epsilon} \tau)(\varphi_{z}^{\lambda})|^{2} \lesssim \lambda^{2|\tau| + \theta}, \qquad \E |(\hPi_{z}^{\epsilon} \tau - \hPi_{z} \tau)(\varphi_{z}^{\lambda})|^{2} \lesssim \epsilon^{\theta} \lambda^{2 |\tau| + \theta}. 
\end{equation}
Assume furthermore that for every $\eE^{\beta}(\tau) \in \wW$ with $\beta + |\tau| > 0$, one has
\begin{equation} \label{eq:group_bound}
\E |D^{\ell} \widehat{f}_{z}^{\epsilon}(\sE^{\beta}_{0} \tau)| \lesssim \epsilon^{|\tau| + \beta - |\ell| + \theta}
\end{equation}
for some positive $\theta$, and that for any $\tau \in \uU$, one has the bound
\begin{equation} \label{eq:positive_bound}
\E |(\hPi_{z}^{\epsilon} \tau)(\psi_{z}^{\lambda})| \lesssim \lambda^{\zeta + \theta} \epsilon^{|\tau|-\zeta}, \quad \zeta = \frac{6}{5}, 
\end{equation}
for all test functions $\psi \in \bB$ that annihilate affine functions, uniformly over $\lambda \in (0, \epsilon)$. Then, there exists a random model $(\hPi, \widehat{f}) \in \sM_{0}$ such that $\$ \hPi^{\epsilon}, \hPi \$_{\epsilon;0} \rightarrow 0$ in probability as $\epsilon \rightarrow 0$. 
\end{prop}

\begin{rmk}
Later, we will consider $(\widehat{\Pi}^{\epsilon}, \widehat{f}^{\epsilon}) = M_{\epsilon} \sL_{\epsilon} (\xi_{\epsilon})$ as in Theorem~\ref{th:main_convergence} with proper renormalisation constants $C_{j}^{(\epsilon)}$'s defined in the next subsection. It is straightforward to see that they indeed belong to $\sM_{\epsilon}$. For the limiting model $\fM$, its action on basis vectors without any appearance of $\eE$ is exactly the same as in the standard $\Phi^4_3$ equation (in fact, these are precisely the terms that appears in $\Phi^4_3$). Its action on terms containing a factor of $\eE^{\beta}$ will yield $0$. Thus, in addition to \eqref{eq:group_bound}, \eqref{eq:positive_bound}, it suffices to prove the second bound in \eqref{eq:negative_bound} for $\tau$ containing 
at least one factor of $\eE$, and with $\widehat{\Pi}_{z} \tau = 0$. 
\end{rmk}

\subsection{Graphical notations and preliminary bounds}

The remainder of this section is devoted to the proof that the random models
$M_\eps \sL_\eps(\xi_\eps)$ as in Theorem~\ref{th:main_convergence} do indeed satisfy the convergence
criterion of Proposition~\ref{pr:convergence_criterion}.
Since we are in a translation invariant setting, it suffices bound the random variables $(\hPi_{0}^{\epsilon} \tau)(\varphi_{0}^{\lambda})$ for various basis vectors $\tau$. 
All these random variables belong to some finite order Wiener chaos. Following \cite{HP14,HQ15}, we use a graphical notation to represent the kernels 
for homogeneous Wiener chaos of finite order. Each node in the graph represents a 
space-time variable in $\R^{1+3}$: the special green node \tikz[baseline=-3] \node [root] {}; 
represents the origin, which is fixed, 
the nodes \tikz[baseline=-3] \node [var] {}; represent the 
arguments in the kernel representation for homogeneous Wiener chaos, and the remaining nodes \tikz[baseline=-3] \node [dot] {}; represent variables to be integrated out. 

Each plain arrow 
\tikz[baseline=-0.1cm] \draw[kernel] (0,0) to (1,0); 
represents the kernel $K(z'-z)$, where $z$ and $z'$ are starting and ending points of the arrow. A dotted arrow
\tikz[baseline=-0.1cm] \draw[kepsilon] (0,0) to (1,0); 
represents the kernel $K_{\epsilon}$ with the same orientation as before, and
a bold green arrow \tikz[baseline=-0.1cm] \draw[testfcn] (1,0) to (0,0); 
represents a generic test function in $\bB$ rescaled by a factor $\lambda$. In addition, we use the barred arrow
\tikz[baseline=-0.1cm] \draw[kernel1] (0,0) to (1,0); 
to represent a factor $K(z'-z) - K(-z)$, where as before $z$ and $z'$ denote starting and ending points of the arrow. Finally, a double barred arrow \tikz[baseline=-0.1cm] \draw[kernel2] (0,0) to (1,0); represents the factor $K(z'-z) - K(-z) - x' \cdot DK(-z)$, where $z=(t,x)$, $z'=(t',x')$, and the differentiation $DK$ is with respect to space variable only. 

With these notations, it follows for example that for $\tau = \Psi \iI(\Psi^{3})$ and the canonical model $\Pi^{\epsilon} = \sL_\eps(\xi_\eps)$, we have the expression
\begin{equation} \label{eq:graph_expression}
(\Pi_{0}^{\epsilon} \Psi \iI(\Psi^{3}))(\varphi_{0}^{\lambda}) = \phantom{1}
\begin{tikzpicture}[scale=0.6,baseline=-0.0cm]
\node at (-1,1) [var] (aboveleft) {}; 
\node at (1,1) [var] (aboveright) {}; 
\node at (0,1.2) [var] (above) {}; 
\node at (0,0) [dot] (middle) {}; 
\node at (0,-1) [dot] (below) {}; 
\node at (1,0) [var] (right) {}; 
\node at (-1.2,-1.5) [root] (farbelow) {}; 
\draw[kepsilon] (aboveleft) to (middle); 
\draw[kepsilon] (aboveright) to (middle); 
\draw[kepsilon] (above) to (middle);
\draw[kepsilon] (right) to (below);  
\draw[kernel1] (middle) to (below); 
\draw[testfcn] (below) to (farbelow); 
\end{tikzpicture}
\phantom{1} + \phantom{1} 3 \phantom{1}
\begin{tikzpicture}[scale=0.6,baseline=-0.0cm]
\node at (-1,1) [var] (aboveleft) {}; 
\node at (1,1) [var] (aboveright) {}; 
\node at (0,0) [dot] (middle) {}; 
\node at (0,-1) [dot] (below) {}; 
\node at (1,-0.5) [dot] (right) {}; 
\node at (-1.2,-1.5) [root] (farbelow) {}; 
\draw[kepsilon] (aboveleft) to (middle); 
\draw[kepsilon] (aboveright) to (middle); 
\draw[kepsilon] (right) to (middle); 
\draw[kepsilon] (right) to (below);  
\draw[kernel1] (middle) to (below); 
\draw[testfcn] (below) to (farbelow); 
\end{tikzpicture}
\phantom{1} + \phantom{1} 3 \phantom{1}
\begin{tikzpicture}[scale=0.6,baseline=-0.0cm]
\node at (-1,1) [dot] (aboveleft) {}; 
\node at (1,1) [var] (aboveright) {}; 
\node at (0,0) [dot] (middle) {}; 
\node at (0,-1) [dot] (below) {}; 
\node at (1,0) [var] (right) {}; 
\node at (-1.2,-1.5) [root] (farbelow) {}; 
\draw[kepsilon, bend left = 60] (aboveleft) to (middle); 
\draw[kepsilon, bend right = 60] (aboveleft) to (middle); 
\draw[kepsilon] (aboveright) to (middle); 
\draw[kepsilon] (right) to (below);  
\draw[kernel1] (middle) to (below); 
\draw[testfcn] (below) to (farbelow); 
\end{tikzpicture}
\phantom{1} + \phantom{1} 3
\begin{tikzpicture}[scale=0.6,baseline=-0.0cm]
\node at (0,1.2) [dot] (above) {}; 
\node at (0,0) [dot] (middle) {}; 
\node at (1,-0.5) [dot] (right) {}; 
\node at (0,-1) [dot] (below) {}; 
\node at (-1.2,-1.5) [root] (farbelow) {}; 
\draw[kepsilon, bend left = 60] (above) to (middle); 
\draw[kepsilon, bend right = 60] (above) to (middle); 
\draw[kepsilon] (right) to (middle); 
\draw[kepsilon] (right) to (below);  
\draw[kernel1] (middle) to (below); 
\draw[testfcn] (below) to (farbelow); 
\end{tikzpicture}
\phantom{1}\;.
\end{equation}
Here, the first term represents the component in the fourth homogeneous Wiener chaos (see \cite[Ch.1.1.2]{Nua06}),
the next two terms represent the component in the second homogeneous chaos, and the last term
is the component in the zeroth homogeneous chaos. 
The variance of the first two terms above, for example, are bounded (up to some constant multiple) by
\begin{equation} \label{eq:graph_example}
\begin{tikzpicture}[scale=0.7,baseline=0.4cm]
\node at (0,2) [dot] (farabove) {}; 
\node at (0,1.5) [dot] (above) {}; 
\node at (0,1) [dot] (nearabove) {}; 
\node at (-1,1.5) [dot] (aboveleft) {}; 
\node at (1,1.5) [dot] (aboveright) {}; 
\node at (-1,0.5) [dot] (left) {}; 
\node at (0,0.5) [dot] (middle) {}; 
\node at (1,0.5) [dot] (right) {}; 
\node at (0,-0.5) [root] (below) {}; 
\draw[kepsilon] (farabove) to (aboveleft); 
\draw[kepsilon] (above) to (aboveleft); 
\draw[kepsilon] (nearabove) to (aboveleft); 
\draw[kepsilon] (farabove) to (aboveright); 
\draw[kepsilon] (above) to (aboveright); 
\draw[kepsilon] (nearabove) to (aboveright); 
\draw[kernel1] (aboveleft) to (left); 
\draw[kernel1] (aboveright) to (right); 
\draw[kepsilon] (middle) to (left); 
\draw[kepsilon] (middle) to (right); 
\draw[testfcn] (left) to (below); 
\draw[testfcn] (right) to (below); 
\end{tikzpicture}
\phantom{1} + \phantom{1}
\begin{tikzpicture}[scale=0.7,baseline=-0.2cm]
\node at (0,1.2) [dot] (above) {}; 
\node at (0,-0.2) [dot] (middle) {}; 
\node at (0,-1.5) [root] (below) {}; 
\node at (-1,0.5) [dot] (aboveleft) {}; 
\node at (1,0.5) [dot] (aboveright) {}; 
\node at (-1.8,0) [dot] (left) {}; 
\node at (1.8,0) [dot] (right) {}; 
\node at (-1,-0.5) [dot] (belowleft) {}; 
\node at (1,-0.5) [dot] (belowright) {}; 
\draw[kepsilon] (above) to (aboveleft); 
\draw[kepsilon] (above) to (aboveright); 
\draw[kepsilon] (middle) to (aboveleft); 
\draw[kepsilon] (middle) to (aboveright); 
\draw[kepsilon] (left) to (aboveleft); 
\draw[kepsilon] (left) to (belowleft); 
\draw[kepsilon] (right) to (aboveright); 
\draw[kepsilon] (right) to (belowright); 
\draw[kernel1] (aboveleft) to (belowleft); 
\draw[kernel1] (aboveright) to (belowright); 
\draw[testfcn] (belowleft) to (below); 
\draw[testfcn] (belowright) to (below); 
\end{tikzpicture}
\phantom{1}. 
\end{equation}
To bound this and similar quantities, it is convenient to label the edges of the graph 
to reflect the singularity of the corresponding kernel, and to give a bound of the whole 
object in terms of simple power counting of the labels. For this purpose, and in order to be able to
use the bounds obtained in \cite{HQ15}, we 
introduce \textit{labelled graphs} to represent bounds for quantities like 
$\E |(\hPi_{0}^{\epsilon} \tau)(\varphi_{0}^{\lambda})|^{2}$. 

Let $\gG = (\vV, \eE)$ be a labelled graph, where $\vV$ is the set of vertices, and $\eE$ is the set edges (which are labelled and directed). More precisely, each edge $e = (x_{v_{-}}, x_{v_{+}})$ in the graph has the direction $x_{v_{-}} \leftarrow x_{v_{+}}$, and is associated with a pair of numbers $(a_{e}, r_{e}) \in \R^{+} \times \ZZ$, and the orientation of the edge really matters only if $r_{e} > 0$. As before, edges $e$ are associated to kernels $J_e$, with $a_e$ measuring the singularity of the kernel in question in the sense that we assume that each $J_e$ is compactly supported and satisfies a bound of the type
\begin{equ}[e:boundJe]
|D^k J_e(z)| \lesssim |z|^{-a_e - |k|}\;,
\end{equ}
for every multiindex $k$. The precise factor represented by each edge then furthermore depends
on the value $r_e$.
If $r_{e} = 0$, then the corresponding edge simply represents a factor 
$\widehat{J}_{e}(x_{v_{-}}, x_{v_{+}}) = J_{e}(x_{v_{+}} - x_{v_{-}})$. 
We simply write $a_{e}$ instead of $(a_{e},0)$ in this case. 

If $r_{e} > 0$, then the corresponding edge represents a factor
\begin{equation} \label{eq:positive_renormalisation}
\widehat{J}_{e}(x_{v_{-}}, x_{v_{+}}) = J_{e}(x_{v_{+}}- x_{v_{-}}) - \sum_{|k|_{\fs} < |r_e|} \frac{x_{v_{+}}^{k}}{k!} D^{k}J_{e}(-x_{v_{-}}). 
\end{equation}
On the other hand, if $r_{e} < 0$, then the edge corresponds to a factor
$\widehat{J}_{e}(x_{v_{-}}, x_{v_{+}}) = (\sR J_{e})(x_{v_{+}}- x_{v_{-}})$, where
 $\sR J_{e}$ denotes the \textit{renormalised} distribution
\begin{equation} \label{eq:negative_renormalisation}
(\sR J_{e})(\varphi) = \int J_{e}(x) \Big( \varphi(x) - \sum_{|k|_{\fs} < |r_{e}|} \frac{x^{k}}{k!} \varphi(0) \Big)\, dx\;. 
\end{equation}
In other words, a positive $r_{e}$ corresponds to the re-centering from subtracting lower order Taylor polynomials (or maybe called ``positive renormalisation"), and $r_{e} < 0$ corresponds to ``negative renormalisation". Also, since we will always consider situations where no two edges with $r_{e} < 0$ meet and
all $J_e$ are smooth functions, the meaning of the factor $(\sR J_{e})(x_{v_{+}}- x_{v_{-}})$ is
unambiguous.

Unlike in \cite{HQ15}, each labelled graph does in our case represent a sequence of multiple integrals
depending on a parameter $\epsilon \in (0,1]$. To keep track of some of that dependency,
we consider graphs with both `plain' and `dotted' edges.
If an edge is plain, then the corresponding kernel $J_e$ is allowed to depend on $\eps$
(to make that dependency clear we will also sometimes write $J_e^{(\eps)}$), but the bounds
\eqref{e:boundJe} are assumed to hold uniformly in $\eps \in (0,1]$. 
If an edge is dotted however, then the corresponding kernel $J_{e}^{(\eps)}$ is assumed to satisfy the bound
\begin{equ}
|D^k J_e^{(\eps)}(z)| \lesssim (|z| + \eps)^{-a_e - |k|}\;,
\end{equ}
uniformly in $\eps \in(0,1]$. 
There are two additional edges (in boldface) connecting to the origin that represent a factor $\varphi^{\lambda}(x_{v},0)$. The origin is denoted by $\{0\} \subset \vV$, and we denote by $v_{\star, 1}$ and $v_{\star, 2}$ the two vertices that connect to $0$ by the edges representing test functions. Finally, we set
\begin{align*}
\vV_{\star} = \{0, v_{\star,1}, v_{\star,2}\}, \qquad \vV_{0} = \vV \setminus \{0\}. 
\end{align*}
Thus, as a consequence of Proposition~\ref{pr:kernel_convolution}, 
the quantity in \eqref{eq:graph_example} can be represented by
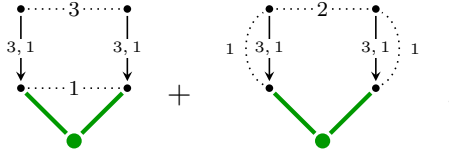
\begin{align*}
\begin{tikzpicture}[scale=0.7,baseline=-0.2cm]
\node at (-1,1.5) [dot] (aboveleft) {}; 
\node at (1,1.5) [dot] (aboveright) {}; 
\node at (-1,0) [dot] (left) {}; 
\node at (1,0) [dot] (right) {}; 
\node at (0,-1) [root] (below) {}; 
\draw[gepsilon] (aboveleft) to node[labl]{\scriptsize $3$} (aboveright); 
\draw[gepsilon] (left) to node[labl]{\scriptsize $1$} (right); 
\draw[kernel1] (aboveleft) to node[labl]{\tiny $3,1$} (left); 
\draw[kernel1] (aboveright) to node[labl]{\tiny $3,1$} (right); 
\draw[dist] (left) to (below); 
\draw[dist] (right) to (below); 
\end{tikzpicture}
\phantom{1} + \phantom{1}
\begin{tikzpicture}[scale=0.7,baseline=-0.2cm]
\node at (-1,1.5) [dot] (aboveleft) {}; 
\node at (1,1.5) [dot] (aboveright) {}; 
\node at (-1,0) [dot] (left) {}; 
\node at (1,0) [dot] (right) {}; 
\node at (0,-1) [root] (below) {}; 
\draw[gepsilon] (aboveleft) to node[labl]{\scriptsize $2$} (aboveright); 
\draw[kernel1] (aboveleft) to node[labl]{\tiny $3,1$} (left); 
\draw[kernel1] (aboveright) to node[labl]{\tiny $3,1$} (right); 
\draw[gepsilon, bend left = 60] (left) to node[anchor=east]{\tiny $1$} (aboveleft); 
\draw[gepsilon, bend right = 60] (right) to node[anchor=west]{\tiny $1$} (aboveright); 
\draw[dist] (left) to (below); 
\draw[dist] (right) to (below); 
\end{tikzpicture}
\phantom{1}. 
\end{align*}
With all these notations at hand, for a labelled graph $\gG$ and the collection of kernels $J_{e}$, we let $I_{\lambda}^{\gG}$ denote the number
\begin{equation} \label{eq:graph_quantity}
I_{\lambda}^{\gG} = \int_{(\R^{4})^{\vV_{0}}} \prod_{e \in \eE} \widehat{J}_{e} (x_{e_{-}}, x_{e_{+}}) dx, 
\end{equation}
where $4$ reflects the space-time dimension.
In order to determine the right scale of the quantity $I_{\lambda}^{\gG}$, we introduce some additional notations. For any subset $\bar{\vV} \subset \vV$, we let
\begin{align*}
\eE^{\uparrow}(\bar{\vV}) &= \{e \in \eE: e \cap \bar{\vV}=e_{-}, r_{e} > 0 \}; \\
\eE^{\downarrow}(\bar{\vV}) &= \{e \in \eE: e \cap \bar{\vV}=e_{+}, r_{e} > 0 \}; \\
\eE_{0}(\bar{\vV}) &= \{e \in \eE: e \cap \bar{\vV }= e \}; \\
\eE(\bar{\vV}) &= \{e \in \eE: e \cap \bar{\vV} \neq \phi \}. 
\end{align*}
In other words, $\eE^{\uparrow}(\bar{\vV})$ is the set of outgoing edges from $\bar{\vV}$ with $r_{e} > 0$, $\eE^{\downarrow}(\bar{\vV})$ is the set of incoming edges to $\bar{\vV}$ with $r_{e} > 0$, $\eE_{0}(\bar{\vV})$ is the set of edges with both vertices in $\bar{\vV}$, and $\eE(\bar{\vV})$ is the set of edges with at least one vertex in $\bar{\vV}$. Note that the definition of $\eE^{\uparrow}(\bar{\vV})$ and $\eE^{\downarrow}(\bar{\vV})$ only considers edges with $r_{e} > 0$. 

Now, consider a labelled graph $\gG = (\vV, \eE)$ satisfying the following properties.

\begin{assumption} \label{as:graph_assump}
The labelled graph $\gG = (\vV, \eE)$ satisfies the following properties. 
\begin{enumerate}
\item For every edge $e \in \eE$, one has $a_{e} + (r_{e} \wedge 0) < 5$; 

\item For every subset $\bar\vV \subset \vV$ of cardinality at least $3$, one has
\begin{align*}
\sum_{e \in \eE_{0}(\bar{\vV})} a_{e} < 5(|\bar{\vV}|-1); 
\end{align*}
\item For every subset $\bar{\vV} \subset \vV$ containing $0$ and of cardinality at least $2$, one has
\begin{align*}
\sum_{e \in \eE_{0}(\bar{\vV})} a_{e} + \sum_{e \in \eE^{\uparrow}(\bar{\vV})} (a_{e} + r_{e} - 1) - \sum_{e \in \eE^{\downarrow}(\bar{\vV})} r_{e} < 5 (|\bar{\vV}|-1); 
\end{align*}
\item For every non-empty subset $\bar{\vV} \subset \vV \setminus \vV_{\star}$, one has
\begin{align*}
\sum_{e \in \eE(\bar{\vV}) \setminus \eE^{\downarrow}(\bar{\vV})} a_{e} + \sum_{e \in \eE^{\uparrow}(\bar{\vV})} - \sum_{e \in \eE^{\downarrow}(\bar{\vV})} (r_{e}-1) > 5 |\vV|. 
\end{align*}
\end{enumerate}
Note that the number $5$ in the above assumptions indicates the parabolic degree of the space-time dimension is $5$. 
\end{assumption}

It turns out that this assumption on the graph $\gG$ is sufficient to guarantee that the quantity $I_{\lambda}^{\gG}$ has the correct scaling behavior for small $\lambda$. This is the content of the following theorem, proved in \cite{HQ15}.

\begin{thm}
Let $\gG$ be a graph that satisfies Assumption~\ref{as:graph_assump}, and its edges represent kernels that satisfy the definitions and bounds in \eqref{e:boundJe}, \eqref{eq:positive_renormalisation} and \eqref{eq:negative_renormalisation}. If $I^{\gG}_{\lambda}$ denotes the quantity defined in \eqref{eq:graph_quantity}, then one has
\begin{equation} \label{eq:general_bound}
I_{\lambda}^{\gG} \lesssim \lambda^{\alpha}
\end{equation}
uniformly over $\lambda \in (0,1)$, where $\alpha = 5 |\vV \setminus \vV_{\star}| - \sum_{e \in \eE} a_{e}$, and the proportionality constant depends on the graph and magnitudes of norms of the corresponding kernels. 
\end{thm}

\begin{rmk}
The proportionality constant in \eqref{eq:general_bound} is a constant multiple of $\prod_{e} \| \widehat{J}_{e} \|_{a_{e},p_{e}}$ for suitable values $p_e$ depending on the
structure of the graph, where
\begin{align*}
\| J \|_{a,p} := \sup_{|z| \le 1, |\ell| \leq p} |z|^{a + |\ell|} |D^{\ell}J(z)|, 
\end{align*}
where we assumed that the kernels are supported in the parabolic unit ball. Since these quantities are finite, we will simply omit them in all the bounds below.
\end{rmk}

Before we prove the bounds in Proposition~\ref{pr:convergence_criterion}, we first choose values of the constants $C_{n}^{(\epsilon)}$ and $C_{n}'^{(\epsilon)}$ that appear in the statement of Theorem~\ref{th:main_convergence}. With the graphic notations, the constant $C_{1}^{(\epsilon)}$ is given by
\begin{equation} \label{eq:C1}
C_{1}^{(\epsilon)} = \int \int K_{\epsilon}^{2}(t,x) dx dt \phantom{1} = G_\eps(0) = \phantom{1}
\begin{tikzpicture}[scale=0.5,baseline=-0.1cm]
\node at (0,1) [dot] (above){};
\node at (0,-1) [root] (below) {};
\draw[kepsilon, bend left = 60] (above) to (below); 
\draw[kepsilon, bend right = 60] (above) to (below); 
\end{tikzpicture}
, \qquad C_{0}^{(\epsilon)} = \epsilon C_{1}^{(\epsilon)}. 
\end{equation}
It is easy to see that, for this definition of $C_1^{(\epsilon)}$ and the renormalised model $\widehat{\Pi}^{\epsilon}$, the expression $(\widehat{\Pi}_{0}^{\epsilon} \Psi \iI(\Psi^{3}))(\varphi_{0}^{\lambda})$ only contains the first two terms in \eqref{eq:graph_expression}, and its variance is indeed bounded by \eqref{eq:graph_example}. 

For $n \geq 2$, we define $C_{n}^{(\epsilon)}$ and $C_{n}'^{(\epsilon)}$ by
\begin{equs} \label{eq:Cn}
C_{n}^{(\epsilon)} &= \epsilon^{n-2} \int K(z) G_{\epsilon}^{n}(z) dz = \epsilon^{n-2} \quad
\begin{tikzpicture}[scale=0.6,baseline=0.2cm]
\node at (-1.5,0) [dot] (left){};
\node at (0,0.6)  [dot] (middle) {};
\node at (1.5,0)   [root] (right) {}; 
\node at (0,2.2) [dot] (above) {}; 
\node at (0,1.05) {\tiny $\vdots$}; 
\node at (0,1.9) {\tiny $\vdots$}; 
\node at (0,1.3) {\scriptsize $n$}; 
\draw[kepsilon] (middle) to (left); 
\draw[kepsilon] (middle) to (right); 
\draw[kepsilon] (above) to (left); 
\draw[kepsilon] (above) to (right); 
\draw[kernel] (left) to (right); 
\end{tikzpicture}
\phantom{1}, \quad n \geq 2\;, \\
C_{n}'^{(\epsilon)} &= \epsilon^{n-\frac{5}{2}} \int K(z) G_{\epsilon}^{n}(z) dz = \epsilon^{-\frac{1}{2}} C_{n}^{(\epsilon)}, \qquad \qquad \qquad n \geq 3\;. 
\end{equs}
It is not hard to check that
\begin{align*}
C_{1}^{(\epsilon)} &= \frac{C_{0}}{\epsilon} + \oO(1), \quad C_{0} = \int (P*\rho)^{2}(z) dz, \\
C_{2}^{(\epsilon)} &= c_{2} |\log \epsilon| + \oO(1)
\end{align*}
for some universal constant $c_2 > 0$. For $n \geq 3$, we have
\begin{align*}
C_{n}^{(\epsilon)} = C_{n} + \oO(\epsilon), \quad C_{n} = \int P(z) (P_{\rho} * P_{\rho})^{n}(z) dz, 
\end{align*}
where $P_{\rho} = P * \rho$. $C_n$ is finite for $n \geq 3$ since the integrand decays like $|z|^{-(n+3)}$ for large $z$.

\subsection{First order renormalisation bounds}

We are now ready to prove Theorem~\ref{th:main_convergence}. In view of Proposition~\ref{pr:convergence_criterion}, it suffices to check the bound \eqref{eq:negative_bound} for all terms that appear in the right hand side of \eqref{eq:rhs}, and the bounds \eqref{eq:group_bound} and \eqref{eq:positive_bound} for relevant terms with positive homogeneities. 

We first prove the bound \eqref{eq:negative_bound} for terms from $\fF^{(n)}$ for $n = 0, 1, 2, 3$. These basis vectors are of the form
\begin{align*}
\tau = \eE^{\frac{k}{2}} \Psi^{k+3-n}. 
\end{align*}
The case $k=0$ has been treated for the case of the standard $\Phi^4_3$ model in \cite[Sec.~10.5]{Hai14a}, 
so we only need to consider $k \geq 1$. For the canonical model $\Pi^{\epsilon}$, we have
\begin{equ}[e:exprProd]
\Pi_{z}^{\epsilon} \tau = \epsilon^{\frac{k}{2}} (\Pi_{z}^{\epsilon} \Psi)^{k+3-n}. 
\end{equ}
If we choose $C_{1}^{(\epsilon)}$ according to \eqref{eq:C1}, then the effect of our renormalisation
procedure is precisely to turn the products in \eqref{e:exprProd} into Wick products, so that
\begin{align*}
(\hPi_{0}^{\epsilon} \tau)(\varphi_{0}^{\lambda}) = \eps^{\frac{k}{2}} \quad
\begin{tikzpicture}[scale=0.6,baseline=-0.9cm]
\node at (-1,0) [var] (left){};
\node at (1,0)  [var] (right) {};
\node at (0,-1)   [dot] (middle) {}; 
\node at (0, -2.8)  [root] (below) {}; 
\node at (0,0) {$\cdots$}; 
\draw[kepsilon] (left) to (middle); 
\draw[kepsilon] (right) to (middle); 
\draw[testfcn] (middle) to (below); 
\draw [decorate,decoration={brace,amplitude=7pt}] (-1,0.3) to node[midway, yshift=0.5cm] {\tiny $k+3-n$} (1,0.3); 
\end{tikzpicture}
\;. 
\end{align*}
The right hand side belongs to the homogeneous Wiener chaos of order $(k+3-n)$, and as a consequence, we can bound its second moment by
\begin{align*}
\E |(\widehat{\Pi}_{0}^{\epsilon}\tau)(\varphi_{0}^{\lambda})|^{2} \lesssim \epsilon^{k} \phantom{1}
\begin{tikzpicture}[scale=0.6,baseline=-0.5cm]
\node at (-1.5,0) [dot] (left){};
\node at (1.5,0)  [dot] (right) {};
\node at (0,-1.5) [root] (below) {}; 
\draw[gepsilon] (left) to node[labl]{\tiny $k+3-n$} (right); 
\draw[dist] (left) to (below); 
\draw[dist] (right) to (below); 
\end{tikzpicture}
\phantom{1} \lesssim \epsilon^{\delta} \phantom{1}
\begin{tikzpicture}[scale=0.6,baseline=-0.5cm]
\node at (-1.5,0) [dot] (left){};
\node at (1.5,0)  [dot] (right) {};
\node at (0,-1.5) [root] (below) {}; 
\draw (left) to node[labl]{\tiny $3-n+\delta$} (right); 
\draw[dist] (left) to (below); 
\draw[dist] (right) to (below); 
\end{tikzpicture}, 
\end{align*}
which satisfies \eqref{eq:negative_bound} since
\begin{align*}
2 |\tau| = n - 3 - 2 (k+3-n)\kappa < n-3-\delta\;,
\end{align*}
if $\delta$ is small enough. The bound for $\eE^{\frac{k}{2}}(\Psi^{k+2}) X$ follows in exactly the same way. We have thus proved the bound \eqref{eq:negative_bound} for $\tau = \eE^{\frac{k}{2}}(\Psi^{k+3-n})$ and $\tau = \eE^{\frac{k}{2}}(\Psi^{k+2})X$.

\subsection{Second order renormalisation bounds}

We now turn to basis vectors coming from the terms $\fF' \iI(\fF)$, $\fF' \iI(\fF')$ and $\fF'' \iI(\fF)$. All these basis elements have the form
\begin{align*}
\tau = \eE^{a} \big( \Psi^{k} \iI (\eE^{b} \Psi^{n}) \big), 
\end{align*}
with the precise values of $a$ and $b$ depending on the element. For each $k$ and $n$, the element
\begin{align*}
(\hPi_{0}^{\epsilon} \tau)(\varphi_{0}^{\lambda})
\end{align*}
can be decomposed into homogeneous Wiener chaoses of orders
\begin{align*}
k + n - 2\ell, \qquad \ell = 0, 1, \cdots, k \wedge n. 
\end{align*}
By examining the homogeneities, we notice that all the $\eE^{\beta}$'s appearing in these elements play the role of multiplication by $\epsilon^{\beta}$ both under the canonical model and the Wick renormalised model. Thus, for the Wick model $\wPi_{0}(\varphi_{0}^{\lambda})$, its component in the $(k+n-2\ell)$-th homogeneous chaos is given by
\begin{equation} \label{eq:chaos_knl}
\ell !  \begin{pmatrix} k \\ \ell \end{pmatrix}  \begin{pmatrix} n \\ \ell \end{pmatrix} \cdot \epsilon^{a+b} \phantom{1}
\begin{tikzpicture}[scale=0.7,baseline=-0.2cm]
\node at (-1.2,0) [dot] (left){};
\node at (0,0.4)  [dot] (middle) {};
\node at (1.2,0)   [dot] (right) {}; 
\node at (0,1.8) [dot] (above) {}; 
\node at (0,0.8) {\tiny $\vdots$}; 
\node at (0,1.6) {\tiny $\vdots$}; 
\node at (0,1.1) {\scriptsize $\ell$}; 
\node at (0,-1.5) [root] (below) {}; 
\node at (-2.4,1.2) [var] (aboveleft) {}; 
\node at (-2.4,-1.2) [var] (belowleft) {}; 
\node at (2.4,1.2) [var] (aboveright) {}; 
\node at (2.4,-1.2) [var] (belowright) {}; 
\node at (-2.4,-0.3) {$\vdots$}; 
\node at (-2.4,0.7) {$\vdots$}; 
\node at (2.4,-0.3) {$\vdots$}; 
\node at (2.4,0.7) {$\vdots$}; 
\draw[kepsilon] (middle) to (left); 
\draw[kepsilon] (middle) to (right); 
\draw[kepsilon] (above) to (left); 
\draw[kepsilon] (above) to (right); 
\draw[kernel1] (left) to (right); 
\draw[testfcn] (right) to (below); 
\draw[kepsilon] (aboveleft) to (left); 
\draw[kepsilon] (belowleft) to (left); 
\draw[kepsilon] (aboveright) to (right); 
\draw[kepsilon] (belowright) to (right); 
\draw [decorate,decoration={brace,amplitude=7pt}] (-2.6,-1.2) to node[midway, xshift=-0.7cm] {\scriptsize $n - \ell$} (-2.6,1.2); 
\draw [decorate,decoration={brace,amplitude=7pt}] (2.6,1.2) to node[midway, xshift=0.6cm] {\scriptsize $k - \ell$} (2.6,-1.2); 
\end{tikzpicture}
\quad. 
\end{equation}
Note that the above expression is for the Wick renormalised model, and does not include effect of the map $M_{0}$ defined in Section~\ref{sec:renormalisation}.  We now 
discuss the convergence for these basis elements for different values of $k$, $n$ and $\ell$.

\begin{rmk}
\eqref{eq:chaos_knl} suggests that the bounds below will in general include a labelled graph introduced above as well as a factor of a positive power of $\epsilon$. With an abuse use of notation, in what follows, we will use $\gG$ to denote a labelled graph multiplied a certain power of $\epsilon$ (see for example \eqref{eq:chaos_knl_var} below). 
\end{rmk}

\subsubsection{$k + n - 2\ell \geq 2$}

We show below that in this case, there is no need for renormalisation beyond Wick ordering. For simplicity, we focus on the elements from the term $\fF' \iI (\fF)$, and the bounds for other basis vectors follow in essentially the same way. Such basis elements have the form and homogeneity
\begin{align*}
\tau_{k,n} = \eE^{\frac{k}{2}-1} \big( \Psi^{k} \iI(\eE^{\frac{n-3}{2}} \Psi^{n}) \big), \qquad |\tau_{k,n}| = -\frac{1}{2} - (k+n) \kappa. 
\end{align*}
If $k+n-2\ell \geq 2$, then as a consequence of the expression \eqref{eq:chaos_knl}, the second moment of the component of $(\hPi_{0}^{\epsilon} \tau_{k,n})(\varphi_{0}^{\lambda})$ in the $(k+n-2\ell)$-th homogeneous chaos is bounded by the graph
\begin{equation} \label{eq:chaos_knl_var}
\gG = \epsilon^{k+n-5}
\begin{tikzpicture} [scale=0.8,baseline=-0.0cm]
\node at (-1.5,1.5) [dot] (aboveleft) {}; 
\node at (1.5,1.5) [dot] (aboveright) {}; 
\node at (-1.5,0) [dot] (left) {}; 
\node at (1.5,0) [dot] (right) {}; 
\node at (0,-1.5) [root] (below) {}; 
\draw[gepsilon] (aboveleft) to node[labl]{\tiny $n-\ell$} (aboveright); 
\draw[gepsilon] (left) to node[labl]{\tiny $k-\ell$} (right); 
\draw[kernel1] (aboveleft) to node[labl]{\tiny $3,1$} (left); 
\draw[gepsilon, bend right = 70] (aboveleft) to node[anchor = east]{\tiny $\ell$} (left); 
\draw[kernel1] (aboveright) to node[labl]{\tiny $3,1$} (right); 
\draw[gepsilon, bend left = 70] (aboveright) to node[anchor = west]{\tiny $\ell$} (right); 
\draw[dist] (left) to (below); 
\draw[dist] (right) to (below); 
\end{tikzpicture}
. 
\end{equation}
Here, the upper and lower edges both have $r_{e} = 0$, so we simply omit it and just write the ``singularity'' $a_{e}$ for edges that does not contain either positive or negative renormalisations. According to \eqref{eq:negative_bound} and the homogeneity of $|\tau_{k,n}|$, we need to bound the graph by
\begin{equation} \label{eq:knl_bound}
|I_{\lambda}^{\gG}| \lesssim \epsilon^{\delta} \lambda^{-1 - \delta}
\end{equation}
for some small positive $\delta$. The assumption that there is a positive appearance of $\eE$ gives the condition
\begin{align*}
k \geq 2, \quad n \geq 3, \quad k + n \geq 6. 
\end{align*}
In order to get the bound \eqref{eq:knl_bound}, we need to assign powers of $\epsilon$'s to different edges of the graph to reduce the singularity of each edge to make the whole graph integrable. The assignments are different for various values of $k$, $n$ and $\ell$. 

For $\ell = 0$, we can assign $(n-3)$ powers of $\epsilon$ to the upper edge and $(k-2-\delta)$ powers to the lower edge, so we obtain the bound
\begin{equation} \label{eq:knl_l0}
\gG \lesssim \epsilon^{\delta} \phantom{1}
\begin{tikzpicture} [scale=0.6,baseline=-0.0cm]
\node at (-1.5,1.5) [dot] (aboveleft) {}; 
\node at (1.5,1.5) [dot] (aboveright) {}; 
\node at (-1.5,0) [dot] (left) {}; 
\node at (1.5,0) [dot] (right) {}; 
\node at (0,-1.5) [root] (below) {}; 
\draw[gepsilon] (aboveleft) to node[labl]{\tiny $3$} (aboveright); 
\draw[gepsilon] (left) to node[labl]{\tiny $2+\delta$} (right); 
\draw[kernel1] (aboveleft) to node[labl]{\tiny $3,1$} (left); 
\draw[kernel1] (aboveright) to node[labl]{\tiny $3,1$} (right); 
\draw[dist] (left) to (below); 
\draw[dist] (right) to (below); 
\end{tikzpicture}
\quad, 
\end{equation}
and we need to check that this graph satisfies Assumption~\ref{as:graph_assump}. We check for example the third item for $\bar{\vV}$ that consists of \tikz \node at (0,0) [root] {};  and the two lower vertices in the ``rectangle''. For this $\bar{\vV}$, $\eE_{0}(\bar{\vV})$ is the lower edge in the ``rectangle'' with $a_{e} = 2+\delta$,\footnote{Rigorously speaking, the two green edges representing test functions also belong to $\eE_{0}(\bar{\vV})$, but since we assume their degrees are $0$, so it does not matter.} $\eE^{\uparrow}(\bar{\vV})$ is empty, and $\eE^{\downarrow}(\bar{\vV})$ consists of the left and right edges, both with $r_{e} = 1$. One can then easily verify item~3 for this $\bar{\vV}$, and the rest of Assumption~\ref{as:graph_assump} can be checked in the same way, so that we obtain \eqref{eq:knl_bound} if $\delta$ is sufficiently small. 

For $\ell = 1$, we still assign $(n-3)$ powers of $\epsilon$ to the upper edge and $(k-2-\delta)$ powers to the lower one, but this time the graph is reduced to
\begin{equation} \label{eq:knl_l1}
\gG \lesssim \epsilon^{\delta} \phantom{1}
\begin{tikzpicture} [scale=0.6,baseline=-0.0cm]
\node at (-1.5,1.5) [dot] (aboveleft) {}; 
\node at (1.5,1.5) [dot] (aboveright) {}; 
\node at (-1.5,0) [dot] (left) {}; 
\node at (1.5,0) [dot] (right) {}; 
\node at (0,-1.5) [root] (below) {}; 
\draw[gepsilon] (aboveleft) to node[labl]{\tiny $2$} (aboveright); 
\draw[gepsilon] (left) to node[labl]{\tiny $1+\delta$} (right); 
\draw[kernel1] (aboveleft) to node[labl]{\tiny $4$} (left); 
\draw[kernel1] (aboveright) to node[labl]{\tiny $4$} (right); 
\draw[dist] (left) to (below); 
\draw[dist] (right) to (below); 
\end{tikzpicture}
\phantom{1} \lesssim \phantom{1} \epsilon^{\delta} \phantom{1}
\begin{tikzpicture} [scale=0.6,baseline=-0.0cm]
\node at (-1.5,1.5) [dot] (aboveleft) {}; 
\node at (1.5,1.5) [dot] (aboveright) {}; 
\node at (-1.5,0) [dot] (left) {}; 
\node at (1.5,0) [dot] (right) {}; 
\node at (0,-1.5) [root] (below) {}; 
\draw[gepsilon] (aboveleft) to node[labl]{\tiny $2$} (aboveright); 
\draw[gepsilon] (left) to node[labl]{\tiny $1+\delta$} (right); 
\draw[kernel1] (aboveleft) to node[labl]{\tiny $4+\delta$} (left); 
\draw[kernel1] (aboveright) to node[labl]{\tiny $4$} (right); 
\draw[dist] (left) to (below); 
\draw[dist] (right) to (below); 
\end{tikzpicture}
\phantom{1}. 
\end{equation}
Again, one can check that the conditions in Assumption~\ref{as:graph_assump} are all satisfied for this graph. 

We now turn to the situation when $\ell \geq 2$. By assigning $(\ell-2+\delta)$ powers of $\epsilon$ to both the leftmost and the rightmost edge with weight $\ell$, we reduce the graph to
\begin{equation} \label{eq:knl_l2}
\gG \lesssim \epsilon^{k+n-2 \ell-1-2 \delta} \phantom{1}
\begin{tikzpicture} [scale=0.6,baseline=-0.0cm]
\node at (-1.5,1.5) [dot] (aboveleft) {}; 
\node at (1.5,1.5) [dot] (aboveright) {}; 
\node at (-1.5,0) [dot] (left) {}; 
\node at (1.5,0) [dot] (right) {}; 
\node at (0,-1.5) [root] (below) {}; 
\draw[gepsilon] (aboveleft) to node[labl]{\tiny $n-\ell$} (aboveright); 
\draw[gepsilon] (left) to node[labl]{\tiny $k-\ell$} (right); 
\draw[kernel1] (aboveleft) to node[labl]{\tiny $5-\delta$} (left); 
\draw[kernel1] (aboveright) to node[labl]{\tiny $5-\delta$} (right); 
\draw[dist] (left) to (below); 
\draw[dist] (right) to (below); 
\end{tikzpicture}
\quad, 
\end{equation}
and the assumption $k+n-2\ell \geq 2$ guarantees there is still a positive power of $\epsilon$ left. If $\ell = n$, then we assign $(k-n-1-3\delta)$ powers to the lower edge, and we assign $(k-n+1-3\delta)$ powers to the lower edge if $\ell = n-1$. The graphs we get in these cases becomes
\begin{align*}
\gG \lesssim \epsilon^{\delta} \phantom{1}
\begin{tikzpicture} [scale=0.5,baseline=-0.0cm]
\node at (-1.5,1.5) [dot] (aboveleft) {}; 
\node at (1.5,1.5) [dot] (aboveright) {}; 
\node at (-1.5,0) [dot] (left) {}; 
\node at (1.5,0) [dot] (right) {}; 
\node at (0,-1.5) [root] (below) {}; 
\draw[gepsilon] (left) to node[labl]{\tiny $1+3\delta$} (right); 
\draw[kernel1] (aboveleft) to node[labl]{\tiny $5-\delta$} (left); 
\draw[kernel1] (aboveright) to node[labl]{\tiny $5-\delta$} (right); 
\draw[dist] (left) to (below); 
\draw[dist] (right) to (below); 
\end{tikzpicture}
\phantom{1} \lesssim \epsilon^{\delta}  \phantom{1}
\begin{tikzpicture} [scale=0.6,baseline=0.2cm]
\node at (0,0) [root] (below) {}; 
\node at (-1,1) [dot] (left) {}; 
\node at (1,1) [dot] (right) {}; 
\draw[dist] (below) to (left); 
\draw[dist] (below) to (right); 
\draw[gepsilon] (left) to node[anchor=south]{\tiny $1+3\delta$} (right); 
\end{tikzpicture}
\phantom{1} (\ell = n), \quad \gG \lesssim \epsilon^{\delta} \phantom{1}
\begin{tikzpicture} [scale=0.5,baseline=-0.0cm]
\node at (-1.5,1.5) [dot] (aboveleft) {}; 
\node at (1.5,1.5) [dot] (aboveright) {}; 
\node at (-1.5,0) [dot] (left) {}; 
\node at (1.5,0) [dot] (right) {}; 
\node at (0,-1.5) [root] (below) {}; 
\draw[gepsilon] (aboveleft) to node[labl]{\tiny $1$} (aboveright); 
\draw[gepsilon] (left) to node[labl]{\tiny $3 \delta$} (right); 
\draw[kernel1] (aboveleft) to node[labl]{\tiny $5-\delta$} (left); 
\draw[kernel1] (aboveright) to node[labl]{\tiny $5-\delta$} (right); 
\draw[dist] (left) to (below); 
\draw[dist] (right) to (below); 
\end{tikzpicture}
\phantom{1} (\ell = n - 1) \phantom{1}. 
\end{align*}
In both cases, one can easily verify Assumption~\ref{as:graph_assump} and conclude the desired bounds. Also, the first graph above does not contain the upper edge since in this case ($\ell = n$) that edge is a bounded continuous function, and we can simply omit edges with $a_{e} = 0$. 

We finally turn to $n - \ell \geq 2$. In this case, we assign powers of $\epsilon$'s in the following way: 
\begin{enumerate}
\item $(n-\ell-2)$ powers to the upper edge; 
\item $(1-3\delta)$ powers to the left edge; 
\item and $(k-\ell)$ powers to the lower edge. 
\end{enumerate}
The condition $n - \ell \geq 2$ guarantees that all the powers assigned above are positive, and there is still a $\delta$ power of $\epsilon$ left. In fact, we get the reduced graph
\begin{equation} \label{eq:knl_l2_more}
\epsilon^{\delta} \phantom{1}
\begin{tikzpicture} [scale=0.6,baseline=-0.0cm]
\node at (-1.5,1.5) [dot] (aboveleft) {}; 
\node at (1.5,1.5) [dot] (aboveright) {}; 
\node at (-1.5,0) [dot] (left) {}; 
\node at (1.5,0) [dot] (right) {}; 
\node at (0,-1.5) [root] (below) {}; 
\draw[gepsilon] (aboveleft) to node[labl]{\tiny $2$} (aboveright); 
\draw[kernel1] (aboveleft) to node[labl]{\tiny $4+2\delta$} (left); 
\draw[kernel1] (aboveright) to node[labl]{\tiny $5-\delta$} (right); 
\draw[dist] (left) to (below); 
\draw[dist] (right) to (below); 
\end{tikzpicture}
\qquad (n - \ell \geq 2) \quad. 
\end{equation}
Again, it is straightforward to check the Assumption \eqref{as:graph_assump} for this graph, and thus the bound \eqref{eq:knl_bound} is satisfied for small enough $\delta$. This finishes the proof of the case $k + n - 2\ell \geq 2$ for elements from $\fF' \iI (\fF)$. The case for the elements from the terms $\fF' \iI (\fF')$ and $\fF'' \iI (\fF)$ can be treated in essentially the same way, and we do not repeat the details here.

\subsubsection{$k = n = \ell$}

The basis elements in this category includes the following types: 
\begin{align*}
\eE^{\frac{n}{2}-1} \big( \Psi^{n} \iI(\eE^{\frac{n}{2}-1} \Psi^{n}) \big), \quad \eE^{\frac{n-1}{2}} \big( \Psi^{n} \iI(\eE^{\frac{n-3}{2}} \Psi^{n}) \big), \quad \eE^{\frac{n}{2}-1} \big( \Psi^{n} \iI(\eE^{\frac{n-3}{2}} \Psi^{n}) \big). 
\end{align*}
The homogeneities are just below $0$ for the first two elements, and just below $-\frac{1}{2}$ for the third one. For $\ell = n$, the $0$-th chaos component of the modelled distribution on these elements are just constants. 

We first treat the first two elements. For both of them, the contribution to the $0$-th chaos of $(\hPi_{0}^{\epsilon} \tau)(\varphi_{0}^{\lambda})$ is given by
\begin{equation} \label{eq:graph_0}
n! \cdot \epsilon^{n-2} \quad
\begin{tikzpicture}[scale=0.5,baseline=-0.0cm]
\node at (-1.5,0) [dot] (left){};
\node at (0,0.6)  [dot] (middle) {};
\node at (1.5,0)   [dot] (right) {}; 
\node at (0,2.2) [dot] (above) {}; 
\node at (0,1.05) {\tiny $\vdots$}; 
\node at (0,1.9) {\tiny $\vdots$}; 
\node at (0,1.3) {\scriptsize $n$}; 
\node at (0,-1) [root] (below) {}; 
\draw[kepsilon] (middle) to (left); 
\draw[kepsilon] (middle) to (right); 
\draw[kepsilon] (above) to (left); 
\draw[kepsilon] (above) to (right); 
\draw[kernel1] (left) to (right); 
\draw[testfcn] (right) to (below); 
\end{tikzpicture}
\quad - \phantom{1} n! \cdot C_{n}^{(\epsilon)} \quad
\begin{tikzpicture}[scale=0.6,baseline=-0.3cm]
\node at (0,1) [dot] (above){};
\node at (0,-1) [root] (below) {}; 
\draw[testfcn] (above) to (below); 
\end{tikzpicture}
\quad = \quad - n! \cdot \epsilon^{n-2} \phantom{1}
\begin{tikzpicture}[scale=0.6,baseline=-0.0cm]
\node at (-1.5,0) [dot] (left){};
\node at (0,0)  [dot] (middle) {};
\node at (1.5,0)   [dot] (right) {}; 
\node at (0,2) [dot] (above) {}; 
\node at (0,0.6) {\tiny $\vdots$}; 
\node at (0,1.6) {\tiny $\vdots$}; 
\node at (0,1) {\scriptsize $n$}; 
\node at (0,-1) [root] (below) {}; 
\draw[kepsilon] (middle) to (left); 
\draw[kepsilon] (middle) to (right); 
\draw[kepsilon] (above) to (left); 
\draw[kepsilon] (above) to (right); 
\draw[kernel] (left) to (below); 
\draw[testfcn] (right) to (below); 
\end{tikzpicture}
\phantom{1}, 
\end{equation}
where the equality comes from the definition of the kernel
\begin{tikzpicture}[scale=0.5,baseline=-0.1cm]
\node at (-1,0) [dot] (left){};
\node at (1,0) [dot] (right) {}; 
\draw[kernel1] (left) to (right); 
\end{tikzpicture}
as well as $C_{n}^{(\epsilon)}$ in \eqref{eq:Cn}. Since there is a strictly positive power of $\epsilon$, by assigning $(n-2-\delta)$ powers to the dotted line in the above graph, we deduce that this object can be bounded by the graph
\begin{equation} \label{eq:bound_0}
\gG = \epsilon^{\delta} \phantom{1}
\begin{tikzpicture}[scale=0.7,baseline=-0.5cm]
\node at (-1,0) [dot] (left){};
\node at (1,0) [dot] (right) {}; 
\node at (0,-1) [root] (below) {}; 
\draw[generic] (left) to node[labl]{\tiny $3$} (below); 
\draw[gepsilon] (left) to node[labl]{\tiny $2+\delta$} (right); 
\draw[dist] (right) to (below); 
\end{tikzpicture}
\phantom{1}. 
\end{equation}
It is then clear that one has $I_{\lambda}^{\gG} \lesssim \epsilon^{\delta} \lambda^{-\delta}$, 
which satisfies the bound \eqref{eq:knl_bound}. We now turn to the third element $\eE^{\frac{n}{2}-1} \big( \Psi^{n} \iI(\eE^{\frac{n-3}{2}} \Psi^{n}) \big)$. The expression of the $0$-th chaos is essentially the same as the previous two, except that one replaces $\epsilon^{n-2}$ by $\epsilon^{n-\frac{5}{2}}$, as well as the renormalisation constant $C_{n}^{(\epsilon)}$ by $C_{n}'^{(\epsilon)}$. Noting from \eqref{eq:Cn} that
\begin{equation} \label{eq:C_n}
C_{n}'^{(\epsilon)} = \epsilon^{-\frac{1}{2}} C_{n}^{(\epsilon)}, 
\end{equation}
we obtain the expression of the $0$-th chaos component of the element $(\hPi_{0}^{\epsilon} \tau)(\varphi_{0}^{\lambda})$ (up to the sign) as
\begin{align*}
n! \cdot \epsilon^{n-2} \phantom{1}
\begin{tikzpicture}[scale=0.5,baseline=-0.0cm]
\node at (-1.5,0) [dot] (left){};
\node at (0,0)  [dot] (middle) {};
\node at (1.5,0)   [dot] (right) {}; 
\node at (0,2) [dot] (above) {}; 
\node at (0,0.6) {\tiny $\vdots$}; 
\node at (0,1.6) {\tiny $\vdots$}; 
\node at (0,1) {\scriptsize $n$}; 
\node at (0,-1) [root] (below) {}; 
\draw[kepsilon] (middle) to (left); 
\draw[kepsilon] (middle) to (right); 
\draw[kepsilon] (above) to (left); 
\draw[kepsilon] (above) to (right); 
\draw[kernel] (left) to (below); 
\draw[testfcn] (right) to (below); 
\end{tikzpicture}
\quad \lesssim \quad \epsilon^{\delta} \phantom{1}
\begin{tikzpicture}[scale=0.8,baseline=-0.5cm]
\node at (-1,0) [dot] (left){};
\node at (1,0) [dot] (right) {}; 
\node at (0,-1) [root] (below) {}; 
\draw[generic] (left) to node[labl]{\tiny $3$} (below); 
\draw[gepsilon] (left) to node[labl]{\tiny $\frac{5}{2}+\delta$} (right); 
\draw[dist] (right) to (below); 
\end{tikzpicture}
, \quad n \geq 3, 
\end{align*}
where the above bound follows from assigning $n-\frac{5}{2}-\delta$ powers of $\epsilon$ to the kernels represented by the dotted lines. This expression is bounded by $\epsilon^{\delta} \lambda^{-\frac{1}{2} - \delta}$, and corresponds to the correct homogeneity (below $-\frac{1}{2}$) if $\delta$ is sufficiently small. We have thus proved the bound \eqref{eq:knl_bound} for the case $k=n=\ell$.

\subsubsection{$k=n+1$, $\ell = n$}

We now deal with the case $k = n+1$ and $\ell = n$, which belongs to the first order homogeneous chaos. There are two situations in this case; the first one includes basis vectors of the form
\begin{align*}
\tau = \eE^{\frac{n-1}{2}} \big(\Psi^{n+1} \iI( \eE^{\frac{n-3}{2}} \Psi^{n})  \big), \qquad |\tau| = -\frac{1}{2} - (2n+1) \kappa. 
\end{align*}
The $1$-st chaos component of $(\hPi_{0}^{\epsilon} \tau)(\varphi_{0}^{\lambda})$ is given by
\begin{equation} \label{eq:graph_1_renorm}
(n+1)! \phantom{1} \Bigg( \epsilon^{n-2} \phantom{1}
\begin{tikzpicture}[scale=0.5,baseline=-0.0cm]
\node at (-1.5,0) [dot] (left){};
\node at (0,0.6)  [dot] (middle) {};
\node at (1.5,0)   [dot] (right) {}; 
\node at (0,2.2) [dot] (above) {}; 
\node at (0,1.05) {\tiny $\vdots$}; 
\node at (0,1.9) {\tiny $\vdots$}; 
\node at (0,1.3) {\scriptsize $n$}; 
\node at (1.5,2.2) [var] (aboveright) {}; 
\node at (1.5,-2.2) [root] (belowright) {}; 
\draw[kepsilon] (middle) to (left); 
\draw[kepsilon] (middle) to (right); 
\draw[kepsilon] (above) to (left); 
\draw[kepsilon] (above) to (right); 
\draw[kernel1] (left) to (right); 
\draw[kepsilon] (aboveright) to (right); 
\draw[testfcn] (right) to (belowright); 
\end{tikzpicture}
\phantom{1} - C_{n}^{(\epsilon)} \phantom{1}
\begin{tikzpicture}[scale=0.6,baseline=-0.0cm]
\node at (0,1.5) [var] (above) {}; 
\node at (0,0) [dot] (middle) {}; 
\node at (0,-1.5) [root] (below) {}; 
\draw[kepsilon] (above) to (middle); 
\draw[testfcn] (middle) to (below); 
\end{tikzpicture}
\phantom{1} \Bigg)
\phantom{1} = \phantom{1} - (n+1)! \cdot \epsilon^{n-2} \phantom{1}
\begin{tikzpicture}[scale=0.5,baseline=-0.0cm]
\node at (0,1.8) [dot] (above) {}; 
\node at (0,0) [dot] (middle) {}; 
\node at (-1.5,0) [dot] (left) {}; 
\node at (1.5,0) [dot] (right) {}; 
\node at (0,-1.5) [root] (below) {}; 
\node at (1.5,1.8) [var] (aboveright) {}; 
\node at (0,1.5) {\tiny $\vdots$}; 
\node at (0,0.5) {\tiny $\vdots$}; 
\node at (0,0.9) {\scriptsize $n$}; 
\draw[kepsilon] (above) to (left); 
\draw[kepsilon] (above) to (right); 
\draw[kepsilon] (middle) to (left); 
\draw[kepsilon] (middle) to (right); 
\draw[kernel] (left) to (below); 
\draw[testfcn] (right) to (below); 
\draw[kepsilon] (aboveright) to (right); 
\end{tikzpicture}
\phantom{1}, 
\end{equation}
where we have used the expression of $C_{n}^{(\epsilon)}$ in \eqref{eq:Cn}. The second moment of this expression is then bounded (up to a constant multiple) by the graph
\begin{equation} \label{eq:bound_1_renorm}
\gG = \phantom{1} \epsilon^{2n-4} \phantom{1}
\begin{tikzpicture}[scale=1,baseline= 0.3cm]
\node at (0,0) [root] (middle) {}; 
\node at (-1,0) [dot] (left) {}; 
\node at (1,0) [dot] (right) {}; 
\node at (-1,1) [dot] (farleft) {}; 
\node at (1,1) [dot] (farright) {}; 
\draw[gepsilon] (farleft) to node[labl]{\tiny $1$} (farright); 
\draw[gepsilon] (farleft) to node[labl]{\tiny $n$} (left); 
\draw[gepsilon] (farright) to node[labl]{\tiny $n$} (right); 
\draw[kernel] (middle) to node[labl]{\tiny $3$} (left); 
\draw[kernel] (middle) to node[labl]{\tiny $3$} (right); 
\draw[dist] (middle) to (farleft); 
\draw[dist] (middle) to (farright); 
\end{tikzpicture}
\phantom{1} \lesssim \phantom{1} \epsilon^{2 \delta} \phantom{1}
\begin{tikzpicture}[scale=1,baseline= 0.3cm]
\node at (0,0) [root] (middle) {}; 
\node at (-1,0) [dot] (left) {}; 
\node at (1,0) [dot] (right) {}; 
\node at (-1,1) [dot] (farleft) {}; 
\node at (1,1) [dot] (farright) {}; 
\draw[gepsilon] (farleft) to node[labl]{\tiny $1$} (farright); 
\draw[gepsilon] (farleft) to node[labl]{\tiny $2+\delta$} (left); 
\draw[gepsilon] (farright) to node[labl]{\tiny $2+\delta$} (right); 
\draw[kernel] (middle) to node[labl]{\tiny $3$} (left); 
\draw[kernel] (middle) to node[labl]{\tiny $3$} (right); 
\draw[dist] (middle) to (farleft); 
\draw[dist] (middle) to (farright); 
\end{tikzpicture}
\phantom{1}, 
\end{equation}
which clearly satisfies the bound
\begin{align*}
I_{\lambda}^{\gG} \lesssim \epsilon^{2 \delta} \lambda^{-1 - 2 \delta}. 
\end{align*}
The exponent on $\lambda$ will be bigger than twice the homogeneity of $\tau$ for small enough $\delta$. Thus, the bound \eqref{eq:knl_bound} holds for the element $\eE^{\frac{n-1}{2}} \big(\Psi^{n+1} \iI( \eE^{\frac{n-3}{2}} \Psi^{n})  \big)$. 

The second situation for $k=n+1$ includes the basis elements
\begin{align*}
\tau = \eE^{\frac{n-1}{2}} \big( \Psi^{n+1} \iI (\eE^{\frac{n}{2}-1} \Psi^{n}) \big) \quad \text{or} \quad \tau = \eE^{\frac{n}{2}} \big( \Psi^{n+1} \iI (\eE^{\frac{n-3}{2}}\Psi^{n}) \big). 
\end{align*}
In both cases, we have $|\tau| = -(2n+1) \kappa$, just below $0$. Since there is no renormalisation beyond Wick ordering on these elements, the $1$-st chaos component of $(\hPi_{0}^{\epsilon} \tau)(\varphi_{0}^{\lambda})$ (for both of them) is given by
\begin{equation} \label{eq:graph_1_good}
(n+1)! \cdot \epsilon^{n-\frac{3}{2}} \phantom{1}
\begin{tikzpicture}[scale=0.5,baseline=-0.0cm]
\node at (-1.5,0) [dot] (left){};
\node at (0,0.6)  [dot] (middle) {};
\node at (1.5,0)   [dot] (right) {}; 
\node at (0,2.2) [dot] (above) {}; 
\node at (0,1.05) {\tiny $\vdots$}; 
\node at (0,1.9) {\tiny $\vdots$}; 
\node at (0,1.3) {\scriptsize $n$}; 
\node at (1.5,2.2) [var] (aboveright) {}; 
\node at (1.5,-2.2) [root] (belowright) {}; 
\draw[kepsilon] (middle) to (left); 
\draw[kepsilon] (middle) to (right); 
\draw[kepsilon] (above) to (left); 
\draw[kepsilon] (above) to (right); 
\draw[kernel1] (left) to (right); 
\draw[kepsilon] (aboveright) to (right); 
\draw[testfcn] (right) to (belowright); 
\end{tikzpicture}
\phantom{1}. 
\end{equation}
The second moment of this expression is bounded by the graph
\begin{equation} \label{eq:bound_1_good}
\gG = \phantom{1} \epsilon^{2n-3} \phantom{1}
\begin{tikzpicture} [scale=0.65,baseline=-0.0cm]
\node at (-1.5,1.5) [dot] (aboveleft) {}; 
\node at (1.5,1.5) [dot] (aboveright) {}; 
\node at (-1.5,0) [dot] (left) {}; 
\node at (1.5,0) [dot] (right) {}; 
\node at (0,-1.2) [root] (below) {}; 
\draw[gepsilon] (left) to node[labl]{\tiny $1$} (right); 
\draw[kernel1] (aboveleft) to node[labl]{\tiny $3,1$} (left); 
\draw[gepsilon, bend right = 70] (aboveleft) to node[anchor = east]{\tiny $n$} (left); 
\draw[kernel1] (aboveright) to node[labl]{\tiny $3,1$} (right); 
\draw[gepsilon, bend left = 70] (aboveright) to node[anchor = west]{\tiny $n$} (right); 
\draw[dist] (left) to (below); 
\draw[dist] (right) to (below); 
\end{tikzpicture}
\phantom{1} \lesssim \phantom{1} \epsilon^{\delta} \phantom{1}
\begin{tikzpicture} [scale=0.65,baseline=-0.0cm]
\node at (-1.5,1.5) [dot] (aboveleft) {}; 
\node at (1.5,1.5) [dot] (aboveright) {}; 
\node at (-1.5,0) [dot] (left) {}; 
\node at (1.5,0) [dot] (right) {}; 
\node at (0,-1.2) [root] (below) {}; 
\draw[gepsilon] (left) to node[labl]{\tiny $3 \delta$} (right); 
\draw[kernel1] (aboveleft) to node[labl]{\tiny $3,1$} (left); 
\draw[gepsilon, bend right = 70] (aboveleft) to node[anchor = east]{\tiny $2-\delta$} (left); 
\draw[kernel1] (aboveright) to node[labl]{\tiny $3,1$} (right); 
\draw[gepsilon, bend left = 70] (aboveright) to node[anchor = west]{\tiny $2-\delta$} (right); 
\draw[dist] (left) to (below); 
\draw[dist] (right) to (below); 
\end{tikzpicture}
\phantom{1}, 
\end{equation}
which immediately gives
\begin{align*}
I_{\lambda}^{\gG} \lesssim \epsilon^{\delta} \lambda^{-3\delta}. 
\end{align*}
Since the homogeneities for these two $\tau$'s are below $0$, we thus conclude the bound \eqref{eq:knl_bound} for this case.

\subsubsection{$n = k+1$, $\ell = k$}

We now turn to this last case. To keep notations consistent, we switch $n$ to $n+1$ and write the symbols as $\eE^{a} \big( \Psi^{n} \iI(\eE^{b} \Psi^{n+1}) \big)$ and $\ell = n$. The symbols in this category that need a mass renormalisation are of the form
\begin{align*}
\tau = \eE^{\frac{n}{2} - 1} \big( \Psi^{n} \iI (\eE^{\frac{n}{2}-1} \Psi^{n+1}) \big), \quad |\tau| = -\frac{1}{2} - (2n+1) \kappa, \quad n \geq 3. 
\end{align*}
The component in the $1$-st Wiener chaos of the object $(\hPi_{0}^{\epsilon} \tau)(\varphi_{0}^{\lambda})$ is given by
\begin{equation} \label{eq:graph_2_renorm}
(n+1)! \phantom{1} \Bigg( \epsilon^{n-2} \phantom{1}
\begin{tikzpicture}[scale=0.5,baseline=-0.0cm]
\node at (-1.5,0) [dot] (left){};
\node at (0,0.6)  [dot] (middle) {};
\node at (1.5,0)   [dot] (right) {}; 
\node at (0,2.2) [dot] (above) {}; 
\node at (0,1.05) {\tiny $\vdots$}; 
\node at (0,1.9) {\tiny $\vdots$}; 
\node at (0,1.3) {\scriptsize $n$}; 
\node at (-1.5,-2.2) [var] (belowleft) {}; 
\node at (1.5,-2.2) [root] (belowright) {}; 
\draw[kepsilon] (middle) to (left); 
\draw[kepsilon] (middle) to (right); 
\draw[kepsilon] (above) to (left); 
\draw[kepsilon] (above) to (right); 
\draw[kernel] (left) to (right); 
\draw[kepsilon] (belowleft) to (left); 
\draw[testfcn] (right) to (belowright); 
\end{tikzpicture}
\phantom{1} - C_{n}^{(\epsilon)} \phantom{1}
\begin{tikzpicture}[scale=0.6,baseline=-0.0cm]
\node at (0,1.7) [var] (above) {}; 
\node at (0,0) [dot] (middle) {}; 
\node at (0,-1.7) [root] (below) {}; 
\draw[kepsilon] (above) to (middle); 
\draw[testfcn] (middle) to (below); 
\end{tikzpicture}
\phantom{1} \Bigg) \phantom{1} - \phantom{1} (n+1)! \cdot \epsilon^{n-2} \phantom{1}
\begin{tikzpicture}[scale=0.7,baseline=-0.0cm]
\node at (-1,1.5) [var] (aboveleft) {}; 
\node at (-1,0) [dot] (left) {}; 
\node at (0,0) [dot] (middle) {}; 
\node at (1,0) [dot] (right) {}; 
\node at (0,1.5) [dot] (above) {}; 
\node at (0,-1.5) [root] (below) {}; 
\node at (0,1.2) {\tiny $\vdots$}; 
\node at (0,0.3) {\tiny $\vdots$}; 
\node at (0,0.7) {\scriptsize $n$}; 
\draw[kepsilon] (aboveleft) to (left); 
\draw[kepsilon] (above) to (left); 
\draw[kepsilon] (above) to (right); 
\draw[kepsilon] (middle) to (left); 
\draw[kepsilon] (middle) to (right); 
\draw[kernel] (left) to (below); 
\draw[testfcn] (right) to (below); 
\end{tikzpicture}
\phantom{1}. 
\end{equation}
The second moment of the last term above is relatively easier to to treat. In fact, it is bounded by the graph
\begin{equation} \label{eq:bound_22_renorm}
\gG = \epsilon^{n-2} \phantom{1}
\begin{tikzpicture}[scale=0.7,baseline=0.3cm]
\node at (0,0) [root] (middle) {}; 
\node at (-1.5,0) [dot] (left) {}; 
\node at (1.5,0) [dot] (right) {}; 
\node at (-1.5,1.5) [dot] (aboveleft) {}; 
\node at (1.5,1.5) [dot] (aboveright) {}; 
\draw[gepsilon] (aboveleft) to node[labl]{\tiny $1$} (aboveright); 
\draw[gepsilon] (aboveleft) to node[anchor=east]{\tiny $n$} (left); 
\draw[gepsilon] (aboveright) to node[anchor=west]{\tiny $n$} (right); 
\draw[kernel] (aboveleft) to node[labl]{\tiny $3$} (middle); 
\draw[kernel] (aboveright) to node[labl]{\tiny $3$} (middle); 
\draw[dist] (left) to (middle); 
\draw[dist] (right) to (middle); 
\end{tikzpicture}
\phantom{1} \lesssim \phantom{1} \epsilon^{2 \delta} \phantom{1}
\begin{tikzpicture}[scale=0.7,baseline=0.3cm]
\node at (0,0) [root] (middle) {}; 
\node at (-1.5,0) [dot] (left) {}; 
\node at (1.5,0) [dot] (right) {}; 
\node at (-1.5,1.5) [dot] (aboveleft) {}; 
\node at (1.5,1.5) [dot] (aboveright) {}; 
\draw[gepsilon] (aboveleft) to node[labl]{\tiny $1$} (aboveright); 
\draw[gepsilon] (aboveleft) to node[anchor=east]{\tiny $2+\delta$} (left); 
\draw[gepsilon] (aboveright) to node[anchor=west]{\tiny $2+\delta$} (right); 
\draw[kernel] (aboveleft) to node[labl]{\tiny $3$} (middle); 
\draw[kernel] (aboveright) to node[labl]{\tiny $3$} (middle); 
\draw[dist] (left) to (middle); 
\draw[dist] (right) to (middle); 
\end{tikzpicture}
\phantom{1}, 
\end{equation}
which clearly gives the desired bound $I_{\lambda}^{\gG} \lesssim \epsilon^{2 \delta} \lambda^{-1-2\delta}$. For the two terms in the parenthesis, by the definition of $C_{n}^{(\epsilon)}$, their difference can be expressed by the graph
\begin{equation} \label{eq:graph_21_renorm}
\epsilon^{n-2} \phantom{1}
\begin{tikzpicture}[scale=0.7,baseline=0.0cm]
\node at (-1,1) [dot] (left) {}; 
\node at (1,1) [dot] (right) {}; 
\node at (-1,-0.5) [var] (belowleft) {}; 
\node at (1,-0.5) [root] (belowright) {}; 
\draw[kernelBig] (left) to (right); 
\draw[kepsilon] (belowleft) to (left); 
\draw[testfcn] (right) to (belowright); 
\end{tikzpicture}
\phantom{1}, 
\end{equation}
where
\begin{tikzpicture}[scale=0.5,baseline=-0.15cm]
\node at (-1,0) [dot] (left) {}; 
\node at (1,0) [dot] (right) {}; 
\draw[kernelBig] (left) to (right); 
\end{tikzpicture}
denotes the renormalised distribution/kernel $\sR (K G_{\epsilon}^{n})$, which has degree $n+3$. Thus, the second moment of this object is bounded by
\begin{equation} \label{eq:bound_21_renorm}
\epsilon^{2n-4} \phantom{1}
\begin{tikzpicture}[scale=0.7,baseline=0.0cm]
\node at (-1,1) [dot] (left) {}; 
\node at (1,1) [dot] (right) {}; 
\node at (-1,-0.5) [dot] (belowleft) {}; 
\node at (1,-0.5) [dot] (belowright) {}; 
\node at (0,-1.5) [root] (below) {}; 
\draw[gepsilon] (left) to node[labl]{\tiny $1$} (right); 
\draw[gepsilon, ->] (left) to node[anchor=east]{\tiny $n+3,-1$} (belowleft); 
\draw[gepsilon, ->] (right) to node[anchor=west]{\tiny $n+3,-1$} (belowright); 
\draw[dist] (below) to (belowleft); 
\draw[dist] (below) to (belowright); 
\end{tikzpicture}
\phantom{1} \lesssim \phantom{1} \epsilon^{2\delta}
\begin{tikzpicture}[scale=0.7,baseline=0.0cm]
\node at (-1,1) [dot] (left) {}; 
\node at (1,1) [dot] (right) {}; 
\node at (-1,-0.5) [dot] (belowleft) {}; 
\node at (1,-0.5) [dot] (belowright) {}; 
\node at (0,-1.5) [root] (below) {}; 
\draw[gepsilon] (left) to node[labl]{\tiny $1$} (right); 
\draw[gepsilon, ->] (left) to node[anchor=east]{\tiny $5+\delta,-1$} (belowleft); 
\draw[gepsilon, ->] (right) to node[anchor=west]{\tiny $5+\delta,-1$} (belowright); 
\draw[dist] (below) to (belowleft); 
\draw[dist] (below) to (belowright); 
\end{tikzpicture}
\phantom{1}. 
\end{equation}
Again, one can verify that Assumption~\ref{as:graph_assump} is satisfied, and thus one has
\begin{align*}
I_{\lambda}^{\gG} \lesssim \epsilon^{2 \delta} \lambda^{-1-2\delta}, 
\end{align*}
which vanishes at the right homogeneities if $\delta$ is sufficiently small. 

We now turn to the other two terms in this category, which are of the forms
\begin{align*}
\eE^{\frac{n}{2}-1} \big( \Psi^{n} \iI(\eE^{\frac{n-1}{2}} \Psi^{n}) \big), \qquad \eE^{\frac{n-1}{2}} \big( \Psi^{n} \iI(\eE^{\frac{n}{2}-1} \Psi^{n}) \big), 
\end{align*}
and both have homogeneities just below $0$. For both symbols, the components of $(\hPi_{0}^{\epsilon}\tau)(\varphi_{0}^{\lambda})$ in the $1$-st chaos can be expressed by
\begin{equation} \label{eq:graph_2_good}
(n+1)! \cdot \epsilon^{n-\frac{3}{2}} \phantom{1}
\begin{tikzpicture}[scale=0.5,baseline=-0.0cm]
\node at (-1.5,0) [dot] (left){};
\node at (0,0.6)  [dot] (middle) {};
\node at (1.5,0)   [dot] (right) {}; 
\node at (0,2.2) [dot] (above) {}; 
\node at (0,1.05) {\tiny $\vdots$}; 
\node at (0,1.9) {\tiny $\vdots$}; 
\node at (0,1.3) {\scriptsize $n$}; 
\node at (-1.5,-2.2) [var] (belowleft) {}; 
\node at (1.5,-2.2) [root] (belowright) {}; 
\draw[kepsilon] (middle) to (left); 
\draw[kepsilon] (middle) to (right); 
\draw[kepsilon] (above) to (left); 
\draw[kepsilon] (above) to (right); 
\draw[kernel1] (left) to (right); 
\draw[kepsilon] (belowleft) to (left); 
\draw[testfcn] (right) to (belowright); 
\end{tikzpicture}
\phantom{1}, 
\end{equation}
whose second moment is bounded by the graph
\begin{equation} \label{eq:bound_2_good}
\gG = \phantom{1} \epsilon^{2n-3} \phantom{1}
\begin{tikzpicture} [scale=0.65,baseline=-0.0cm]
\node at (-1.5,1.5) [dot] (aboveleft) {}; 
\node at (1.5,1.5) [dot] (aboveright) {}; 
\node at (-1.5,0) [dot] (left) {}; 
\node at (1.5,0) [dot] (right) {}; 
\node at (0,-1.5) [root] (below) {}; 
\draw[gepsilon] (aboveleft) to node[labl]{\tiny $1$} (aboveright); 
\draw[kernel1] (aboveleft) to node[labl]{\tiny $3,1$} (left); 
\draw[gepsilon, bend right = 70] (aboveleft) to node[anchor = east]{\tiny $n$} (left); 
\draw[kernel1] (aboveright) to node[labl]{\tiny $3,1$} (right); 
\draw[gepsilon, bend left = 70] (aboveright) to node[anchor = west]{\tiny $n$} (right); 
\draw[dist] (left) to (below); 
\draw[dist] (right) to (below); 
\end{tikzpicture}
\phantom{1} \lesssim \phantom{1} \epsilon^{\delta} \phantom{1}
\begin{tikzpicture} [scale=0.65,baseline=-0.0cm]
\node at (-1.5,1.5) [dot] (aboveleft) {}; 
\node at (1.5,1.5) [dot] (aboveright) {}; 
\node at (-1.5,0) [dot] (left) {}; 
\node at (1.5,0) [dot] (right) {}; 
\node at (0,-1.5) [root] (below) {}; 
\draw[gepsilon] (aboveleft) to node[labl]{\tiny $3 \delta$} (aboveright); 
\draw[kernel1] (aboveleft) to node[labl]{\tiny $3,1$} (left); 
\draw[gepsilon, bend right = 70] (aboveleft) to node[anchor = east]{\tiny $2-\delta$} (left); 
\draw[kernel1] (aboveright) to node[labl]{\tiny $3,1$} (right); 
\draw[gepsilon, bend left = 70] (aboveright) to node[anchor = west]{\tiny $2-\delta$} (right); 
\draw[dist] (left) to (below); 
\draw[dist] (right) to (below); 
\end{tikzpicture}
\phantom{1}. 
\end{equation}
Again, this object vanishes at the correct homogeneity. This concludes the proof of the bound \eqref{eq:knl_bound} for all symbols with negative homogeneity that contains a strictly positive appearance of $\eE$.

\subsection{The bounds \eqref{eq:group_bound} and \eqref{eq:positive_bound}}

We first deal with the bound \eqref{eq:group_bound} on $\widehat{f}^{\epsilon}$. By inspection of the formal right hand side of the abstract equation, we need to prove \eqref{eq:group_bound} for $\beta = \frac{j-1}{2}$ and formal symbols $\tau$ of the form
\begin{align*}
\tau = \Psi^{j+2-n} \underbrace{\big(\iI(\eE^{\frac{q-1}{2}} \Psi^{q+2}) \big)^{a} \big(\iI(\eE^{\frac{q-1}{2}} \Psi^{q+1}) \big)^{b} X^{c}}_{\sigma}, \qquad n \geq 4, \quad a+b+c \leq n. 
\end{align*}
Since $\widehat{f}_{z}^{\epsilon} = \wf_{z}$, and that the Wick renormalised model $(\wPi, \wf)$ satisfies the relation \eqref{eq:model_epsilon}, we have
\begin{align*}
\widehat{f}_{z}^{\epsilon}(\sE^{\frac{j-1}{2}}_{0}\tau) = -\epsilon^{\frac{j-1}{2}} (\wPi_{z} \Psi^{j+2-n})(z) (\wPi_{z}\sigma)(z), 
\end{align*}
where $\sigma$ is the basis vector as indicated above. Since the homogeneity of $\sigma$ is strictly positive, the expression above is $0$ if any of the factors of $\sigma$ has a positive power. Thus, the only situation we need to consider for the bound on $\widehat{f}_{z}^{\epsilon} (\sE^{\frac{j-1}{2}}_{0}\tau)$ is $\tau = \Psi^{j+2-n}$, and as a consequence, we get
\begin{align*}
D^{\ell} \widehat{f}_{z}^{\epsilon}(\sE^{\frac{j-1}{2}}_{0}(\Psi^{j+2-n})) = -\epsilon^{\frac{j-1}{2}} (D^{\ell} \wPi_{0} \Psi^{j+2-n})(z). 
\end{align*}
By the definition of $\wPi$, the right hand side above can be expressed as a Hermite polynomial, each term being proportional to
\begin{align*}
\epsilon^{\frac{j-1}{2}} (C_{1}^{(\epsilon)})^{k} (D^{\ell} \Psi_{\epsilon}^{j+2-n-2k})(z) = \epsilon^{\frac{j-1}{2}} (C_{1}^{(\epsilon)})^{k} \sum_{\sum |q_{i}| = |\ell|} (D^{q_{1}} \Psi_{\epsilon})(z) \cdots (D^{q_{j+2-n-2k}} \Psi_{\epsilon}) (z). 
\end{align*}
where we have written $\Psi_{\epsilon} = \Pi_{0} \Psi = K * \xi_{\epsilon}$ for simplicity. Now, taking expectation on the right hand side above, using generalised H\"{o}lder's inequality, and the fact that $C_{1}^{\epsilon} \sim \epsilon^{-1}$, we get
\begin{equation} \label{eq:bound_group}
\E |D^{\ell} \widehat{f}_{z}^{\epsilon} (\sE^{\frac{j-1}{2}}_{0}(\Psi^{j+2-n}))| \lesssim \max_{k \leq \frac{1}{2}(j+2-n)} \epsilon^{\frac{j-1}{2}-k} \sum_{\sum |q_{i}| = |\ell|} \prod_{i} \big(\E|D^{q_{i}}  \Psi_{\epsilon}|^{j+2-n-2k} \big)^{\frac{1}{j+2-n-2k}}. 
\end{equation}
By equivalence of moments in Wiener chaos, each of the above factor is equivalent to $\E |(D^{q_{i}} \Psi_{\epsilon})(z)|$, which could be bounded by
\begin{equation} \label{eq:equiv_moments}
\E |(D^{q_{i}} \Psi_{\epsilon})(z)| \lesssim \big( \E |(D^{q_{i}} \Psi_{\epsilon})(z)|^{2} \big)^{\frac{1}{2}} \lesssim \epsilon^{-\frac{1}{2} - |q_{i}|}, 
\end{equation}
where we have used $\E |(D^{q_{i}}K * \xi_{\epsilon})|^{2} \lesssim \epsilon^{-1-2|q_{i}|}$. Combining \eqref{eq:bound_group} and \eqref{eq:equiv_moments}, we get
\begin{align*}
\E |D^{\ell} \widehat{f}_{z}^{\epsilon} (\sE^{\frac{j-1}{2}}_{0}(\Psi^{j+2-n}))| \lesssim \epsilon^{\frac{j-1}{2} - \frac{j+2-n}{2} - |\ell|}, 
\end{align*}
where we used the fact that there are totally $j+2-n-2k$ factors in the product, and $\sum |q_{i}| = |\ell|$. Since $|\tau| < - \frac{j+2-n}{2}$, this establishes the bound \eqref{eq:group_bound}. 

We now turn to the bound \eqref{eq:positive_bound} for $\tau \in \uU$, which includes $\Psi$, $\iI(\eE^{\frac{j-1}{2}} \Psi^{j+2})$, $\iI(\eE^{\frac{j-1}{2}} \Psi^{j+1})$, $\1$ and $X$. The bound is trivial for $\1$ and $X$, and is also straightforward for $\Psi$. The treatment for the rest two basis elements are similar, and we only give details for $\tau = \iI(\eE^{\frac{j-1}{2}} \Psi^{j+2})$. Since the test function $\psi$ annihilates affine functions, we have
\begin{equation} \label{eq:graph_positive}
(\hPi_{z}^{\epsilon} \tau)(\psi_{z}^{\lambda}) = \epsilon^{\frac{j-1}{2}} \phantom{1} 
\begin{tikzpicture}[scale=0.8,baseline=-0.0cm]
\node at (-1.5,0) [root] (left) {}; 
\node at (0,0) [dot] (middle) {}; 
\node at (1.5,0) [dot] (right) {}; 
\node at (2.5,1) [var] (above) {}; 
\node at (2.5,-1) [var] (below) {}; 
\node at (2.5,-0.25) {$\vdots$}; 
\node at (2.5, 0.55) {$\vdots$}; 
\draw[testfcn] (middle) to (left); 
\draw[kernel2] (right) to (middle); 
\draw[kepsilon] (above) to (right); 
\draw[kepsilon] (below) to (right); 
\draw [decorate,decoration={brace,amplitude=7pt}] (2.7,1) to node[midway, xshift=0.6cm] {\scriptsize $j+2$} (2.7,-1); 
\end{tikzpicture}
\phantom{1}. 
\end{equation}
It then follows that we have the bound
\begin{equation} \label{eq:bound_positive}
\E |(\hPi_{z}^{\epsilon} \tau)(\psi_{z}^{\lambda})|^{2} \lesssim \epsilon^{j-1} \phantom{1}
\begin{tikzpicture} [scale=0.75,baseline=-0.0cm]
\node at (-1.5,1.5) [dot] (aboveleft) {}; 
\node at (1.5,1.5) [dot] (aboveright) {}; 
\node at (-1.5,0) [dot] (left) {}; 
\node at (1.5,0) [dot] (right) {}; 
\node at (0,-1.5) [root] (below) {}; 
\draw[gepsilon] (aboveleft) to node[anchor=south]{\tiny $j+2$} (aboveright); 
\draw[kernel1] (aboveleft) to node[labl]{\tiny $3,2$} (left); 
\draw[kernel1] (aboveright) to node[labl]{\tiny $3,2$} (right); 
\draw[dist] (left) to (below); 
\draw[dist] (right) to (below); 
\end{tikzpicture}
\phantom{1} \lesssim \phantom{1} \epsilon^{2(|\tau|-\zeta)} \phantom{1}
\begin{tikzpicture} [scale=0.75,baseline=-0.0cm]
\node at (-1.5,1.5) [dot] (aboveleft) {}; 
\node at (1.5,1.5) [dot] (aboveright) {}; 
\node at (-1.5,0) [dot] (left) {}; 
\node at (1.5,0) [dot] (right) {}; 
\node at (0,-1.5) [root] (below) {}; 
\draw[gepsilon] (aboveleft) to node[anchor=south]{\tiny $3+2|\tau|-2\zeta$} (aboveright); 
\draw[kernel1] (aboveleft) to node[labl]{\tiny $3,2$} (left); 
\draw[kernel1] (aboveright) to node[labl]{\tiny $3,2$} (right); 
\draw[dist] (left) to (below); 
\draw[dist] (right) to (below); 
\end{tikzpicture}
\phantom{1}. 
\end{equation}
Since $2 \zeta \in (2,3)$ and $2 |\tau| = 1 - 2(j+2) \kappa$, the conditions for Assumption~\ref{as:graph_assump} can be verified straightforwardly, and thus one obtains
\begin{align*}
\E |(\hPi_{z}^{\epsilon} \tau)(\psi_{z}^{\lambda})|^{2}  \lesssim \lambda^{1 - 2|\tau| + 2 \zeta} \epsilon^{2 |\tau| - 2 |\zeta|} = \lambda^{2 \zeta + \theta} \epsilon^{2 |\tau| - 2 |\zeta|}
\end{align*}
for some positive $\theta$. The bound for $\tau = \iI(\eE^{\frac{j-1}{2}} \Psi^{j+1})$ follows in essentially the same way.

\section{Identification of the limits}
\label{sec:limits}

We are now ready to address the main theme of the article: identifying the large scale limits of microscopic models under various assumptions on $V$. As mentioned in the introduction, we will see that, in both the weakly nonlinear and weak noise regime, the large scale limit of these near-critical models are described by $\Phi^4_3$ as long as $V$ is symmetric, but described by either $\Phi^3_3$ or OU processes when asymmetry is present. The only difference is that the critical $\theta$ at which one sees a a pitchfork or saddle-node bifurcation is different. 

We will formulate precisely and prove these results below, starting with the weakly nonlinear regime.

\subsection{Weakly nonlinear regime}

Let $\tilde{u}$ be a process on a large torus satisfying
\begin{align*}
\partial_{t} \tilde{u} = \Delta \tilde{u} - \epsilon V_{\theta}'(\tilde{u}) + \widehat{\xi}, 
\end{align*}
and the re-centered and rescaled process $u_{\epsilon^{\alpha}}$ to be
\begin{align*}
u_{\epsilon^{\alpha}} = \epsilon^{-\frac{\alpha}{2}} \big( \tilde{u} (t/\epsilon^{2 \alpha}, x/\epsilon^{\alpha}) - h  \big), 
\end{align*}
where $\alpha$ is the scale, and $h$ is a small parameter depending on $\epsilon$, both to be chosen later. By setting $\delta = \epsilon^{\alpha}$, it is easy to see that $u_{\delta}$ satisfies the equation
\begin{align*}
\partial_{t} u_{\delta} = \Delta u_{\delta} - \delta^{\frac{1}{\alpha} - \frac{5}{2}} V_{\theta}'(\delta^{\frac{1}{2}} u_{\delta} + h) + \delta^{-\frac{5}{2}} \widehat{\xi} (t/\delta^{2}, x/\delta). 
\end{align*}
Note that the noise term is equivalent in law to $\xi * \rho_{\delta}$ for some mollifier $\rho$ rescaled at size $\delta$, expanding $V_{\theta}'$ with respect to Hermite polynomials, we get
\begin{equation} \label{eq:asymmetric_delta}
\partial_{t} u_{\delta} = \Delta u_{\delta} - \delta^{\frac{1}{\alpha}-1} \sum_{j=0}^{m} \ha_{j}^{(h)}(\theta) \cdot \delta^{\frac{j-3}{2}} H_{j} (u_{\delta}; C_{1}^{(\delta)}) + \xi_{\delta}, 
\end{equation}
where
\begin{align*}
\ha_{j}^{(h)}(\theta) = \sum_{k=j}^{m} \begin{pmatrix} k \\ j \end{pmatrix} \ha_{k}(\theta) \cdot h^{k-j}. 
\end{align*}
We now fix $\gamma \in (1,\frac{6}{5})$, $\eta \in (-\frac{m+1}{2m}, \frac{1}{2})$, and we shall lift the above equation to the abstract $\dD^{\gamma,\eta}_{\epsilon}$ space associated to the model $\fM_{\delta} = M_{\delta} \sL_{\delta} (\xi_{\delta})$ as in Theorem~\ref{th:main_convergence}. We also let $\phi_{0}^{(\delta)} \in \cC^{\gamma,\eta}_{\delta}$ such that $\| \phi_{0}^{(\delta)}; \phi_{0} \|_{\gamma,\eta;\delta} \rightarrow 0$ for some $\phi_{0} \in \cC^{\eta}$. The corresponding abstract fixed point equation then has the form
\begin{equation} \label{eq:abstract_approximate}
\Phi^{(\delta)} = \pP \1_{+} \bigg( \Xi - \sum_{j=4}^{m} \lambda_{j}^{(\delta)} \qQ_{\leq 0} \heE^{\frac{j-3}{2}} \qQ_{\leq 0} \big( (\Phi^{(\delta)})^{j} \big) - \sum_{j=0}^{3} \lambda_{j}^{(\delta)} \qQ_{\leq 0} \big( (\Phi^{(\delta)})^{j} \big) \bigg) + \widehat{P} \phi_{0}^{(\delta)}. 
\end{equation}
Comparing the right hand sides of \eqref{eq:renormalised_equation} and \eqref{eq:asymmetric_delta}, we should choose the coefficients $\lambda_{j}^{(\delta)}$'s to be
\begin{align*}
\lambda_{j}^{(\delta)} &= \delta^{\frac{1}{\alpha}-1} \cdot \ha_{j}^{(h)}(\theta), \qquad j \geq 3; \\
\lambda_{2}^{(\delta)} &= \delta^{\frac{1}{\alpha} - \frac{3}{2}} \cdot \ha_{2}^{(h)}(\theta); \\
\lambda_{1}^{(\delta)} &= \delta^{\frac{1}{\alpha} - 2} \cdot \ha_{1}^{(h)}(\theta) - \delta^{\frac{2}{\alpha}-2} C_{\delta,\theta,h}; \\
\lambda_{0}^{(\delta)} &= \delta^{\frac{1}{\alpha} - \frac{5}{2}} \cdot \ha_{0}^{(h)}(\theta) - \delta^{\frac{2}{\alpha}-\frac{5}{2}} C'_{\delta,\theta,h} - 6 \lambda_{2}^{(\delta)} \lambda_{3}^{(\delta)} C_{2}^{(\delta)}, 
\end{align*}
where
\begin{align*}
C_{\delta,\theta,h} &= \sum_{n=2}^{m-1} (n+1)^{2} n! \cdot \big( \ha_{n+1}^{(h)}(\theta) \big)^{2} \cdot C_{n}^{(\delta)} + \sum_{n=3}^{m-2} (n+2)! \cdot \ha_{n}^{(h)}(\theta) \cdot \ha_{n+2}^{(h)} (\theta) \cdot C_{n}^{(\delta)} \\
&= 18 \ha_{3}^{2} c_{2} |\log \delta| + \oO(1); \\
C_{\delta,\theta,h}' &= \sum_{n=3}^{m-1} (n+1)! \cdot \ha_{n}^{(h)}(\theta) \cdot \ha_{n+1}(\theta) C_{n}^{(\delta)} = A + \oO(\delta, \theta, h), 
\end{align*}
Here, the quantity $A$ is given by
\begin{equation} \label{eq:quantity_a}
A = \sum_{n=3}^{m-1} (n+1)! \cdot \ha_{n} \ha_{n+1} C_{n}, 
\end{equation}
and $C_{n}$'s are the limits of $C_{n}^{(\delta)}$'s (recall that they do converge to a finite limit for $n \geq 3$). It is then clear that the reconstructed solution $u_{\delta} = \widehat{\rR}^{\delta} \Phi^{(\delta)}$ exactly solves \eqref{eq:asymmetric_delta} with initial condition $\phi_{0}^{(\delta)}$. Here, we have used the notation $\oO(a,b)$ to denote $\oO(a \vee b)$. 

By Theorem~\ref{th:main_convergence}, there exists a limiting model $\fM \in \sM_{0}$ such that $\$ \fM_{\delta}; \fM \$_{\delta;0} \rightarrow 0$. If $\lambda_{j}^{(\delta)}$ converges to some $\lambda_{j} \in \R$ for each $j$, then by Theorem~\ref{th:fixed_pt}, we will have $\$ \Phi^{(\delta)}; \Phi \$_{\gamma,\eta;\delta} \rightarrow 0$, where $\Phi \in \dD^{\gamma,\eta}$ associated to the model $\fM$ solves the fixed point equation
\begin{equation} \label{eq:abstract_limit}
\Phi = \pP \1_{+} \bigg( \Xi - \sum_{j=4}^{m} \lambda_{j} \qQ_{\leq 0} \heE^{\frac{j-3}{2}} \big(\qQ_{\leq 0} (\Phi^{j}) \big) - \sum_{j=0}^{3} \lambda_{j} \qQ_{\leq 0} (\Phi^{j}) \bigg) + \widehat{P} \phi_{0}.  
\end{equation}
The continuity of the reconstruction operator thus implies $u_{\delta} \rightarrow u = \widehat{\rR} \Phi$ in $\cC^{\eta}$. In what follows, we will choose the small parameter $h$ as well as the scale $\alpha$ in a proper way such that the coefficients $\lambda_{j}^{(\delta)}$'s do converge to the desired limiting values under various assumptions on $V$. Once these limiting values $\lambda_{j}$'s are known, we can immediately derive the limiting equation that $u$ solves. We will also always assume that $u_{\delta}$ solves \eqref{eq:asymmetric_delta} on $[0,T] \times \TT^{3}$ with initial condition $\phi_{0}^{(\delta)}$.

We now assume that $\langle V_{\theta} \rangle$ satisfies a pitchfork bifurcation at $(0,0)$. Then, by the conditions \eqref{eq:bifurcation} and \eqref{eq:pitchfork}, the coefficients $\ha_{j}^{(h)}(\theta)$ on the right hand side of \eqref{eq:asymmetric_delta} satisfy
\begin{equation} \label{eq:coefficients_asymptotic}
\begin{split}
\ha_{j}^{(h)}(\theta) &= \ha_{j} + \oO(\theta, h), \qquad j \geq 3; \\
\ha_{2}^{(h)}(\theta) &= 3 \ha_{3} h + \oO(\theta, h^{2}); \\
\ha_{1}^{(h)}(\theta) &= 3 \ha_{3} h^{2} + \ha_{1}' \theta + \oO(\theta^{2}, \theta h, h^{3}); \\
\ha_{0}^{(h)}(\theta) &= \ha_{3} h^{3} + \ha_{1}' \theta h + \frac{\ha_{0}''}{2} \cdot \theta^{2} + \oO(\theta^{3}, \theta h^{2}, h^{4}), 
\end{split}
\end{equation}
As already mentioned in the introduction, whether one could obtain $\Phi^4_3$ in the large scale limit depends on whether the quantity $A$ defined in \eqref{eq:quantity_a} is $0$. In the case $A = 0$, we have the following theorem.

\begin{thm} \label{th:main_symmetric}
Let $A=0$. If we set $\alpha = 1$, $h = 0$, and
\begin{align*}
\theta = \theta(\epsilon) = \frac{18 \ha_{3}^{2} c_{2}}{\ha_{1}'} \cdot \epsilon |\log \epsilon| + \lambda \epsilon + o(\epsilon), 
\end{align*}
then, for any fixed $T> 0$, $u_{\epsilon}$ converges in probability in $\cC([0,T],\cC^\eta(\TT^{3}))$ to the $\Phi^4_3(\ha_{3})$ family of solutions indexed by $\lambda$ with initial condition $\phi_0$. 
\end{thm}
\begin{proof}
Since $\alpha = 1$, we actually have $\epsilon = \delta$. From \eqref{eq:coefficients_asymptotic}, we immediately deduce that
\begin{align*}
\lambda_{j}^{(\epsilon)} \rightarrow \ha_{j}, \qquad j \geq 3. 
\end{align*}
Since $h = 0$, we have $\ha_{2}(\theta) \sim \epsilon \log \epsilon$, which gives $\lambda_{2}^{(\epsilon)} \sim \epsilon^{\frac{1}{2}} \log \epsilon \rightarrow 0$. For $\lambda_{0}^{(\epsilon)}$, we have
\begin{align*}
\ha_{0}(\theta) \sim \epsilon^{2} \log^{2} \epsilon, \qquad \lambda_{2}^{(\epsilon)} = \oO(\epsilon^{\frac{1}{2}} \log \epsilon), 
\end{align*}
so the only problematic term is $C'_{\epsilon, \theta}$. But note that $A = 0$, this term also vanishes, so we also have $\lambda_{0}^{(\epsilon)} \rightarrow 0$. 

We now turn to $\lambda_{1}^{(\epsilon)}$. Note that both $\ha_{1}(\theta) \cdot \epsilon^{-1}$ and $C_{\epsilon,\theta}$ diverge logarithmically, but the prefactor of the term $\epsilon |\log \epsilon|$ in $\theta$ guarantees that these two divergent terms cancel each other, so $\lambda_{1}^{(\epsilon)}$ converges to some finite quantity $\lambda_{1}$, depending on the choice $\lambda$ in front of the $\epsilon$ term in $\theta$. This implies that when restricted to basis vectors without an appearance of $\eE$, the formal right hand side of \eqref{eq:abstract_limit} is identical as that of $\Phi^4_3(\ha_{3})$ with a proper linear term. 

Since the action of the model $\fM$ on basis vectors without an appearance of $\eE$ are precisely the same as the limiting model in $\Phi^4_3$, and its action on symbols with $\eE$ yields $0$, it then follows 
that $u = \widehat{\rR} \Phi$ for the limiting equation does coincide with the $\Phi^4_3(\ha_{3})$ 
family indexed by $\lambda$. 
Recall that we assumed that $\scal{V_\theta}$ undergoes a pitchfork bifurcation and that, in particular,
this implies by \eqref{eq:pitchfork} that $\ha_3 > 0$. It then follows from the results in
\cite{Konstantin,Phi4_global} that, for any initial condition in $\cC^\eta(\T^3)$, this limit almost surely admits solutions 
globally in time.
We can therefore apply Theorem~\ref{th:fixed_pt}, which yields the desired convergence,
 thus completing the proof. 
\end{proof}

We now turn to the non-symmetric case where $A \neq 0$. Since changing $u$ to $-u$ and $h_{\epsilon}$ to $-h_{\epsilon}$
has the effect of simply turning $A$ into $-A$, we can assume without loss of generality that $A > 0$. 
Before stating our result in this case, we introduce a way of comparing trajectories up to a possible explosion
time. Consider the set $\xX^\eta = \cC_\star (\R_+, \cC^{\eta}(\TT^{3}))$ of pairs $(\Phi,T_\star)$ where $T_\star > 0$
and $\Phi\in \cC ([0,T_\star), \cC^{\eta}(\TT^{3}))$. We introduce a family of ``distances'' (which however fail to be symmetric!) indexed by $K,T > 0$
in the following way. For $K,T > 0$ and elements $(\bar \Phi,\bar T_\star),(\Phi,T_\star) \in \xX^\eta$,
we set
\begin{equs}
\tau &= T \wedge \inf\{t \in [0,T_\star)\,:\, \|\Phi(t)\|_\eta \ge K\}\;,\\
\bar \tau &= T \wedge \inf\{t \in [0,\bar T_\star)\,:\, \|\bar \Phi(t)\|_\eta \ge K+1\}\;.
\end{equs}
We then set
\begin{equ}[e:funnydistance]
d_\star^{K,T} \bigl((\bar \Phi,\bar T_\star),(\Phi,T_\star)\bigr)
= |\tau - (\tau \wedge \bar \tau)| + \sup_{t \le \tau \wedge \bar \tau}\|\bar \Phi(t)- \Phi(t)\|_\eta\;.
\end{equ}
The reason for this somewhat asymmetric definition is that we want to consider $(\bar \Phi,\bar T_\star)$
as being ``close to'' $(\Phi,T_\star)$ even in situation where $\Phi$ explodes at time $T_\star$,
but $\bar \Phi$ merely gets very large at that time and then explodes at some much later time. 
On the other hand, we \textit{do not} want to allow the converse situation.
Given a sequence of random elements $(\Phi^{(\eps)},T_\star^{(\eps)}) \in \xX^\eta$,
we say that it converges in law to a (random) limit $(\Phi,T_\star) \in \xX^\eta$ if, for every $K,T > 0$
 and every $\delta > 0$, there exists $\bar \eps > 0$ and a coupling between these random variables such that,
for $\eps < \bar \eps$, one has $\PP\bigl(d_\star^{K,T} \bigl((\Phi^{(\eps)}, T_\star^{(\eps)}),(\Phi,T_\star)\bigr) > \delta\bigr) < \delta$.

The main statement is the following.

\begin{thm} \label{th:main_asymmetric}
Let $A > 0$, and assume $\theta = \rho \epsilon^{\beta}$ near the origin ($\rho>0$). 
\begin{enumerate}

\item If $\beta < \frac{2}{3}$, then there exists three distinct choices $h_{\epsilon}^{(1)} < h_{\epsilon}^{(2)} < h_{\epsilon}^{(3)}$ such that at scale $\alpha = \frac{1 + \beta}{2}$, both $u_{\delta}^{(1)}$ and $u_{\delta}^{(3)}$ converges in probability to $u$ while $u_{\delta}^{(2)}$ converges in probability to $v$, where $u$ and $v$ solves the equations
\begin{align*}
\partial_{t} u = \Delta u - 2  |\ha_{1}'| \rho u + \xi, \qquad \partial_{t} v = \Delta v +  |\ha_{1}'| \rho v + \xi, 
\end{align*}
respectively, both with initial data $\phi_{0}$. 

\item If $\beta > \frac{2}{3}$, then there exists a unique choice $h_{\epsilon}$ such that at scale $\alpha = \frac{5}{6}$, the process $u_{\delta}$ converges in probability to the solution $u$ of the equation
\begin{align*}
\partial_{t} u = \Delta u - 3 \big( \ha_{3} A^{2} \big)^{\frac{1}{3}} u + \xi
\end{align*}
with initial data $\phi_{0}$. 

\item If $\beta = \frac{2}{3}$ and $\theta = \rho \epsilon^{\frac{2}{3}}$, then there exists a critical value
\begin{equation} \label{eq:critical_value}
\rho^{*} = \frac{3}{|\ha_{1}'|} \cdot \bigg( \frac{\ha_{3} A^{2}}{4}\bigg)^{\frac{1}{3}}
\end{equation}
such that for $\rho < \rho^{*}$ (and resp. $\rho > \rho^{*}$) there exist one (and three, resp.) choices of $h$ such that at scale $\alpha = \frac{5}{6}$, $u_{\epsilon^{\alpha}}$ converges to one or three distinct O.U. processes. 

For $\rho = \rho^{*}$, there exist two distinct choices $h_{\epsilon}^{(1)} < h_{\epsilon}^{(2)}$ such that for $h = h_{\epsilon}^{(1)}$, at scale $\alpha = \frac{8}{9}$, $u_{\epsilon^{\alpha}}$ converges to the solution $u$ of the equation
\begin{align*}
\partial_{t} u = \Delta u + 3 \bigg( \frac{\ha_{3}^{2} A}{2} \bigg)^{\frac{1}{3}} \Wick{u^{2}} + \xi, 
\end{align*}
while for $h = h_{\epsilon}^{(2)}$, $u_{\epsilon^{\alpha}}$ still converges to O.U. at scale $\alpha = \frac{5}{6}$. 
\end{enumerate}
All the convergences above are convergences in law in $\xX^\eps$, with $T_\star$ (resp. $T_\star^{(\eps)}$)
given by the explosion times of the respective processes in $\cC^\eta$. 
\end{thm}

\begin{rmk}
The situation for $\beta = \frac{2}{3}$ and $\rho < \rho^{*}$ (or $\rho > \rho^{*}$) are similar to that of $\beta > \frac{2}{3}$ (or $\beta < \frac{2}{3}$), except that the coefficients in the limiting equations are different. The other difference is that in the case $\rho > \rho^{*}$, the three choices of $h$ gives three different limiting O.U. processes, unlike when $\beta < \frac{2}{3}$, two of the three $h$'s gives the same limiting equation. Also note that the limiting equation $\Phi^3_3$ does not have a global solution even in the case when the highest power of $V_{0}$ has a positive coefficient. 
\end{rmk}

We will give the proof of the above theorem for the most interesting case $\beta = \frac{2}{3}$, and the proof for the other two situations are essentially the same but only simpler. We will make use of the following elementary lemma.

\begin{lem}
Let $A > 0$. For any $\rho > 0$, let $f_{\rho}(r) = \ha_{3} r^{3} + \rho \ha_{1}' r - A$. Let $\rho^{*}$ be the same as in \eqref{eq:critical_value}. Then, the equation $f_{\rho}(r) = 0$ has one, two, or three distinct real roots for $\rho < \rho^{*}$, $\rho = \rho^{*}$ and $\rho > \rho^{*}$, respectively. In particular, for $\rho = \rho^{*}$, the two roots $r_{1} < r_{2}$ satisfy
\begin{align*}
r_{1} = - \bigg( \frac{A}{2 \ha_{3}} \bigg)^{\frac{1}{3}}, \qquad r_{2} > \bigg( \frac{A}{2 \ha_{3}} \bigg)^{\frac{1}{3}}. 
\end{align*}
\end{lem}
\begin{proof}
If $f_{\rho}(r) = 0$ has exactly two distinct roots, then since $A > 0$, the smaller one must also be a local maximum for $f_{\rho}$. The value of $\rho^{*}$ and that root could then be computed directly, and all other assertions follow. 
\end{proof}

\begin{flushleft}
\textbf{Proof of Theorem~\ref{th:main_asymmetric}. }
\end{flushleft}

We only give details to the case when $\beta = \frac{2}{3}$ so $\theta = \rho \epsilon^{\frac{2}{3}}$. For $\rho = \rho^{*}$, let $r_{1} < r_{2}$ be the two roots to the equation $f_{\rho^{*}}(r) = 0$, and set
\begin{align*}
\alpha_{1} = \frac{8}{9}, \qquad \theta \sim \rho^{*} \delta_{1}^{\frac{3}{4}}, \qquad h_{\delta}^{(1)} = r_{1} \delta_{1}^{\frac{3}{8}} = r_{1} \epsilon^{\frac{1}{3}}; \\
\alpha_{2} = \frac{5}{6}, \qquad \theta \sim \rho^{*} \delta_{2}^{\frac{4}{5}}, \qquad h_{\delta}^{(2)} = r_{2} \delta_{1}^{\frac{2}{5}} = r_{2} \epsilon^{\frac{1}{3}}, 
\end{align*}
For the choice of $(\alpha_{1}, h^{(1)})$, we deduce from the properties of $\ha_{j}^{(h)}(\theta)$'s that $\lambda_{j}^{(\delta_{1})} \rightarrow 0$ for all $j \neq 2$, while
\begin{align*}
\lambda_{2}^{(\delta_{1})} \rightarrow -3 \bigg( \frac{\ha_{3}^{2} A}{2} \bigg)^{\frac{1}{3}}. 
\end{align*}
The claim then follows immediately. For the choice $(\alpha_{2}, h^{(2)})$, we have $\lambda_{j}^{(\delta_{2})} \rightarrow 0$ for all $j \neq 1$ and $\lambda_{1}^{(\delta_{2})}$ converges to some positive real number. Thus, the limiting process in this case is O.U.. 

For $\rho < \rho^{*}$ and $\rho > \rho^{*}$, one should note that there exist one (or three, respectively) distinct real solutions to the equation
\begin{align*}
f_{\rho}(r) = 0. 
\end{align*}
By setting $\alpha = \frac{5}{6}$ and $h_{\delta} = r \delta^{\frac{2}{5}} = r \epsilon^{\frac{1}{3}}$ with the roots $r$, one can show that all $\lambda_{j}^{(\delta)}$'s vanish in the limit except $\lambda_{1}^{(\delta)}$ which converges to a finite quantity. The form of the limiting equation then follows immediately. The coefficient of the drift term can be found by computing the roots to $f_{\rho}(r) = 0$, but this is not important here. This completes the proof.

\begin{rmk} \label{rm:chart}
	One can also adjust $\theta$ to the second order. In fact, for
	\begin{align*}
	\theta = \rho_{1} \epsilon^{\beta_1} + \rho_{2} \epsilon^{\beta_2}
	\end{align*}
	with $\beta_1 = \frac{2}{3}$ and $\rho_1 = \rho^{*}$, it is not difficult to show that if $\beta_2 < \frac{8}{9}$ and $\rho_{2} > 0$, then one still gets three OU's, but two of them are observed at larger scales than $\frac{5}{6}$. If $\beta_2 \geq \frac{8}{9}$, then one can get $\Phi^3_3$. This can be illustrated by the following figure. 
	
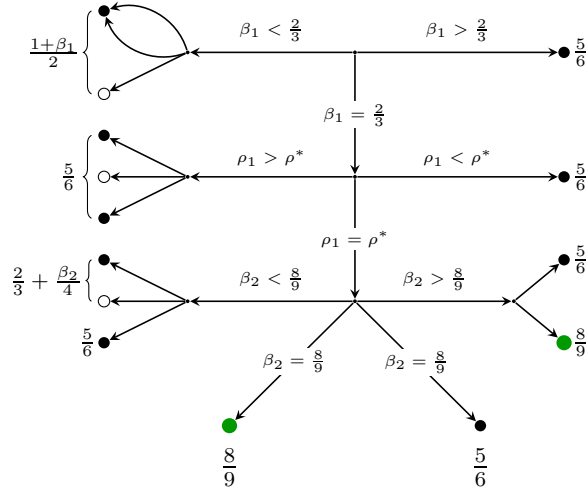
\begin{align*}
	\begin{tikzpicture}[scale=1.1,baseline=0cm]
	\node at (0,0) [sdot] (r1) {};
	\node at (-2,0) [sdot] (r1left) {}; 
	\node at (2.5,0) [sou] (r1right) {}; 
	\node at (0,1.5) [sdot] (b1) {}; 
	\node at (-2,1.5) [sdot] (b1left) {}; 
	\node at (2.5,1.5) [sou] (b1right) {}; 
	\node at (-3,2) [sou] (b1s) {}; 
	\node at (-3,1) [nou] (b1n) {}; 
	\node at (-3,0.5) [sou] (r1s1) {}; 
	\node at (-3,0) [nou] (r1n) {}; 
	\node at (-3,-0.5) [sou] (r1s2) {}; 
	\node at (2.7,1.5) [] () {\scriptsize $\frac{5}{6}$}; 
	\node at (2.7,0) [] () {\scriptsize $\frac{5}{6}$}; 
	\node at (0,-1.5) [sdot] (b2) {}; 
	\node at (1.9,-1.5) [sdot] (b2right) {}; 
	\node at (-2,-1.5) [sdot] (b2left) {}; 
	\node at (-3,-1) [sou] (b2s1) {}; 
	\node at (-3,-1.5) [nou] (b2n) {}; 
	\node at (-3,-2) [sou] (b2s2) {}; 
	\node at (2.5,-1) [sou] (b2s) {}; 
	\node at (2.5,-2) [root] (b2wick) {}; 
	\node at (2.7,-1) [] () {\scriptsize $\frac{5}{6}$}; 
	\node at (2.7,-2) [] () {\scriptsize $\frac{8}{9}$}; 
	\node at (-3.2,-2) [] () {\scriptsize $\frac{5}{6}$}; 
	\node at (-1.5,-3) [root] (ou) {}; 
	\node at (1.5,-3) [sou] (saddle) {}; 
	\node at (-1.5,-3.5) [] () {$\frac{8}{9}$}; 
	\node at (1.5,-3.5) [] () {$\frac{5}{6}$}; 
	\draw[akernel] (b1) to node[anchor = south]{\tiny $\beta_1<\frac{2}{3}$} (b1left); 
	\draw[akernel] (b1) to node[anchor = south]{\tiny $\beta_1>\frac{2}{3}$} (b1right); 
	\draw[akernel] (b1) to node[labl]{\tiny $\beta_1=\frac{2}{3}$} (r1); 
	\draw[akernel] (b1left) to (b1n); 
	\draw[akernel] (r1) to node[anchor = south]{\tiny $\rho_1 > \rho^*$} (r1left); 
	\draw[akernel] (r1) to node[anchor = south]{\tiny $\rho_1 < \rho^*$} (r1right); 
	\draw[akernel] (r1left) to (r1s1); 
	\draw[akernel] (r1left) to (r1s2); 
	\draw[akernel] (r1left) to (r1n); 
	\draw[akernel, bend left = 45] (b1left) to (b1s); 
	\draw[akernel, bend right = 45] (b1left) to (b1s); 
	\draw [decorate,decoration={brace,amplitude=3pt}] (-3.15,1) to node[midway, xshift=-0.5cm] {\scriptsize $\frac{1+\beta_1}{2}$} (-3.15,2); 
	\draw [decorate,decoration={brace,amplitude=3pt}] (-3.15,-0.5) to node[midway, xshift=-0.3cm] {\scriptsize $\frac{5}{6}$} (-3.15,0.5); 
	\draw[akernel] (r1) to node[labl]{\tiny $\rho_1 = \rho^{*}$} (b2); 
	\draw[akernel] (b2) to node[anchor = south]{\tiny $\beta_2<\frac{8}{9}$} (b2left); 
	\draw[akernel] (b2) to node[anchor = south]{\tiny $\beta_2>\frac{8}{9}$} (b2right); 
	\draw[akernel] (b2left) to (b2s1); 
	\draw[akernel] (b2left) to (b2s2); 
	\draw[akernel] (b2left) to (b2n); 
	\draw[akernel] (b2right) to (b2s); 
	\draw[akernel] (b2right) to (b2wick); 
	\draw [decorate,decoration={brace,amplitude=2pt}] (-3.15,-1.5) to node[midway, xshift=-0.6cm] {\scriptsize $\frac{2}{3} + \frac{\beta_2}{4}$} (-3.15,-1); 
	\draw[akernel] (b2) to node[labl]{\tiny $\beta_2 = \frac{8}{9}$} (ou); 
	\draw[akernel] (b2) to node[labl]{\tiny $\beta_2 = \frac{8}{9}$} (saddle); 
	\end{tikzpicture}
	\end{align*}
	Here, each \tikz[baseline=-0.1cm] \node [sou] {}; represents a stable OU process (the one with two arrows pointing to it indicates that the two limiting OU processes have the same coefficient), each \tikz[baseline=-0.1cm] \node [nou] {}; represents an unstable OU process, and the green node \tikz[baseline=-0.1cm] \node [root] {}; represents a $\Phi^3_3$ equation. The difference between the two green dots are that the limit represented by the one at the bottom represents a $\Phi^3_3$ family parametrised by the coefficient $\rho_2$, while the one on the right has the canonical Wick product meaning. Finally, the numbers next to each node indicates the scale $\alpha$. 
	
	Note that for the branch containing the saddle-node bifurcation, the scale increases from $\frac{1}{2}$ to $\frac{8}{9}$ continuously with respect to the exponents $(\beta_1, \beta_2)$. One can also obtain such a complete description for the symmetric case, but we omit the statement of the details for conciseness. Also, the reason that the three nodes \tikz[baseline=-0.1cm] \node [sou] {}; on the right of the figure all exhibit scale $\frac{5}{6}$ is that we only include the case when $\rho_{1} > 0$. In fact, one can also recover the scales from $\frac{1}{2}$ to $\frac{5}{6}$ continuously by considering $\beta_{1} < \frac{2}{3}$ and $\rho_{1} < 0$. 
\end{rmk}

\begin{rmk}
We now very briefly discuss the case when $\langle V_{\theta} \rangle$ has a stable extreme point or a saddle-node bifurcation near the origin. The proofs are much simpler than the pitchfork bifurcation case, so we do not give details here. In both cases, no re-centering is needed so $h = 0$. 

If $\langle V \rangle$ has a stable extreme point at the origin, then $\ha_{1} \neq 0$. In this case, we choose $\alpha = \frac{1}{2}$ (so $\delta = \epsilon^{\frac{1}{2}}$). Since we always assume $\ha_{0} = 0$, then as long as $\theta = o(\epsilon)$, all $\lambda_{j}^{(\delta)}$'s vanish in the limit except $\lambda_{1}^{(\delta)} \rightarrow \ha_{1}$. Thus, the process $u_{\delta}$ converges in probability to the limit
\begin{align*}
\partial_{t} u = \Delta u - \ha_{1} u + \xi. 
\end{align*}
In the case of saddle-node bifurcation when $\ha_{0} = \ha_{1} = 0$ but $\ha_{2} \neq 0$, the correct scale here should be $\alpha = \frac{2}{3}$. Then, as long as $\theta = o(\delta) = o(\epsilon^{\frac{3}{2}})$, all $\lambda_{j}^{(\delta)} \rightarrow 0$ except for $\lambda_{2}^{(\delta)}$ which converges to $\ha_{2}$. This gives the limiting equation
\begin{align*}
\partial_{t} u = \Delta u - \ha_{2} \Wick{u^2} + \xi. 
\end{align*}
If $\theta = \oO (\epsilon^{\frac{3}{2}})$, then the resulting limit is a $\Phi^3_3$ family. Note that in the above two cases, no further renormalisation is needed beyond the usual Wick ordering, so they can actually be treated using the methods developed in \cite{DPD03} and \cite{EJS13}. 
\end{rmk}

\subsection{Weak noise regime}

We now consider the weak noise regime. Here, we assume $V: \theta \mapsto V_{\theta}(\cdot)$ is smooth in $\cC^{8}$ functions so that it can be expanded near $x=0$ as in \eqref{eq:uniform_regular}. We also assume that $V$ has a pitchfork bifurcation near the origin in the sense of \eqref{eq:pitchfork_noise}. Let $\tilde{u}$ be the process satisfying
\begin{align*}
\partial_{t} \tilde{u} = \Delta \tilde{u} - V_{\theta}'(\tilde{u}) + \epsilon^{\frac{1}{2}} \widehat{\xi}, 
\end{align*}
and define $u_{\epsilon^{\alpha}}$ to be
\begin{align*}
u_{\epsilon^{\alpha}} = \epsilon^{-\frac{1+\alpha}{2}} \big( \tilde{u}(t/\epsilon^{2 \alpha}, x/\epsilon^{\alpha}) - h \big). 
\end{align*}
By setting $\delta = \epsilon^{\alpha}$, we see that $u_{\delta}$ satisfies the equation
\begin{equation} \label{eq:noise_macro}
\partial_{t} u_{\delta} = \Delta u_{\delta} - \sum_{j=0}^{6} a_{j}^{(h)}(\theta) \delta^{\frac{j-1}{2\alpha} + \frac{j-5}{2}} u_{\delta}^{j} - \delta^{-\frac{1}{2\alpha} - \frac{5}{2}} F_{\theta,h}(\delta^{\frac{1}{2\alpha} + \frac{1}{2}} u_{\delta}) + \xi_{\delta}
\end{equation}
for certain function $F_{\theta,h}$ satisfying $|F_{\theta,h}(x)| \lesssim |x|^{7}$ uniformly over $|\theta|, |h|, |x| < 1$, and the coefficients $a_{j}^{(h)}$'s satisfy
\begin{equation} \label{eq:noise_a's}
a_{j}^{(h)}(\theta) = \sum_{k=j}^{6} a_{k}(\theta) \begin{pmatrix} k \\ j \end{pmatrix} \cdot h^{k-j} + \oO(h^{7-j}), \qquad 0 \leq j \leq 6. 
\end{equation}
Similar as before, we always assume \eqref{eq:noise_macro} starts with initial data $\phi_{0}^{(\delta)}  \in \cC^{\gamma,\eta}_{\delta}$ such that $\| \phi_{0}^{(\delta)}; \phi_{0} \|_{\gamma,\eta;\delta} \rightarrow 0$ for some $\phi_{0} \in \cC^{\eta}$. 

We still let $\fM_{\delta} = M_{\delta} \sL_{\delta}(\xi_{\delta})$ be the renormalised model as before, $\dD^{\gamma,\eta}_{\delta}$ and $\widehat{\rR}^{\delta}$ be the associated space and reconstruction operator, and consider the abstract fixed point equation
\begin{equation} \label{eq:abstract_noise}
\begin{split}
\Phi^{(\delta)} = &\pP \1_{+} \bigg( \Xi - \sum_{j=4}^{6} \lambda_{j}^{(\delta)} \qQ_{\leq 0} \heE^{\frac{j-3}{2}} \qQ_{\leq 0} \big( (\Phi^{(\delta)})^{j} \big) - \sum_{j=0}^{3} \lambda_{j}^{(\delta)} \qQ_{\leq 0} \big( (\Phi^{(\delta)})^{j} \big) \\
&- \delta^{-\frac{1}{2\alpha} - \frac{5}{2}} F_{\theta,h}(\delta^{\frac{1}{2\alpha} + \frac{1}{2}} \widehat{\rR}^{\delta} \Phi^{\delta}) \cdot \1 \bigg) + \widehat{P} \phi_{0}^{(\delta)}. 
\end{split}
\end{equation}
Here, we allow the parameters $\theta$ and $h$ to depend on $\delta$, which is indeed the case we consider later. The following statement is an analogy to Theorem~\ref{th:fixed_pt}. It will be crucial to proving the convergence of $u_{\delta}$ to corresponding limits in various situations.

\begin{thm} \label{th:abstract_noise}
Let $\fM_{\delta} \in \sM_{\delta}$ and $\fM \in \sM$ be as before, and let $\alpha \leq 1$. Suppose $|F_{\theta,h}(x)| \lesssim |x|^{7}$ near the origin uniformly over $|\theta|, |h| < 1$, and suppose for each $j$, there exists $\lambda_{j} \in \R$ such that $\lambda_{j}^{(\delta)} \rightarrow \lambda_{j}$. Then, there exists a short existence time $T$ such that there is a unique fixed point solution $\Phi \in \dD^{\gamma,\eta}$ to the equation
\begin{align*}
\Phi = \pP \1_{+} \bigg( \Xi - \sum_{j=4}^{6} \lambda_{j} \qQ_{\leq 0} \heE^{\frac{j-3}{2}} \big( \qQ_{\leq 0}(\Phi^{j}) \big) - \sum_{j=0}^{3} \lambda_{j}  \qQ_{\leq 0} (\Phi^{j}) \bigg) + \widehat{P} \phi_{0}. 
\end{align*}
Furthermore, for every small enough $\delta$, there also exists a fixed point solution $\Phi^{(\delta)} \in \dD^{\gamma,\eta}_{\delta}$ to \eqref{eq:abstract_noise} up to the same time $T$ such that
\begin{align*}
	\lim_{\delta \rightarrow 0}  \| \Phi^{(\delta)}; \Phi  \|_{\gamma,\eta;\delta} = 0\;,\quad
	\lim_{\delta \to 0} \sup_{t \in [0,T]} \|(\rR^{(\delta)} \Phi^{(\delta)})(t,\cdot)-(\rR \Phi)(t,\cdot)\|_\eta = 0\;. 
\end{align*}
\end{thm}
\begin{proof}
In view of Theorem~\ref{th:fixed_pt}, it suffices to prove that the map (up to some fixed time $S$)
\begin{equation} \label{eq:smooth}
\Phi^{(\delta)} \mapsto \delta^{-\frac{1}{2\alpha} - \frac{5}{2}} \pP \1_{+} \big( F_{\theta,h}(\delta^{\frac{1}{2\alpha} + \frac{1}{2}} \widehat{\rR}^{\delta} \Phi^{\delta}) \cdot \1 \big)
\end{equation}
is locally Lipschitz from $\dD^{\gamma,\eta}_{\delta}$ to itself with a Lipschitz constant bounded by $\delta^{\sigma}$ for some positive $\sigma$, uniformly over $\theta$ and $h$. We need this uniformity because of the dependence of $\theta$ and $h$ on $\delta$ in \eqref{eq:abstract_noise}. 

To see \eqref{eq:smooth}, we first note that if $\Phi$ solves the fixed point equation \eqref{eq:abstract_noise}, then it necessarily has the form
\begin{align*}
\Phi = \Psi + U(z), 
\end{align*}
where $U$ takes value in a subspace of $\tT$ spanned by $\1$ and elements with strictly positive homogeneities. As a consequence, we have
\begin{align*}
(\widehat{\rR}^{\delta} U)(z) = \langle U(z), \1 \rangle \lesssim (\delta + \sqrt{|t|})^{\eta} \|U\|_{\gamma,\eta;\delta}. 
\end{align*}
It is also straightforward to show that
\begin{align*}
| (\widehat{\rR} \Psi) (z) | = | (K * \xi_{\delta})(z) | \lesssim \delta^{-\frac{1}{2}-\kappa} \|\fM_{\delta}\|_{\delta}. 
\end{align*}
Thus, combining the above two bounds together with the assumption of the behavior of $F$ around $0$, we deduce that the map
\begin{align*}
\Phi^{(\delta)} \mapsto \delta^{-\frac{1}{2\alpha} - \frac{5}{2}}  F_{\theta,h}(\delta^{\frac{1}{2\alpha} + \frac{1}{2}} \widehat{\rR}^{\delta} \Phi^{\delta}) \cdot \1 
\end{align*}
is locally Lipschitz continuous from $\dD^{\gamma,\eta}_{\delta}$ to the space of continuous functions $\cC$ with uniform topology, and that the local Lipschitz constant is proportional to $\delta^{\sigma}$ for some $\sigma > 0$ (independent of $\theta$ and $h$). The additional operation by $\pP \1_{+}$ (up to time $S$) makes the map \eqref{eq:smooth} locally Lipschitz from $\dD^{\gamma,\eta}_{\delta}$ to itself, and the Lipschitz constant is bounded by $(S \delta)^{\sigma}$. 

The rest of the proof follows in the same line as that in Theorem~\ref{th:fixed_pt}. 
\end{proof}

Suppose we have now chosen $\lambda_{j}^{(\delta)}$'s such that $\widehat{\rR}^{\delta} \Phi^{(\delta)}$ exactly solves \eqref{eq:noise_macro}. By the assumptions on the models and initial conditions, Theorem~\ref{th:abstract_noise} guarantees that as long as we can show that these $\lambda_{j}^{(\delta)}$'s converge to the desired limiting values, then the convergence of $u_{\delta}$ to the limiting process with follow automatically as in the previous section. 

Inspecting the right hand side of \eqref{eq:renormalised_equation}, we see that in order for $\widehat{\rR}^{\delta} \Phi^{(\delta)}$ to solve \eqref{eq:noise_macro}, we need to set $\lambda_{j}^{(\delta)}$'s in the following way: 
\begin{equation} \label{eq:noise_coefficients}
\begin{split}
\lambda_{6}^{(\delta)} &= a_{6}^{(h)}(\theta) \cdot \delta^{\frac{5}{2\alpha}-1}, \qquad \lambda_{5}^{(\delta)} = a_{5}^{(h)}(\theta) \cdot \delta^{\frac{2}{\alpha} - 1}; \\
\lambda_{4}^{(\delta)} &= \delta^{\frac{3}{2\alpha}-1} \big( a_{4}^{(h)}(\theta)  + 15 a_{6}^{(h)}(\theta) C_{0} \cdot \delta^{\frac{1}{\alpha}} \big); \\
\lambda_{3}^{(\delta)} &=  \delta^{\frac{1}{\alpha}-1} \big( a_{3}^{(h)}(\theta) + 10 a_{5}^{(h)}(\theta) C_{0} \cdot \delta^{\frac{1}{\alpha}} \big); \\
\lambda_{2}^{(\delta)} &= \delta^{\frac{1}{2\alpha}-\frac{3}{2}} \big( a_{2}^{(h)}(\theta)  + 6 a_{4}^{(h)}(\theta) C_{0} \cdot \delta^{\frac{1}{\alpha}} + 45 a_{6}^{(h)}(\theta) C_{0}^{2} \cdot \delta^{\frac{2}{\alpha}} \big); \\
\lambda_{1}^{(\delta)} &= \delta^{-2} \big( a_{1}^{(h)}(\theta) + 3 a_{3}^{(h)}(\theta) C_{0} \cdot \delta^{\frac{1}{\alpha}} + 15 a_{5}^{(h)}(\theta) C_{0}^{2} \cdot \delta^{\frac{2}{\alpha}} \big) - C_{\delta}; \\
\lambda_{0}^{(\delta)} &= \delta^{-\frac{1}{2\alpha}-\frac{5}{2}} \big( a_{0}^{(h)}(\theta) + a_{2}^{(h)}(\theta) C_{0} \cdot \delta^{\frac{1}{\alpha}} + 3 a_{4}^{(h)}(\theta) C_{0}^{2} \cdot \delta^{\frac{2}{\alpha}} + 15 a_{6}^{(h)}(\theta) C_{0}^{3} \cdot \delta^{\frac{3}{\alpha}} \big) \\
&\phantom{11} - C_{\delta}' - 6 \lambda_{2}^{(\delta)} \lambda_{3}^{(\delta)} C_{2}^{(\delta)}, 
\end{split}
\end{equation}
where the constants $C_{\delta}$ and $C_{\delta}'$ are given by
\begin{align*}
C_{\delta} &= \sum_{n=2}^{5} (n+1)^{2} n! \cdot (\lambda_{n+1}^{(\delta)})^{2} C_{n}^{(\delta)} + \sum_{n=3}^{4} (n+2)! \cdot \lambda_{n}^{(\delta)} \lambda_{n+2}^{(\delta)} C_{n}^{(\delta)}; \\
C_{\delta}' &= \delta^{-\frac{1}{2}} \sum_{n=3}^{5} (n+1)! \cdot \lambda_{n}^{(\delta)} \lambda_{n+1}^{(\delta)} C_{n}^{(\delta)}
\end{align*}
The additional term $F_{\theta,h} \cdot \1$ in \eqref{eq:abstract_noise} does not affect the choice as it precisely gives the corresponding term in \eqref{eq:noise_macro} when hit with the reconstruction operator. The following statement gives the situation where we can observe $\Phi^4_3$.

\begin{thm} \label{th:noise_symmetric}
Suppose
\begin{equation} \label{eq:quantity_b}
B = a_{4}  + \frac{3 a_{0}'' a_{3}^{2}}{2 a_{1}'^{2}} - \frac{a_{2}'a_{3}}{a_{1}'} = 0. 
\end{equation}
Then, there exists
\begin{align*}
\theta(\epsilon) = - \frac{3 a_{3} C_{0}}{a_{1}'} \cdot \epsilon + \frac{18 a_{3}^{2} c_{2}}{a_{1}'} \cdot \epsilon^{2} |\log \epsilon| + \lambda \epsilon^{2}, \qquad h(\epsilon) = \rho_{1} \epsilon + \rho_{2} \epsilon^{2}, 
\end{align*}
such that at scale $\alpha = 1$, the solution $u_{\epsilon}$ to \eqref{eq:noise_macro} with initial condition $\phi_{0}^{(\epsilon)}$ converges in probability in $\cC\big([0,T], \cC^{\eta}(\TT^{3})\big)$ for every $T>0$ to the $\Phi^4_3 (a_{3})$ family (with initial data $\phi_{0}$) indexed by $\lambda$. Here, $\rho_{1}$ depends on $C_{0}$ and the coefficients $a_{j}$'s, and $\rho_{2}$ is chosen depending on $\lambda$. 
\end{thm}
\begin{proof}
At $\alpha = 1$, we have $\delta = \epsilon$. It is easy to see that if $B = 0$, then with the above choice of $\theta$, all $\lambda_{j}^{(\epsilon)}$'s converge to a finite limit. In particular, we have
\begin{align*}
\lambda_{j}^{(\epsilon)} \rightarrow 0 \phantom{1} (j \geq 4), \qquad \lambda_{3}^{(\epsilon)} \rightarrow a_{3}, \qquad \lambda_{2}^{(\epsilon)} \rightarrow \lambda_{2} = - \frac{3 a_{2}' a_{3} C_{0}}{a_{1}'} + 3 a_{3} \rho_{1} + 6 a_{4} C_{0}. 
\end{align*}
Since $a_3 \neq 0$, we can choose $\rho_1$ such that $\lambda_2 = 0$. For $\lambda_{0}^{(\epsilon)}$, it is straightforward to show that it converges to a finite limiting $\lambda_{0}$ whose value depends on $\lambda$ and $\rho_2$. Since $\rho_2$ is multiplied by $a_3$ which is non-zero, one can then choose $\rho_2$ to make $\lambda_{0}$ vanish. The assertion then follows from \cite{Konstantin, Phi4_global}, Theorem~\ref{th:abstract_noise} and the continuity of the reconstruction operators. 
\end{proof}


In the case when $B \neq 0$, we need to look at a different scale to observe a non-trivial limit. The value of $\theta$ at which one sees a saddle-node bifurcation turns out to be
\begin{align*}
\theta^{*}(\epsilon) = \rho_{1}^{*} \epsilon + \rho_{2}^{*} \epsilon^{\frac{4}{3}} + \rho_{3}^{*} \epsilon^{\frac{5}{3}} + \oO(\epsilon^{\frac{16}{9}})
\end{align*}
with
\begin{equation} \label{eq:rho's}
\rho_{1}^{*} = \frac{3 a_{3} C_{0}}{|a_{1}'|}, \quad \rho_{2}^{*} = \frac{9}{(12)^{1/3} |a_{1}'|} (a_{3} B^{2} C_{0}^{4})^{\frac{1}{3}}, \quad \rho_{3}^{*} = 2 B C_{0} \bigg(\frac{3 \rho_{2}^{*}}{|a_{1}'| a_{3}} \bigg)^{\frac{1}{2}}. 
\end{equation}
%
%
%
We then have the following theorem.

\begin{thm} \label{th:noise}
	Suppose $V_{\theta}$ is smooth (in $\theta$) in the space of $\cC^{8}$ functions, and exhibits pitchfork bifurcation at the origin. Suppose also $B \neq 0$. Let $u_{\epsilon^{\alpha}}$ be the solution to the PDE \eqref{eq:noise_macro} with initial data $\phi_{0}^{(\epsilon^{\alpha})}$, and let $\theta = \theta(\epsilon)$ be of the form
	\begin{align*}
	\theta = \rho_{1} \epsilon^{\beta_1} + \rho_{2} \epsilon^{\beta_2} + \rho_{3} \epsilon^{\beta_3} + \rho_{4} \epsilon^{\beta_4}
	\end{align*}
	with $0 < \beta_1 < \beta_2 < \beta_3 < \beta_4$ and $\rho_j > 0$. Let $\rho_{j}^{*}$'s be as in \eqref{eq:rho's}. Then, we have the following (with all the limiting processes starting with initial data $\phi_{0}$): 
	
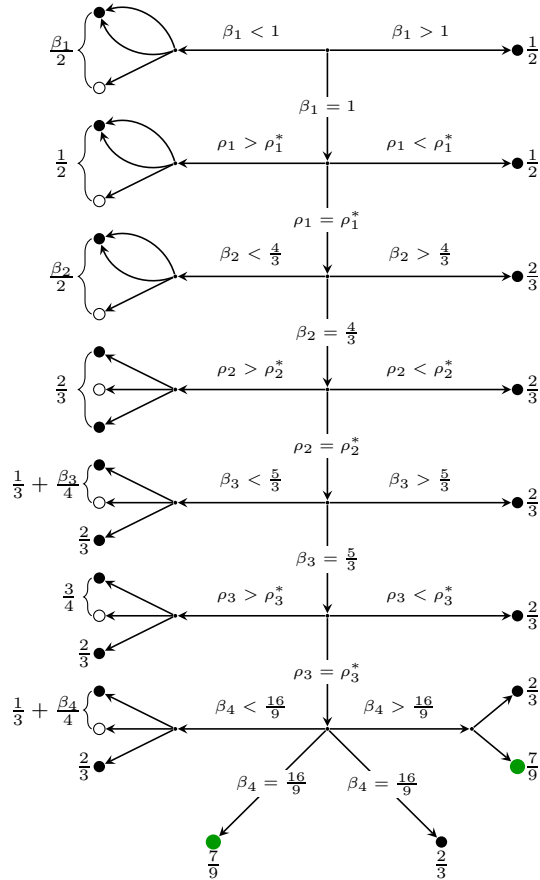
\begin{align*}
	\begin{tikzpicture}[scale=1.0,baseline=0cm]
	\node at (-1.5,-3) [root] (saddle) {}; 
	\node at (-1.5,-3.3) [] () {\scriptsize $\frac{7}{9}$}; 
	\node at (1.5,-3) [sou] (ou) {}; 
	\node at (1.5,-3.3) [] () {\scriptsize $\frac{2}{3}$}; 
	\node at (0,0) [sdot] (r3) {}; 
	\node at (-2,0) [sdot] (r3left) {}; 
	\node at (-3,0) [nou] (r3n) {}; 
	\node at (-3,0.5) [sou] (r3s1) {}; 
	\node at (-3,-0.5) [sou] (r3s2) {}; 
	\node at (-3.2,-0.5) [] () {\scriptsize $\frac{2}{3}$}; 
	\node at (2.5,0) [sou] (r3right) {}; 
	\node at (2.7,0) [] () {\scriptsize $\frac{2}{3}$}; 
	\node at (0,1.5) [sdot] (b3) {}; 
	\node at (-2,1.5) [sdot] (b3left) {}; 
	\node at (-3,1.5) [nou] (b3n) {}; 
	\node at (-3,2) [sou] (b3s1) {}; 
	\node at (-3,1) [sou] (b3s2) {}; 
	\node at (-3.2,1) [] () {\scriptsize $\frac{2}{3}$}; 
	\node at (2.5,1.5) [sou] (b3right) {}; 
	\node at (2.7,1.5) [] () {\scriptsize $\frac{2}{3}$}; 
	\node at (0,3) [sdot] (r2) {}; 
	\node at (-2,3) [sdot] (r2left) {}; 
	\node at (-3,3) [nou] (r2n) {}; 
	\node at (-3,3.5) [sou] (r2s1) {}; 
	\node at (-3,2.5) [sou] (r2s2) {}; 
	\node at (2.5,3) [sou] (r2right) {}; 
	\node at (2.7,3) [] () {\scriptsize $\frac{2}{3}$}; 
	\node at (0,4.5) [sdot] (b2) {}; 
	\node at (-2,4.5) [sdot] (b2left) {}; 
	\node at (-3,4) [nou] (b2n) {}; 
	\node at (-3,5) [sou] (b2s) {}; 
	\node at (2.5,4.5) [sou] (b2right) {}; 
	\node at (2.7,4.5) [] () {\scriptsize $\frac{2}{3}$}; 
	\node at (0,6) [sdot] (r1) {}; 
	\node at (-2,6) [sdot] (r1left) {}; 
	\node at (-3,5.5) [nou] (r1n) {}; 
	\node at (-3,6.5) [sou] (r1s) {}; 
	\node at (2.5,6) [sou] (r1right) {}; 
	\node at (2.7,6) [] () {\scriptsize $\frac{1}{2}$}; 
	\node at (0,7.5) [sdot] (b1) {}; 
	\node at (-2,7.5) [sdot] (b1left) {}; 
	\node at (-3,7) [nou] (b1n) {}; 
	\node at (-3,8) [sou] (b1s) {}; 
	\node at (2.5,7.5) [sou] (b1right) {}; 
	\node at (2.7,7.5) [] () {\scriptsize $\frac{1}{2}$}; 
	\node at (0,-1.5) [sdot] (b4) {}; 
	\node at (-2,-1.5) [sdot] (b4left) {}; 
	\node at (1.9,-1.5) [sdot] (b4right) {}; 
	\node at (-3,-1) [sou] (b4s1) {}; 
	\node at (-3,-2) [sou] (b4s2) {}; 
	\node at (-3,-1.5) [nou] (b4n) {}; 
	\node at (2.5,-1) [sou] (b4s) {}; 
	\node at (2.5,-2) [root] (b4wick) {}; 
	\node at (2.7,-1) [] () {\scriptsize $\frac{2}{3}$}; 
	\node at (2.7,-2) [] () {\scriptsize $\frac{7}{9}$}; 
	\node at (-3.2,-2) [] () {\scriptsize $\frac{2}{3}$}; 
	\draw[akernel] (r3) to node[labl]{\tiny $\rho_3 = \rho_3^*$} (b4); 
	\draw[akernel] (b4) to node[anchor = south]{\tiny $\beta_4 < \frac{16}{9}$} (b4left); 
	\draw[akernel] (b4) to node[anchor = south]{\tiny $\beta_4 > \frac{16}{9}$} (b4right); 
	\draw[akernel] (b4left) to (b4s1); 
	\draw[akernel] (b4left) to (b4s2); 
	\draw[akernel] (b4left) to (b4n); 
	\draw[akernel] (b4right) to (b4s); 
	\draw[akernel] (b4right) to (b4wick); 
	\draw [decorate,decoration={brace,amplitude=3pt}] (-3.1,-1.5) to node[midway, xshift=-0.6cm] {\scriptsize $\frac{1}{3} + \frac{\beta_4}{4}$} (-3.1,-1); 
	\draw[akernel] (b1) to node[labl]{\tiny$\beta_1=1$} (r1); 
	\draw[akernel] (r1) to node[labl]{\tiny$\rho_1=\rho_1^*$} (b2); 
	\draw[akernel] (b2) to node[labl]{\tiny$\beta_2=\frac{4}{3}$} (r2); 
	\draw[akernel] (r2) to node[labl]{\tiny$\rho_2=\rho_2^*$} (b3); 
	\draw[akernel] (b3) to node[labl]{\tiny$\beta_3=\frac{5}{3}$} (r3); 
	\draw[akernel] (b4) to node[labl]{\tiny$\beta_4 = \frac{16}{9}$} (saddle); 
	\draw[akernel] (b4) to node[labl]{\tiny$\beta_4 = \frac{16}{9}$} (ou); 
	\draw[akernel] (r3) to node[anchor = south]{\tiny $\rho_3>\rho_3^*$} (r3left); 
	\draw[akernel] (r3left) to (r3n); 
	\draw[akernel] (r3) to node[anchor = south]{\tiny $\rho_3<\rho_3^*$} (r3right); 
	\draw[akernel] (r3left) to (r3s1); 
	\draw[akernel] (r3left) to (r3s2); 
	\draw[akernel] (b3) to node[anchor = south]{\tiny $\beta_3<\frac{5}{3}$} (b3left); 
	\draw[akernel] (b3left) to (b3n); 
	\draw[akernel] (b3) to node[anchor = south]{\tiny $\beta_3>\frac{5}{3}$} (b3right); 
	\draw[akernel] (b3left) to (b3s1); 
	\draw[akernel] (b3left) to (b3s2); 
	\draw[akernel] (r2) to node[anchor = south]{\tiny $\rho_2>\rho_2^*$} (r2left); 
	\draw[akernel] (r2left) to (r2n); 
	\draw[akernel] (r2) to node[anchor = south]{\tiny $\rho_2<\rho_2^*$} (r2right); 
	\draw[akernel] (r2left) to (r2s1); 
	\draw[akernel] (r2left) to (r2s2); 
	\draw[akernel] (b2) to node[anchor = south]{\tiny $\beta_2<\frac{4}{3}$} (b2left); 
	\draw[akernel] (b2) to node[anchor = south]{\tiny $\beta_2>\frac{4}{3}$} (b2right); 
	\draw[akernel, bend left = 45] (b2left) to (b2s); 
	\draw[akernel, bend right = 45] (b2left) to (b2s); 
	\draw[akernel] (b2left) to (b2n); 
	\draw[akernel] (r1) to node[anchor = south]{\tiny $\rho_1>\rho_1^*$} (r1left); 
	\draw[akernel] (r1) to  node[anchor = south]{\tiny $\rho_1<\rho_1^*$} (r1right); 
	\draw[akernel, bend left=45] (r1left) to (r1s); 
	\draw[akernel, bend right=45] (r1left) to (r1s); 
	\draw[akernel] (r1left) to (r1n); 
	\draw[akernel] (b1) to node[anchor = south]{\tiny $\beta_1<1$} (b1left); 
	\draw[akernel] (b1) to node[anchor = south]{\tiny $\beta_1>1$} (b1right); 
	\draw[akernel, bend left = 45] (b1left) to (b1s); 
	\draw[akernel, bend right = 45] (b1left) to (b1s); 
	\draw[akernel] (b1left) to (b1n); 
	\draw [decorate,decoration={brace,amplitude=5pt}] (-3.1,7) to node[midway, xshift=-0.4cm] {\scriptsize $\frac{\beta_1}{2}$} (-3.1,8); 
	\draw [decorate,decoration={brace,amplitude=5pt}] (-3.1,5.5) to node[midway, xshift=-0.4cm] {\scriptsize $\frac{1}{2}$} (-3.1,6.5); 
	\draw [decorate,decoration={brace,amplitude=5pt}] (-3.1,4) to node[midway, xshift=-0.4cm] {\scriptsize $\frac{\beta_2}{2}$} (-3.1,5); 
	\draw [decorate,decoration={brace,amplitude=5pt}] (-3.1,2.5) to node[midway, xshift=-0.4cm] {\scriptsize $\frac{2}{3}$} (-3.1,3.5); 
	\draw [decorate,decoration={brace,amplitude=3pt}] (-3.1,1.5) to node[midway, xshift=-0.6cm] {\scriptsize $\frac{1}{3}+\frac{\beta_3}{4}$} (-3.1,2); 
	\draw [decorate,decoration={brace,amplitude=3pt}] (-3.1,0) to node[midway, xshift=-0.3cm] {\scriptsize $\frac{3}{4}$} (-3.1,0.5); 
	\end{tikzpicture}
	\end{align*}
Here, the notations are the same as in Remark~\ref{rm:chart}: each \tikz[baseline=-0.1cm] \node [sou] {}; represents a stable OU process, each \tikz[baseline=-0.1cm] \node [nou] {}; represents an unstable OU process, and each green node \tikz[baseline=-0.1cm] \node [root] {}; represents a $\Phi^3_3$. Each black node \tikz[baseline=-0.1cm] \node [sou] {}; with two arrows pointing to it indicates that the two limiting OU processes, obtained by shifting the field to the left and to the right, have the same coefficients. The numbers next to each dot indicates the scale $\alpha$ at which one observes the corresponding limit. 

All the convergences above are convergences in law in $\xX^\eps$, with $T_\star$ (resp. $T_\star^{(\eps)}$)
given by the explosion times of the respective processes in $\cC^\eta$ (in the same sense as in Theorem \ref{th:main_asymmetric}). 
\end{thm}
\begin{proof}
	The key in the proof is to show the convergence of $\lambda_{j}^{(\delta)}$'s as defined in \eqref{eq:noise_coefficients} to the desired limiting values at various choices of $\alpha$ and $h_{\epsilon}$. In particular, for the $\Phi^3_3$ limit, the coefficient of the quadratic Wick term is proportional to $B^{\frac{1}{3}}$. The details of the proof are very similar to those in Theorem~\ref{th:main_asymmetric}, and is straightforward by the expression of the $a_{j}^{(h)}$'s in \eqref{eq:noise_a's}, so we do not repeat them here. 
\end{proof}

\begin{rmk}
	If any of the $\rho_{j}$'s is negative, it will make $\theta$ further away from the effective critical value $\theta^{*}$ (but close to $0$), and one could only see one stable OU in the limit. In fact, by including negative $\rho_{j}$'s, one will fill in the jump of the scale (from $\frac{1}{2}$ to $\frac{2}{3}$) on the right of the figure, and obtain a complete description (in terms of the continuous change of the scale) as the left side of the figure.  
\end{rmk}

\bibliographystyle{Martin}
\bibliography{Refs}

\end{document}